\documentclass[11pt,letterpaper]{article}
\usepackage{graphicx} 
\usepackage{xcolor}
\usepackage[margin=1in]{geometry}
\usepackage{multirow}
\usepackage{authblk}
\usepackage{amsmath,amsthm,amssymb}
\usepackage{verbatim}
\usepackage{enumerate}
\usepackage{enumitem}
\usepackage{thmtools}
\usepackage{thm-restate}
\usepackage{todonotes}
\usepackage[bottom]{footmisc}
\usepackage[normalem]{ulem}
\usepackage[noend,linesnumbered]{algorithm2e}

\usepackage[colorlinks,citecolor=green!50!black]{hyperref}
\usepackage[capitalise]{cleveref}

\crefname{subappendix}{Appendix}{Appendices}
\crefname{appendix}{Appendix}{Appendices}
\crefname{section}{Section}{Sections}
\crefname{figure}{Figure}{Figures}

\newtheorem{theorem}{Theorem}[section]
\newtheorem{fact}[theorem]{Fact}
\newtheorem{lemma}[theorem]{Lemma}

\theoremstyle{definition}
\newtheorem{definition}[theorem]{Definition}

\newtheorem{claim}[theorem]{Claim}
\theoremstyle{remark}

\newtheorem{remark}[theorem]{Remark}

\crefname{section}{Section}{Sections}
\crefname{construction}{Construction}{Constructions}
\crefname{claim}{Claim}{Claims}
\crefname{fact}{Fact}{Facts}
\crefname{corollary}{Corollary}{Corollaries}

\newcommand{\T}{\mathcal{T}}
\renewcommand{\S}{\mathcal{S}}
\newcommand{\LCP}{\mathsf{LCP}}

\newcommand{\TreeLCP}{\mathsf{TreeLCP}}
\newcommand{\M}{\mathcal{M}}

\newcommand{\TS}{\widetilde{\mathcal{T}}}

\newcommand{\TT}{\mathcal{T}}

\newcommand{\eps}{\varepsilon}

\newcommand{\per}{\mathsf{per}}

\newcommand{\Oh}{\mathcal{O}}
\newcommand{\occ}{\mathsf{occ}}
\newcommand{\ceil}[1]{\left\lceil #1 \right\rceil}
\newcommand{\floor}[1]{\left\lfloor #1 \right\rfloor}
\newcommand{\dd}{\mathinner{\ldotp\ldotp}}

\newcommand{\LeftMis}{\mathsf{LeftMisper}}
\newcommand{\RightMis}{\mathsf{RightMisper}}
\newcommand{\Mis}{\mathsf{Misper}}

\title{Space-Efficient $k$-Mismatch Text Indexes}
\author[1]{Tomasz Kociumaka}
\author[2]{Jakub Radoszewski}
\affil[1]{Max Planck Institute for Informatics, Saarland Informatics Campus, Saarbrücken, Germany  (\texttt{tomasz.kociumaka@mpi-inf.mpg.de})}
\affil[2]{Institute of Informatics, University of Warsaw, Warsaw, Poland
  (\texttt{jrad@mimuw.edu.pl})}
\date{}

\begin{document}
\maketitle
\begin{abstract}
A central task in string processing is \emph{text indexing}, where the goal is to preprocess a \emph{text} (a~string of length $n$) into an efficient \emph{index} (a data structure) supporting queries about the text.
While the most fundamental \emph{exact pattern matching} queries ask to find all the occurrences of a \emph{pattern} (a string of length $m$) as substrings of the text, many applications call for \emph{approximate pattern matching} queries, where the pattern may differ slightly from the matching substrings.
A breakthrough in the extensive study of approximate text indexing came from Cole, Gottlieb, and Lewenstein (STOC 2004), who proposed \emph{$k$-errata trees} --- a family of text indexes supporting several closely related flavors of approximate pattern matching queries. 
In particular, $k$-errata trees yield an elegant solution to \emph{$k$-mismatch} queries, where the similarity is quantified using an upper bound $k\ge1$ on the Hamming distance between the pattern and its approximate occurrences.
The resulting $k$-mismatch index uses $\mathcal{O}(n\log^k n)$ space and answers a query for a length-$m$ pattern in $\mathcal{O}(\log^k n \log \log n + m + \mathsf{occ})$ time, where $\mathsf{occ}$ is the number of approximate occurrences.

In retrospect, $k$-errata trees appear very well optimized: even though a large body of work has adapted $k$-errata trees to various settings throughout the past two decades, the original time-space trade-off for $k$-mismatch indexing has not been improved in the general case.
We present the first such improvement, a $k$-mismatch index with $\mathcal{O}(n\log^{k-1} n)$ space and the same query time as $k$-errata trees.

Previously, due to a result of Chan, Lam, Sung, Tam, and Wong (Algorithmica 2010), such an $\mathcal{O}(n\log^{k-1} n)$-size index has been known only for texts over alphabets of constant size $\sigma=\mathcal{O}(1)$.
In this setting, however, we obtain an even smaller $k$-mismatch index of size only $\mathcal{O}(n \log^{k-2+\varepsilon+\frac{2}{k+2-(k \bmod 2)}} n)\subseteq \mathcal{O}(n\log^{k-1.5+\varepsilon} n)$ for $2\le k\le\mathcal{O}(1)$ and any constant $\varepsilon>0$.
Along the way, we also develop improved indexes for short patterns, offering better trade-offs in this practically relevant special case.
\end{abstract}

\section{Introduction}\label{sec:intro}
In (full-)text indexing, the goal is to convert a string $T$ of length $n$ (referred to as the \emph{text}) into a data structure (called an \emph{index}) that can efficiently answer subsequent queries about the text~$T$.
The most fundamental are \emph{exact pattern matching} queries, which ask to report all occurrences of a given string $P$ of length $m$ (called a \emph{pattern}) as substrings of $T$.
The study of exact text indexes dates back to the early 1970s, when Weiner~\cite{DBLP:conf/focs/Weiner73} introduced \emph{suffix trees}.
This textbook data structure occupies $\Oh(n)$ space, supports $\Oh(n)$-time construction, and answers pattern matching queries in $\Oh(m+\occ)$ time, where $\occ$ is the number of occurrences of the pattern $P$ in the text~$T$.\footnote{
The original implementation applies to texts over alphabets of constant size $\sigma=\Oh(1)$, but modern ones~\cite{DBLP:conf/focs/Farach97,DBLP:conf/cpm/0001G15} support integer alphabets $\{0,\ldots,\sigma-1\}$ of any size $\sigma=n^{\Oh(1)}$.
This comes at the cost of increasing the construction time to $\Oh(n \log \log n)$, increasing the query time to $\Oh(\log \log \sigma + m + \occ)$, \textbf{or} making the construction randomized.}
Further work in exact text indexing focused mainly on improving the space complexity (the main bottleneck in practice), often at the cost of degraded query time.
Selected milestones in this direction include the suffix array~\cite{DBLP:journals/siamcomp/ManberM93,DBLP:journals/jacm/KarkkainenSB06}, the compressed suffix array~\cite{DBLP:journals/siamcomp/GrossiV05} and the FM-index~\cite{FM05}, as well as the $r$-index~\cite{DBLP:journals/jacm/GagieNP20} and other indexes for highly repetitive texts~\cite{N22}.

Another extensively studied direction, motivated by numerous applications~\cite{DBLP:journals/csur/Navarro01}, is approximate text indexing.
An established way to define approximate occurrences of the pattern is through a threshold $k\ge1$ on the Hamming distance:
if $T[i \dd i+m)$ and $P$ differ in at most $k$ positions, that is, the two strings have at most $k$ mismatches, then $T[i \dd i+m)$ is a \emph{$k$-mismatch occurrence} of $P$ in~$T$.
Given a pattern $P$, a \emph{$k$-mismatch index} over $T$ is required to report all the $k$-mismatch occurrences of $P$ in $T$. 
Unlike offline $k$-mismatch pattern matching algorithms~\cite{DBLP:journals/jal/AmirLP04,DBLP:conf/soda/CliffordFPSS16,DBLP:conf/icalp/GawrychowskiU18}, which operate on a fixed pattern $P$, an index supports many online queries.  

The earliest approximate indexes were designed for the special case of $k=1$~\cite{DBLP:journals/jal/AmirKLLLR00,DBLP:conf/esa/BuchsbaumGW00}.
These solutions were later optimized~\cite{DBLP:conf/cpm/Belazzougui09,DBLP:journals/algorithmica/Belazzougui15}, but the underlying techniques do not generalize to $k>1$.
The seminal contribution to approximate indexing for arbitrary $k\ge1$ is the \emph{$k$-errata tree} of Cole, Gottlieb, and Lewenstein~\cite{DBLP:conf/stoc/ColeGL04}, which takes $\Oh(n \log^k n)$ space and lists the $k$-mismatch occurrences of a length-$m$ pattern in $\Oh(\log^k n \log \log n + m + \occ)$ time.\footnote{The $k$-mismatch index of \cite{DBLP:conf/stoc/ColeGL04} actually occupies $\Oh((c^{k}/k!)\cdot n\log^k n)$ space for a constant $c>1$.
Since we focus on the regime of $k=\Oh(1)$, in the introduction we omit the $c^{k}/k!$ factor for simplicity; note that this factor is at most 1 for sufficiently large $k$. Throughout, $\log$ denotes the base-2 logarithm.}
The errata trees have been featured in the \emph{Encyclopedia of Algorithms}~\cite{L16} and adapted to numerous tasks \cite{BGPSSZ24,CIKPRS22,DBLP:conf/esa/Charalampopoulos21,GGMPP25,GLS18,GHPT25,TAA16,ZLPT24}.
Even though several alternative $k$-mismatch indexes have been presented, none of the known trade-offs improve upon the $k$-errata trees in the general case.
Our $\Oh(n \log^{k-1} n)$-size index provides the first such improvement for any $k$: it retains the query time of $k$-errata trees and supports texts over large alphabets.

\begin{restatable}{theorem}{thmone}\label{thm:1}
For every text of length $n$ and integer threshold $k\ge 1$, there exists a $k$-mismatch index of size $\Oh(n \log^{k-1} n)$ that answers queries in $\Oh(\log^k n \log\log n+m+\occ)$ time, where $m$ is the length of the pattern and $\occ$ is the number of $k$-mismatch occurrences.
\end{restatable}

Fifteen years ago, Chan, Lam, Sung, Tam, and Wong~\cite{DBLP:journals/algorithmica/ChanLSTW10} presented an analogous trade-off for texts over constant-size alphabets;
for a general alphabet of size $\sigma$, the query complexity of their index becomes $\Oh(\sigma\cdot \log^k n \log \log n+m+\occ)$.
In this work, we also address the case of $\sigma=\Oh(1)$: for every constant $k\ge 2$, we present an even smaller $k$-mismatch index.

\begin{restatable}{theorem}{thmtwo}\label{thm:2}
For every text of length $n$ over an alphabet of size $\sigma=\Oh(1)$, integer threshold $k=\Oh(1)$, and real constant $\eps > 0$, there exists a $k$-mismatch index of size $\Oh(n \log^{k-2+\eps+\frac{2}{k+2}} n)$ for even $k$ or $\Oh(n \log^{k-2+\eps+\frac{2}{k+1}} n)$ for odd $k$ that answers queries in $\Oh(\log^k n \log\log n+m+\occ)$ time, where $m$ is the length of the pattern and $\occ$ is the number of $k$-mismatch occurrences.
\end{restatable}

In particular, our $2$-mismatch index is of size $\Oh(n\log^{0.5+\eps}n)$, that is, $\log^{0.5-\eps} n$ times smaller than the index of~\cite{DBLP:journals/algorithmica/ChanLSTW10} and $\log^{1.5-\eps} n$ times smaller than the $2$-errata tree. 
The savings become more prominent, growing arbitrarily close to $\log n$ and $\log^2 n$, respectively, as $k$ increases.
Not only does \cref{thm:1} provide the first general-case improvement in the two-decade-long history of errata trees, but the restricted case of $\sigma=\Oh(1)$ has also seen no progress since 2010.

\Cref{tab:main} summarizes our and previously known upper bounds for $k$-mismatch indexing. We assume the word RAM model with word size $\Omega(\log n)$.
Besides the contributions discussed above, earlier work also offers two results that extend the original $k$-errata tree trade-off in two directions.
The $k$-mismatch index by Tsur~\cite{DBLP:journals/jda/Tsur10} allows decreasing the query time at the cost of a larger space. In particular, for constant $k$ and $\sigma$, a query time of $\Oh(\log^k n / \log^k \log n + m + \occ)$ can be achieved using $\Oh(n \log^{k+\eps} n)$ space, for any constant $\eps > 0$.
Alternatively, $\Oh(\log \log n + m + \occ)$ query time can be obtained using $\Oh(n^{1+\eps})$ space. A different trade-off, proposed by Chan, Lam, Sung, Tam, and Wong~\cite{DBLP:journals/jda/ChanLSTW11}, allows decreasing the space at the cost of slower queries. Among others, they obtain an $\Oh(n)$-size index with query time $\Oh(\log^{k(k+1)}  n\cdot  \log \log n + m + \occ)$ and an $\Oh(n \log^{k-1} n)$-size index with $\Oh(\log^{2k} n \log \log n + m + \occ)$-time queries. 

\begin{table}[t!]
\renewcommand{\arraystretch}{1.3}
\centering
\begin{tabular}{c|c|c|p{2.75cm}}
\textbf{Space} & \textbf{Query time} (plus $\Oh(m+\occ)$) & \textbf{Reference} & \textbf{Comment} \\\hline
\textcolor{gray}{$\Oh(n \log^k n)$} & \multirow{4}{*}{$\Oh(\log^k n \log \log n)$} & \cite{DBLP:conf/stoc/ColeGL04} & improved here\\\cline{1-1}\cline{3-4}
\textcolor{gray}{$\Oh(n \log^{k-1} n)$} & & \cite{DBLP:journals/algorithmica/ChanLSTW10} & \textcolor{gray}{$\sigma=\Oh(1)$,}\newline improved here\\\cline{1-1}\cline{3-4}
$\Oh(n \log^{k-1} n)$ & & \textbf{Thm.~\ref{thm:1}} & \\\cline{1-1}\cline{3-4}
$\Oh(n \log^{k-2+\eps+\frac{2}{k+2-(k \bmod 2)}} n)\vphantom{n^{n^{n^{n^n}}}}$ & & \textbf{Thm.~\ref{thm:2}} & $\sigma=\Oh(1),$\newline constant $\eps>0$\\\hline
$\Oh(n (\alpha \log \alpha \log n)^k)$ & $\Oh((\log_{\alpha} n)^k\log \log n)$ & \cite{DBLP:journals/jda/Tsur10} & $2 \le \alpha \le n/2$\\\hline
$\Oh(n \log^{k-h+1} n)$ & $\Oh(\log^{\max(kh,k+h)} n \log \log n)$ & \cite{DBLP:journals/jda/ChanLSTW11} & $1 \le h \le k+1$
\end{tabular}
\caption{The $k$-mismatch indexes for arbitrary pattern lengths $m$, for $k=\Oh(1)$.}\label{tab:main}
\end{table}

\begin{table}[b!]
\renewcommand{\arraystretch}{1.3}
\centering
\begin{tabular}{c|c|c|p{3.9cm}}
\textbf{Space} & \textbf{Query time} (plus $\Oh(m+\occ)$) & \textbf{Reference} & \textbf{Comment} \\\hline
\multirow{5}{*}{$\Oh(n)$} & \textcolor{gray}{$\Oh(m^{k+2}\sigma^k \log(m+\sigma))$} & \cite{DBLP:conf/cpm/Ukkonen93} & improved by \cite{DBLP:conf/cpm/Cobbs95} \\\cline{2-4}
& $\Oh(m^{k+2}\sigma^k)$ & \cite{DBLP:conf/cpm/Cobbs95} &\\\cline{2-4}
& \textcolor{gray}{$\Oh((\sigma m)^k\log n)$} & \cite{DBLP:conf/cpm/HuynhHLS04} & improved by \cite{DBLP:journals/algorithmica/ChanLSTW10,DBLP:journals/algorithmica/LamSW08}\\\cline{2-4}
& $\Oh((\sigma m)^k\log \log n)$ & \cite{DBLP:journals/algorithmica/LamSW08} & \\\cline{2-4}
& $\Oh((\sigma m)^{k-1}\log n \log \log n)$ & \cite{DBLP:journals/algorithmica/ChanLSTW10} &\\\hline
$\Oh(n\mu^h\log^2 \mu/\log n)$ & $\Oh(m^{k-h}\log m\log \log n)$ & \textbf{Thm.~\ref{thm:sk1}} & $m \le \mu = \Omega(\log n)$,\newline $h \in [1 \dd k)$, $\sigma=\Oh(1)$\\\hline
$\Oh(n \log_\sigma^k n)$ & $\Oh(m^k)$ & \cite{DBLP:journals/jcb/Al-Okaily15} & \\
\end{tabular}
\caption{The $k$-mismatch indexes for specific pattern lengths $m$, for $k=\Oh(1)$.}\label{tab:main_small}
\end{table}

One of the technical contributions behind \cref{thm:2} is a new $k$-mismatch index optimized for handling short patterns, of lengths polylogarithmic in $n$.
Our and previous solutions relevant for this setting are presented in \cref{tab:main_small}.
They are all characterized by the product of the index size and query time being $\Omega(nm^{k-1})$.
In the query complexity of the indexes from \cite{DBLP:conf/cpm/Cobbs95,DBLP:conf/cpm/Ukkonen93}, one can trade the $m^{k+1}\sigma^k$ factor by an $n$ factor (which is not compelling for $k=\Oh(1)$, since offline $k$-error pattern matching works in $\Oh(nk)=\Oh(n)$ time for a string over an integer alphabet~\cite{DBLP:journals/jal/LandauV89}).

Beyond the worst-case $k$-mismatch indexes, data structures with good average-case complexity for general $k$ \cite{DBLP:journals/jcb/Al-Okaily15,DBLP:journals/pvldb/ChenN21,DBLP:conf/spire/CoelhoO06,DBLP:journals/tcs/EpifanioGMRS07,DBLP:conf/ciac/GabrieleMRS03,DBLP:journals/algorithmica/Maass06,DBLP:journals/jda/MaassN07} as well as for $k=1$ \cite{DBLP:journals/ipl/MaassN05} were proposed.
Solutions that are efficient for particular data have also been studied~\cite{Chen2018,DBLP:journals/tcs/KucherovST16,DBLP:conf/bibm/LamLTWWY09,DBLP:journals/bioinformatics/LiD09}. Approximate indexing in external memory was considered in~\cite{DBLP:journals/tcs/HonLSTV11}.
Indexes for a small $k$ (e.g., $k \le 7$) were evaluated in practice~\cite{DBLP:journals/pvldb/ChenN21,DBLP:conf/bibm/LamLTWWY09}.

On the lower bound side, Cohen-Addad, Feuilloley, and Starikovskaya~\cite{DBLP:conf/soda/Cohen-AddadFS19} showed that if a $k$-mismatch index answers queries in $\Oh(\log^k n / (2k)^k + m + \occ)$ time in the pointer machine model, then it requires $\Oh(C^k n)$ space for a constant $C>0$. This result holds for every even $k$ that satisfies $8\sqrt{\log n}/\sqrt{3} \le k = o(\log n)$. Furthermore, they showed that, under the Strong Exponential Time Hypothesis, no $k$-mismatch index for $k=\Theta(\log n)$ can be constructed in polynomial time and decide in $\Oh(n^{1-\delta})$ time (for some constant $\delta>0$) whether a given pattern of length $m=\Theta(\log n)$ has any $k$-mismatch occurrence. These results are a step forward in our understanding of the complexity of approximate indexing, but they are still quite far from the complexity of $k$-errata trees. In particular, no non-trivial lower bounds can be deduced for the basic case of $k=\Oh(1)$.

\subsection{Related indexing problems}
A related problem is indexing a text to support queries for patterns with wildcards (i.e., don't care symbols). A significant amount of work has been devoted to this problem~\cite{DBLP:journals/mst/BilleGVV14,DBLP:conf/stoc/ColeGL04,DBLP:conf/isaac/LamSTY07,DBLP:journals/tcs/LewensteinMRT14,DBLP:conf/stacs/LewensteinNV14,DBLP:conf/sofsem/RahmanI07}.
Lower bounds of similar flavor as in \cite{DBLP:conf/soda/Cohen-AddadFS19} for indexing for patterns with wildcards were presented earlier by Afshani and Nielsen~\cite{DBLP:conf/icalp/AfshaniN16}.
The algorithms presented in~\cite{DBLP:journals/mst/BilleGVV14,DBLP:conf/stoc/ColeGL04,DBLP:journals/tcs/LewensteinMRT14} all employ variants of errata trees.
The problem of indexing for patterns with wildcards appears to be easier, in terms of complexity, than approximate text indexing.
In particular, the index of Cole, Gottlieb, and Lewenstein~\cite{DBLP:conf/stoc/ColeGL04} occupies $\Oh(n \log^{k} n)$ space and answers queries in $\Oh(2^{k}\log\log n + m + \occ)$ time, that is, with a $2^{k}$ rather than a $\log^{k} n$ factor. 
For constant $k$ and $\sigma$, a trivial linear-space index, already noted in~\cite{DBLP:conf/stoc/ColeGL04}, tests all possible substitutions of wildcards in the pattern.
This yields $\Oh(n)$ space (the suffix tree of $T$) and $\Oh(m\sigma^{k}+\occ)=\Oh(m+\occ)$ query time.
As a warm-up for the proof of \cref{thm:1}, we present the first general improvement over the complexities of~\cite{DBLP:conf/stoc/ColeGL04} in the trade-off for $k$-errata trees when the pattern contains $k$ wildcards, as stated in the next theorem.

\begin{restatable}{theorem}{wildP}\label{thm:kwildP}
For a string $T$ of length $n$, there exists an index using $\Oh(n\log^{k-1} n)$ space that answers pattern matching queries for patterns of length $m$ with up to $k$ wildcards in $\Oh(2^k \log\log n + m + \occ)$ time.
\end{restatable}

A variant of the problem in which the text contains wildcards was also considered~\cite{DBLP:conf/stoc/ColeGL04,DBLP:journals/jda/HonKSTV13,DBLP:conf/isaac/LamSTY07,DBLP:conf/spire/TamWLY09,DBLP:journals/tcs/Thachuk13}. Here, for $\sigma,k=\Oh(1)$, one can fill in all the wildcards in the text in every possible way and then construct a suffix tree for every resulting string text. This leads to $\Oh(n\sigma^k)=\Oh(n)$ space and $\Oh(m+\sigma^k \occ)=\Oh(m+\occ)$-time queries.

Following~\cite{DBLP:conf/stoc/ColeGL04}, one can also consider \emph{$k$-edit indexes}—indexes that report substrings of the text whose edit (Levenshtein) distance to the pattern is at most~$k$.
In this work, we focus on the Hamming distance, which is technically simpler: for example, a $k$-mismatch occurrence of the pattern in the text always has the same length as the pattern, unlike in the case of $k$-edit occurrences.
In particular, all known $k$-edit indexes can be adapted (and in fact simplified) to operate as $k$-mismatch indexes.

\subsection{Technical overview}\label{sec:techov}
\subsubsection*{Indexing for short patterns with $\sigma,k=\Oh(1)$} As a stepping stone to \cref{thm:2}, we present a $k$-mismatch index that supports patterns of length $m \le \mu$, where $\mu = \Omega(\log n)$ is a parameter specified at construction time. Existing linear-space $k$-mismatch indexes (see \cref{tab:main_small}) follow a shared blueprint: exhaustively apply all up to $k$ modifications to the pattern $P$ in every possible way and then perform exact pattern matching for each variant. This approach yields varying query times depending on how efficiently the exact search is implemented.

Our index, in contrast, balances the cost of these modifications between index size and query time.
Suppose we are searching for occurrences of a pattern $P[0\dd m)$ with exactly $k$ mismatches.
For a trade-off parameter $h \in [1\dd k)$, we consider each possible \emph{pivot} $j \in [0\dd m)$, aiming to find occurrences in which $P[0\dd j)$ has $h$ mismatches and $P[j\dd m)$ has $k-h$ mismatches.
To support such queries, we preprocess the text~$T$ by modifying every suffix $T[i\dd n)$ through all possible $h$ substitutions within its prefix $T[i\dd i+j)$, and then we build a compact trie containing all resulting modified suffixes.
At query time, we modify~$P$ in all possible ways by applying $k-h$ substitutions within $P[j\dd m)$, and we search for the resulting patterns in the trie.
Because this scheme must consider multiple pivots~$j$, care is needed to avoid reporting duplicates, particularly when also handling $(<k)$-mismatch occurrences of~$P$, while ensuring that no additional $\Omega(m\cdot\occ)$ query-time overhead is incurred.

We refine the basic scheme through several improvements:
\begin{itemize}
    \item Naively, locating each modified pattern in the trie takes $\Theta(m)$ time. We reduce this to $\Oh(\log\log n)$ using a generalization of rooted LCP queries~\cite{DBLP:conf/stoc/ColeGL04} to tries of modified suffixes (see \cref{lem:modifiedLCP} in \cref{sec:unrooted}).
    \item  Another $\Theta(m)$ factor arises from the need to examine all possible pivots $j \in [0 \dd m)$, as a given modified pattern may be compatible with up to $\Omega(m)$ pivots. We limit the number of such pivots to $\Oh(\log m)$ by favoring those divisible by larger powers of two.
    \item Finally, we save a nearly-$\log n$ factor in space by explicitly storing only the modified suffixes with up to $h-1$ changes. Suffixes with exactly $h$ changes can still be retrieved indirectly via the rank/select data structures of Raman, Raman, and Rao~\cite{DBLP:journals/talg/RamanRS07}.
\end{itemize}
Together, these refinements yield the following frugal data structure described in \cref{sec:short}.

\begin{theorem}\label{thm:sk1}
Let $T$ be a string of length $n$ over an integer alphabet of size $\sigma$, and let $\mu = \Omega(\log n)$ with $k,\sigma = \Oh(1)$. For every $h \in [1 \dd k)$, there exists a $k$-mismatch index using $\Oh(n\mu^h \log^2 \mu/\log n)$ space that answers queries for patterns of length $m \le \mu$ in $\Oh(m^{k-h} \log m \log\log n + \occ)$ time.
\end{theorem}

Compared to existing data structures for short patterns, ours can be viewed as trading space for improved query time.
In particular, for $k,\sigma = \Oh(1)$ and patterns of length up to $\mu = \Oh(2^{\log^{1/3} n}) = \Oh(n^{1/\log^{2/3} n})$, that is, when $\log^3 \mu = \Oh(\log n)$, the product of space and query time is $\Oh(n\mu^k\log^3 \mu \log \log n/\log n)=\Oh(n\mu^k \log \log n)$ matching the corresponding bound of the data structure in~\cite{DBLP:journals/algorithmica/LamSW08}.
For patterns of length $\mu = \log^{\Oh(1)} n$, this product in our data structure is actually smaller by a factor of $(\log\log n)^3 / \log n$.

\subsubsection*{Indexing for long patterns}
At a high level, our strategy for patterns of length at least $m$ is to designate $\Oh(n/m \cdot \operatorname{poly}(k))$ positions in the text $T$ as \emph{anchors} and, at query time, to select $\Oh(\operatorname{poly}(k))$ anchors in $P$ such that, for every $k$-mismatch occurrence of $P$ in $T$, at least one anchor in $P$ coincides with an anchor in $T$.
In general, it is impossible to cover all $k$-mismatch occurrences in this way, e.g., when $P = \mathrm{a}^m$ and $T = \mathrm{a}^n$.
Nonetheless, our anchor selection procedure, which relies on local consistency techniques (string synchronizing sets~\cite{DBLP:conf/stoc/KempaK19}) and the structure of periodic fragments in strings (including Lyndon roots of runs~\cite{DBLP:journals/siamcomp/BannaiIINTT17}), misses only highly structured $k$-mismatch occurrences. These can be listed efficiently using a separate subroutine based on interval stabbing~\cite{DBLP:conf/focs/AlstrupBR00}.

To find typical $k$-mismatch occurrences, we handle each anchor $j$ in $P$ separately. We first identify all anchors $a$ in $T$ for which $T[a\dd n)$ starts with a $k$-mismatch occurrence of $P[j\dd m)$, and all anchors $a$ in $T$ for which $T[0\dd a)$ ends with a $k$-mismatch occurrence of $P[0\dd j)$. These two sets are then intersected, taking into account that $k$ is the total mismatch budget across $P[0\dd j)$ and $P[j\dd m)$.
For these steps, we employ two $k$-errata trees~\cite{DBLP:conf/stoc/ColeGL04} built on a subset of suffixes and reversed prefixes of $T$, together with an orthogonal range reporting data structure~\cite{DBLP:conf/focs/AlstrupBR00,DBLP:conf/esa/BuchsbaumGW00} constructed on top of them.

This high-level blueprint can be seen as a generalization of the approaches in exact indexing~\cite{DBLP:conf/soda/KempaK23} (see \cite[Proposition 6.27]{DBLP:journals/corr/abs-2106-12725}) and~\cite[Theorem 4]{DBLP:journals/vldb/AyadLP25} to the $k$-mismatch setting. Even prior to these works, similar techniques were used for computing the longest common factor with $k$ mismatches~\cite{CCIKPRRW18,DBLP:conf/esa/Charalampopoulos21} and for circular pattern matching with $k$ mismatches~\cite{DBLP:journals/jcss/Charalampopoulos21}.
Exact~\cite{CEKNP21,KNO24} and $k$-mismatch~\cite{GHPT25} indexes for highly repetitive texts also share many aspects of this strategy. The key difference is that they exploit compressibility, rather than a lower bound on $|P|$, to reduce the number of anchors in $T$, and they require additional tools to report so-called secondary occurrences.

We now provide a more detailed overview of our index for long patterns.
Suppose the pattern $P$ has length $m \ge (k+1)\gamma$ for some positive integer~$\gamma$. A prefix of $P$ of length $(k+1)\gamma$ can then be split into $k+1$ substrings $P_1, \ldots, P_{k+1}$ of length $\gamma$ each, so that every $k$-mismatch occurrence of $P$ in $T$ contains an \emph{exact} occurrence of at least one of these substrings, say $P_i$, in $T$.
We compute a synchronizing set in $T$ (cf.~\cite{DBLP:conf/stoc/KempaK19}). In the simpler case where $T$ does not contain substrings that are high-exponent string powers, a $\tau$-synchronizing set for an integer $\tau > 0$ is a subset of $\Oh(n/\tau)$ positions in $T$, selected consistently based on the following $2\tau$ characters. The selection is dense: there is a synchronizing position among every $\tau$ consecutive positions in $T$ that are not within the final $2\tau-1$ positions. A length-$2\tau$ substring of $T$ starting at a synchronizing position is called a \emph{$\tau$-synchronizing fragment}.
Thus, for $\tau = \floor{(\gamma+1)/3}$, the occurrence of $P_i$ in $T$ is guaranteed to contain a $\tau$-synchronizing fragment.

A $\tau$-synchronizing set in $T$ does not necessarily extend to a $\tau$-synchronizing set in $P$. However, since $P_i$ matches a corresponding fragment of $T$ exactly, it must contain a fragment that matches a $\tau$-synchronizing fragment of $T$; we are interested in the leftmost such fragment $P[j\dd j+2\tau)$. We use two $(\le k)$-errata trees: one built on suffixes of $T$ starting at synchronizing positions, and the other on the complementary reversed prefixes. Here, we exploit the fact that the original $k$-errata tree, constructed for $x = \Oh(n/\gamma)$ suffixes of a text, uses only $\Oh(n + x \log^k x)$ space~\cite{DBLP:conf/stoc/ColeGL04}.
For each $P_i$, we query $P[j \dd m)$ in the errata tree of suffixes and $(P[0 \dd j))^R$ in the errata tree of reversed prefixes, considering all possible distributions of the $k$ mismatches between $P[j \dd m)$ and $P[0 \dd j)$. The results of these two queries are merged into $k$-mismatch occurrences of the entire $P$ in $T$ using a 2D orthogonal range query over points generated in a manner similar to a tree cross product~\cite{DBLP:conf/esa/BuchsbaumGW00}.
To keep the space of the range query data structure small, we prove a new but intuitive bound on the number of occurrences of a given terminal label in an errata tree. Overall, the space complexity is $\Oh(n \log^{k+\eps}n \, /\, \gamma)$ for any constant $\eps>0$, where the $\log^\eps n$ factor arises from the 2D orthogonal range reporting structure of Alstrup, Brodal, and Rauhe~\cite{DBLP:conf/focs/AlstrupBR00}. The query time matches that of a standard errata tree.

If $T$ contains highly periodic substrings, $P_i$ may not include a $\tau$-synchronizing fragment. In this case, either the entire pattern $P$ is nearly periodic, or there are $\Theta(k)$ \emph{misperiods} in $P$ with respect to the period of $P_i$~\cite{DBLP:journals/jcss/Charalampopoulos21}, and at least one of these misperiods must align with a corresponding misperiod with respect to a run~\cite{DBLP:conf/focs/KolpakovK99}~of $T$.
In the second case, we can apply the same construction as before, but using misperiods in $T$ and $P$ instead of synchronizing positions. In the nearly periodic case, we design a tailored index based on runs and their Lyndon roots~\cite{DBLP:journals/siamcomp/BannaiIINTT17}. All nearly periodic substrings of $T$ are grouped according to the run that generates them, the number of misperiods, and the position where the Lyndon root of the run starts, modulo the period of the run. Subsequently, all runs are grouped by their Lyndon root. For a nearly periodic pattern, we locate all nearly periodic substrings of $T$ at Hamming distance at most $k$ using an interval stabbing data structure.
In \cref{sec:long}, we formalize this approach in the following theorem.

\begin{restatable}{theorem}{llong}\label{thm:long}
For a text $T$ of length $n$, positive integers $k=\Oh(1)$, $\gamma \in [2 \dd n]$, and constant $\eps>0$, there exists an $\Oh(n+ n \log^{k+\eps} n / \gamma)$-space $k$-mismatch index that answers $k$-mismatch queries for patterns of length $m \ge (k+1)\gamma$ in $\Oh(\log^k n \log \log n + m + \occ)$ time.
\end{restatable}

By combining this approach with the data structure for short patterns, we obtain our $k$-mismatch index for $k,\sigma = \Oh(1)$ as stated in \cref{thm:2}.

\subsubsection*{General compact index}
A $k$-errata tree~\cite{DBLP:conf/stoc/ColeGL04} consists of compact tries at levels $0,\ldots,k$. Roughly speaking, each compact trie stores selected suffixes of the text. The compact tries at a higher level are generated from those at the previous level in two different ways, based on nodes and heavy paths, in both cases via so-called group~trees.

For a general alphabet of size $\sigma$, we save space by explicitly constructing only the $(k-1)$-errata tree in $\Oh(n \log^{k-1} n)$ space; this corresponds to the compact tries at levels $0,\ldots,k-1$ of the $k$-errata tree. Chan, Lam, Sung, Tam, and Wong~\cite{DBLP:journals/algorithmica/ChanLSTW10} devised a rather involved approach, based on~\cite{DBLP:journals/talg/RamanRS07}, to avoid generating heavy path-based group trees at level $k$ while still reporting $k$-mismatch occurrences. Instead of constructing node-based group trees, they used a naive query algorithm, which incurs an extra factor of $\sigma$. We show that their technique for heavy path-based group trees can be carefully adapted to node-based group trees as well, effectively eliminating the $\sigma$-factor in query time.
As a subroutine, we need to report all suffixes of $T$ in a given sorted list that do not contain a specified character at a given position. This can be done by storing an appropriate rank/select data structure of Jacobson~\cite{DBLP:conf/focs/Jacobson89} over the list of suffixes. Our approach also arguably simplifies that of~\cite{DBLP:journals/algorithmica/ChanLSTW10}.

When indexing a text for pattern matching queries with patterns containing $\le k$ wildcards, heavy path-based group trees are not created. For clarity, in \cref{sec:general}, we first describe an index of size $\Oh(n \log^{k-1} n)$ in this setting, achieving the same query time as the corresponding variant of the errata tree.

\section{Basics, Strings, Compact Tries}
We denote $[i \dd j]=[i \dd j+1)=\{i,i+1,\ldots,j\}$.
A string $U$ is a sequence of characters over a finite alphabet. By $|U|$ we denote the length of $U$. 
For every position $i \in [0 \dd |U|)$, we denote the $i$th character of $U$ by $U[i]$. A substring of $U$ is a string of the form $U[i]\cdot U[i+1]\cdots U[j-1]$ for integers $0\le i \le j \le |U|$. A fragment of $U$ is a positioned substring of $U$, that is, a substring of $U$ together with its specified occurrence in $U$. By $U[i \dd j)=U[i \dd j-1]$ we denote the fragment composed of characters $U[i],U[i+1],\ldots,U[j-1]$; if $i=j$, the substring is empty. A fragment $U[i \dd j)$ is a prefix of $U$ if $i=0$ and a suffix if $j=|U|$. By $U^R$ we denote the reversal of $U$, that is, $U^R=U[|U|-1] \cdots U[0]$. For two strings $U$ and $V$, by $\LCP(U,V)$ we denote the length of their longest common prefix.

The Hamming distance of two strings $U$, $V$ of equal length, denoted as $\delta_H(U,V)$, is the total number of positions $i \in [0 \dd |U|)$ such that $U[i] \ne V[i]$. We say that a pattern $P$ of length $m$ has a $k$-mismatch occurrence in $T$ at position $i$ if $\delta_H(T[i \dd i+m),P) \le k$.

A string $U$ on which $k$ (at most $k$) substitutions were performed is called a \emph{$k$-modified string $U$} (\emph{($\le k$)-modified string $U$}, respectively). Given a string $T$, its $k$-modified fragment can be represented in $\Oh(k)$ space by storing the original fragment together with the positions of substitutions and the new characters at these positions.

Let us consider an increasing sequence $a_1<a_2<\cdots<a_\ell$ composed of integers in $[1 \dd r]$. A (restricted) \emph{rank} query, given $x\in [1\dd r]$, returns $i\in [1\dd \ell]$ such that $a_i=x$ or states that no such index $i$ exists. A \emph{select} query, given $i\in [1\dd \ell]$, returns $a_i$.

\begin{theorem}[{Raman, Raman, and Rao~\cite[Theorem 4.6]{DBLP:journals/talg/RamanRS07}}]\label{thm:Raman}
An increasing sequence $a_1 < a_2 < \cdots < a_\ell$ composed of integers in $[1 \dd r]$ can be represented using $\Oh(\ell \ceil{\log(r/\ell)})$ bits of space%
\footnote{The data structure of \cite{DBLP:journals/talg/RamanRS07} actually uses $\ceil{\log \binom{r}{\ell}} + o(\ell) + \Oh(\log \log r)$ bits.  
If $\ell = r$, then $a_i = i$ for each $i \in [1 \dd r]$ and rank/select queries can trivially be answered in $\Oh(1)$ time (a pointer to the data structure can be replaced by a null pointer).  
Otherwise, the first two terms are clearly in $\Oh(\ell \ceil{\log(r/\ell)})$ since $\binom{r}{\ell} \le (re/\ell)^\ell$, whereas the third term can be bounded as follows: if $\ell > \log \log r$, then $\ell \cdot \ceil{\log(r/\ell)} \ge \ell > \log \log r$, and if $\ell \le \log \log r$, then $\ell \cdot \ceil{\log(r/\ell)} \ge \log(r/\ell) \ge \log r - \log \log r = \omega(\log \log r)$.} 
so that rank and select queries are supported in $\Oh(1)$ time.
\end{theorem}

\subsection{Compact Tries}
A \emph{trie} $\M$ of a set of strings $\S$ is a rooted tree for which there is a 1-to-1 correspondence between nodes of $\M$ and prefixes of strings in $\S$. There is an edge with label $a$ from node $u$ to node $v$ in $\M$ if and only if $u$ represents a string $U$ and $v$ represents a string $Ua$ for some character $a$. The root node of $\M$ represents the empty string. Nodes of $\M$ that represent strings from $\S$ are called terminal nodes.
For a node $v$, the concatenation of the labels of edges on the path from the root to $v$ is called the \emph{path label} of $v$. The \emph{string depth} of a node $v$ is the length of the path label of $v$. The \emph{locus} of a string $U$ in $\M$, defined only if $U$ is a prefix of a string from $\S$, is the node $v$ in $\M$ with path label $U$.

A \emph{compact trie} $\T$ of a set of strings $\S$ is a trie of $\S$ in which all non-terminal nodes other than the root with exactly one child are dissolved. Nodes that remain (i.e., the root, branching nodes and terminals) are called explicit, whereas the remaining, implicit nodes can still be addressed indirectly, by specifying their nearest explicit descendant and the length of the path to the descendant in the uncompact trie. By a \emph{node} of $\T$ we mean an explicit or implicit node. Each edge that connects two explicit nodes has a label that is represented as a fragment of one of the strings in $\S$. For every $S \in \S$, we also write $S \in \T$.
A compact trie which $t$ terminals contains $\Oh(t)$ explicit nodes.

The best known example of a compact trie is the suffix tree of the text $T$, the compact trie of all suffixes of $T$. We will be considering compact tries of (selected) suffixes of the text $T$, also called \emph{sparse suffix trees} of $T$, or of ($\le k$)-modified suffixes of the text~$T$. In these cases, the label of every edge can be specified in $\Oh(1)$ or $\Oh(k)$ space, respectively, as a fragment or a ($\le k$)-modified fragment of $T$.

For a length-$n$ text $T$ there is an $\Oh(n)$-sized data structure~\cite{DBLP:conf/latin/BenderF00} based on the suffix tree that computes $\LCP(T[i \dd n),T[j \dd n))$ for any $i,j \in [0 \dd n)$ in $\Oh(1)$ time.

\subsection{Longest Common Prefix Queries on a Compact Trie}
In a $\TreeLCP$ query, we are given a compact trie $\T$ and a pattern $P$, and we are to report the longest common prefix between $P$ and the strings $S \in \T$. The locus of this longest common prefix is also called the node where $P$ \emph{diverges} from $\T$ and is denoted as $\TreeLCP(\T,P)$.

$\TreeLCP(\T,P)$ queries can be extended by specifying a node $v$ in the compact trie $\T$. Such a query, denoted $\TreeLCP_v(\T,P)$ and referred to as an \emph{unrooted} $\TreeLCP$ query, is equivalent to a $\TreeLCP(\T_v,P)$ query for $\T_v$ being the subtree of $\T$ rooted at $v$. If $v$ is the root of $\T$, the query $\TreeLCP_v(\T,P)=\TreeLCP(\T,P)$ is called \emph{rooted}.
In computation of $\TreeLCP$, perfect hashing can be used.

\begin{theorem}[\cite{DBLP:journals/jacm/FredmanKS84}]\label{thm:WExp}
A compact trie $\T$ of $n$ static strings can be stored in $\Oh(n)$ space (in addition to the stored strings themselves) so that a $\TreeLCP_v(\T,P)$ query for node $v$ of $\T$ and a pattern $P$ can be answered in $\Oh(d)$ time, where $d$ is the string depth of the returned node in $\T_v$.
\end{theorem}

Data structures for answering rooted and unrooted $\TreeLCP$ queries on sparse suffix trees were proposed by Cole et al.~\cite{DBLP:conf/stoc/ColeGL04}; the space complexity of unrooted queries was improved by Chan et al.~\cite{DBLP:journals/algorithmica/ChanLSTW10}.

\begin{theorem}[\cite{DBLP:journals/algorithmica/ChanLSTW10,DBLP:conf/stoc/ColeGL04}]\label{thm:unrootedLCP}
Assume we are given compact tries $\T_1,\ldots,\T_t$ of total size $N$, each containing suffixes of a length-$n$ text $T$, and that the suffix tree $\TS$ of $T$ is given. There exists a data structure of size $\Oh(n+N)$ that, given a length-$m$ pattern $P$, preprocesses $P$ in $\Oh(m)$ time so that later an unrooted query $\TreeLCP_v(\T_i,P')$, for any tree $\T_i$, $i \in [1 \dd t]$, its node $v$, and suffix $P'$ of $P$, can be answered in $\Oh(\log \log n)$ time.
\end{theorem}

The following lemma is a component of the data structure behind \cref{thm:unrootedLCP}.

\begin{lemma}[Cole et al.~\cite{DBLP:conf/stoc/ColeGL04}]\label{lem:rootedLCP_special}
The suffix tree $\TS$ of a length-$n$ text $T$ can be enhanced with a data structure of size $\Oh(n)$ that, given a length-$m$ pattern $P$,
returns $\TreeLCP(\TS,U)$ for every suffix $U$ of $P$ in $\Oh(m)$ total time.
\end{lemma}

\section{LCPs on Modified Suffixes}\label{sec:unrooted}
The lemmas in this section help us optimize the $k$-mismatch index for short patterns. We believe that the data structures from this section will find further applications. The proofs of lemmas in this section are deferred to \cref{sec:unrooted_proofs}.

Let $\T$ be a compact trie of a given set of fragments of a length-$n$ text $T$.
We say that $\T$ is given in a \emph{canonical form} if, for every terminal node $v$ of $\T$ that corresponds to a fragment $T[a \dd b)$, either $b=n$ or $v$ does not have an outgoing edge along the character $T[b]$ (see \cref{fig:canonical}). Let us note that a sparse suffix tree of $T$ is automatically in a canonical form. Any compact trie $\T$ of fragments of $T$ can be transformed into a tree of the same shape and with the same number of terminals but in a canonical form by repeating the following correcting process: while $\T$ has a terminal $v$ that corresponds to a fragment $T[a \dd b)$ and $v$ has an outgoing edge with label $T[b]$, move the terminal $v$ one character down along $T[b]$, effectively replacing the fragment $T[a \dd b)$ with $T[a \dd b]$.

\begin{figure}[htpb]
\centering
\renewcommand{\tabcolsep}{0pt}
\newcommand{\dy}{0.6}
\begin{tikzpicture}[scale=0.9]

\begin{scope}
\draw (0,0.5*\dy) -- node[sloped,left,rotate=270] {
\begin{tabular}{c}
a
\end{tabular}
} (-2,-1*\dy);

\draw (-2,-1*\dy) -- node[sloped,left,rotate=270] {
\begin{tabular}{c}
a
\end{tabular}
} (-3,-2*\dy);

\draw (-3,-2*\dy) -- node[sloped,right,rotate=90] {
\begin{tabular}{c}
b
\end{tabular}
} (-3,-3*\dy);

\draw (-3,-3*\dy) -- node[sloped,right,rotate=90] {
\begin{tabular}{c}
a\\b
\end{tabular}
} (-3,-5*\dy);

\draw (-2,-1*\dy) -- node[sloped,right,rotate=90] {
\begin{tabular}{c}
b\\a
\end{tabular}
} (-1,-3*\dy);

\draw (-1,-3*\dy) -- node[sloped,right,rotate=90] {
\begin{tabular}{c}
a
\end{tabular}
} (-1,-4*\dy);

\draw (0,0.5*\dy) -- node[sloped,right,rotate=90] {
\begin{tabular}{c}
b
\end{tabular}
} (2,-1*\dy);

\draw (2,-1*\dy) -- node[sloped,right,rotate=90] {
\begin{tabular}{c}
a\\a
\end{tabular}
} (2,-3*\dy);

\draw (2,-3*\dy) -- node[sloped,right,rotate=90] {
\begin{tabular}{c}
b\\a\\a
\end{tabular}
} (2,-6*\dy);

\foreach \x/\y in {2/-1,2/-3,2/-6,-3/-2,-3/-3,-3/-5,-1/-3,-1/-4}{
    \filldraw (\x,\y*\dy) circle (0.04cm);
}

\draw (2,-1*\dy) node[above right=-0.1cm] {\small$T[4 \dd 4]$};
\draw (2,-3*\dy) node[left] {\small$T[1 \dd 3]$};
\draw (2,-6*\dy) node[below] {\small$T[6 \dd 11]$};
\draw (-3,-2*\dy) node[left] {\small$T[7 \dd 8]$};
\draw (-3,-3*\dy) node[left] {\small$T[10 \dd 12]$};
\draw (-3,-5*\dy) node[left] {\small$T[2 \dd 6]$};
\draw (-1,-3*\dy) node[left] {\small$T[3 \dd 5]$};
\draw (-1,-4*\dy) node[left] {\small$T[5 \dd 8]$};
\end{scope}

\begin{scope}[xshift=9cm]
\draw (0,0.5*\dy) -- node[sloped,left,rotate=270] {
\begin{tabular}{c}
a
\end{tabular}
} (-2,-1*\dy);

\draw (-2,-1*\dy) -- node[sloped,left,rotate=270] {
\begin{tabular}{c}
a\\b
\end{tabular}
} (-3,-3*\dy);

\draw (-3,-3*\dy) -- node[sloped,right,rotate=90] {
\begin{tabular}{c}
a
\end{tabular}
} (-3,-4*\dy);

\draw (-3,-4*\dy) -- node[sloped,right,rotate=90] {
\begin{tabular}{c}
b
\end{tabular}
} (-3,-5*\dy);

\draw (-2,-1*\dy) -- node[sloped,right,rotate=90] {
\begin{tabular}{c}
b\\a
\end{tabular}
} (-1,-3*\dy);

\draw (-1,-3*\dy) -- node[sloped,right,rotate=90] {
\begin{tabular}{c}
a
\end{tabular}
} (-1,-4*\dy);

\draw (0,0.5*\dy) -- node[sloped,right,rotate=90] {
\begin{tabular}{c}
b\\a
\end{tabular}
} (2,-2*\dy);

\draw (2,-2*\dy) -- node[sloped,right,rotate=90] {
\begin{tabular}{c}
a\\b\\a
\end{tabular}
} (2,-5*\dy);

\draw (2,-5*\dy) -- node[sloped,right,rotate=90] {
\begin{tabular}{c}
a
\end{tabular}
} (2,-6*\dy);

\foreach \x/\y in {2/-2,2/-5,2/-6,-3/-3,-3/-4,-3/-5,-1/-3,-1/-4}{
    \filldraw (\x,\y*\dy) circle (0.04cm);
}

\draw (2,-2*\dy) node[above right=-0.1cm] {\small\textcolor{red}{$T[4 \dd 5]$}};
\draw (2,-5*\dy) node[left] {\small\textcolor{red}{$T[1 \dd 5]$}};
\draw (2,-6*\dy) node[below] {\small$T[6 \dd 11]$};
\draw (-3,-3*\dy) node[left] {\small$T[10 \dd 12]$};
\draw (-3,-4*\dy) node[left] {\small\textcolor{red}{$T[7 \dd 10]$}};
\draw (-3,-5*\dy) node[left] {\small$T[2 \dd 6]$};
\draw (-1,-3*\dy) node[left] {\small$T[3 \dd 5]$};
\draw (-1,-4*\dy) node[left] {\small$T[5 \dd 8]$};
\end{scope}

\begin{scope}[yshift=-9*\dy cm]
\draw (-0.75,0) node[above] {$T:$};
\foreach \i/\c in {0/a,1/b,2/a,3/a,4/b,5/a,6/b,7/a,8/a,9/b,10/a,11/a,12/b}{
  \draw (\i*0.5,0) node[above] {\c};
  \draw (\i*0.5,-0.5) node[above] {\footnotesize \i};
}
\end{scope}

\end{tikzpicture}
\vspace*{-1cm}
\caption{Left: a compact trie of a set of fragments of $T$. Right: a compact trie of the same shape given in a canonical form; the updated terminals are shown in red.}\label{fig:canonical}
\end{figure}
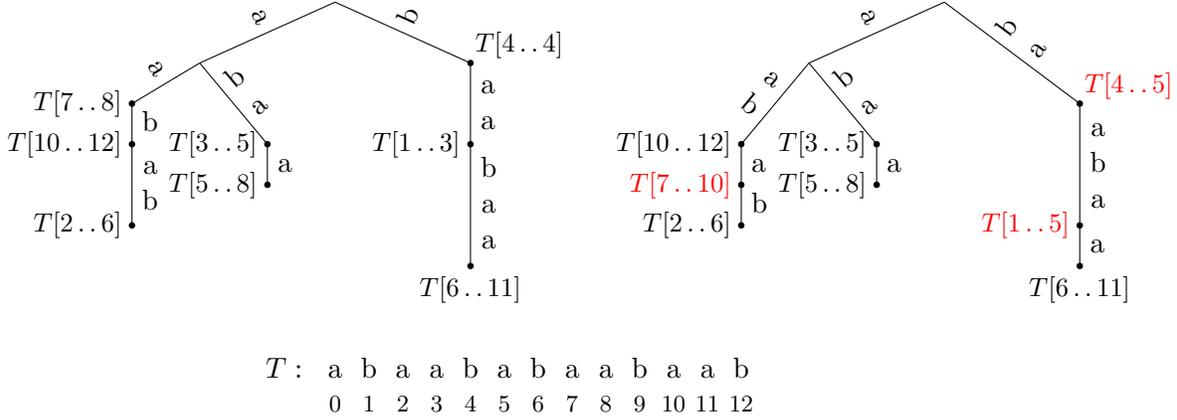

Unrooted $\TreeLCP$ queries can be answered efficiently also on compact tries of fragments of $T$ given in a canonical form. In short, we compute the result of a $\TreeLCP$ query obtained for the compact tries of the corresponding suffixes of $T$ and trim it, which is possible thanks to the original trees being in a canonical form.

\begin{restatable}{lemma}{canonical}\label{thm:unrootedLCPfactors}
\cref{thm:unrootedLCP} holds also if $\T_1,\ldots,\T_t$ are compact tries of fragments of $T$, each given in a canonical form.
\end{restatable}

In the proof of the next lemma, we decompose each compact trie of $k$-modified suffixes into compact tries of fragments of $T$, each given in a canonical form, and very small tries that correspond to modifications.

\begin{restatable}{lemma}{modified}\label{lem:modifiedLCP}
For all constants $k,k'=\Oh(1)$,
\cref{thm:unrootedLCP} holds also if $\T_1,\ldots,\T_t$ are compact tries of ($\le k$)-modified suffixes of $T$ and queries are for ($\le k'$)-modified suffixes of $P$.
\end{restatable}

\noindent
For superconstant $k,k'$, the space of the data structure in the lemma is $\Oh(n+Nk)$ and the query time for an unrooted $\TreeLCP$ query is $\Oh((k+k'+1) \log \log n)$.

The next lemma follows by applying kangaroo jumps~\cite{DBLP:journals/tcs/GalilG87,DBLP:journals/tcs/LandauV86} with the aid of the data structure of \cref{lem:rootedLCP_special}.

\begin{restatable}{lemma}{kangaroo}\label{lem:kangaroo_k}
A length-$n$ text $T$ can be enhanced with a data structure of size $\Oh(n)$ that, given a length-$m$ pattern $P$, preprocesses $P$ in $\Oh(m)$ time so that later, given a $k$-modified suffix $S$ of $T$ and a $k'$-modified suffix $P'$ of $P$, we can compute $\LCP(S,P')$ in $\Oh(k+k'+1)$ time.
\end{restatable}

\section{\texorpdfstring{$k$}{k}-Mismatch Index for \texorpdfstring{$\sigma=\Oh(1)$}{σ=\O(1)} for Short Patterns}\label{sec:short}
We proceed with an index for small $m$ designed for constant $k$ and $\sigma$. We start with a simpler index with worse complexity. Then, in \cref{thm:sk1}, we optimize its complexity.

\begin{fact}\label{fct:sk1slow}
Let $T$ be a string of length $n$ over an alphabet of size $\sigma=\Oh(1)$ and $k=\Oh(1)$. For every $\mu \in [1\dd n]$ and $h \in [1 \dd k)$, there exists a $k$-mismatch index using $\Oh(n\mu^h)$ space that can answer queries for patterns of length $m \le \mu$ in $\Oh(m^{k-h+1}\log\log n+\occ)$ time.
\end{fact}
\begin{proof}
\emph{Data structure.}
For every $\kappa \in [1 \dd h]$ and $j \in [\kappa-1 \dd \mu)$, the data structure stores a compact trie $\T_{\kappa,j}$ of $\kappa$-modified suffixes $S$ of $T$ such that all substitutions in $S$ are at positions $[0 \dd j]$ of $S$ and the rightmost ($\kappa$th) substitution in $S$ is at the $j$th position of $S$. Moreover, the data structure contains the suffix tree of $T$, which we denote as $\T_{0,0}$; for $j > 0$, we assume that $\T_{0,j}$ are empty. For answering rooted $\TreeLCP$ queries for compact tries of ($\le k$)-modified suffixes of $T$ and ($\le k$)-modified patterns $P$, we use \cref{lem:modifiedLCP}.

For a given $\kappa \in [0 \dd h]$, there are $\Oh(n\binom{\mu}{\kappa}\sigma^\kappa) = \Oh(n \mu^h)$ $\kappa$-modified suffixes with modifications at the first $\mu$ positions. Each ($\le h$)-modified suffix belongs to exactly one of the compact tries. Overall, the compact tries use $\Oh(n\mu^h)$ space (as each edge is represented in $\Oh(h)=\Oh(1)$ space). The data structure of \cref{lem:modifiedLCP} uses $\Oh(n\mu^h)$ total space.

\emph{Queries.}
Assume we are given a query pattern $P$ of length $m \le \mu$. First, we list $h$-mismatch occurrences of $P$. To this end, for each $\kappa \in [0 \dd h]$ and $j \in (\kappa \dd m)$, we look for the locus of $P$ in $\T_{\kappa,j}$ and, if it exists, report all the terminals in the subtree of the computed node. Next, we report $k$-mismatch occurrences of $P$ that contain at least $h+1$ substitutions. To this end, we generate all strings $P'$ that are at Hamming distance at least 1 and at most $k-h$ from $P$. For each such string $P'$, if $j$ is the first position where $P$ and $P'$ mismatch, we look for the locus of $P'$ in each compact trie $\T_{h,i}$ for $i \in [h-1 \dd j)$ and, if it exists, report all the terminals in the subtree of the computed node. Each $k$-mismatch occurrence of $P$ is reported exactly once.

The locus of a ($\le k-h$)-modified pattern in a compact trie can be computed in $\Oh(\log \log n)$ time using \cref{lem:modifiedLCP}.
In total, we compute $\Oh(m)$ times the locus of $P$ to compute $h$-mismatch occurrences of $P$ and $\Oh(m^{k-h+1}\sigma^{k-h})$ times the locus of a $(\le k-h)$-modified $P$ to compute the remaining $k$-mismatch occurrences of $P$. The total query time is $\Oh(m^{k-h+1}\log \log n + \occ)$.
\end{proof}

\paragraph{First improvement.}
We use a folklore construction that covers the interval $[1 \dd \mu]$ by $\Oh(\mu)$ (here: just $\mu$) so-called base intervals so that each position in $[1 \dd \mu]$ is covered by $\Oh(\log \mu)$ base intervals and each interval $[1 \dd j]$, for $j \in [1 \dd \mu]$, can be decomposed into $\Oh(\log j)$ base intervals. For a positive integer $t$, by $f(t)$ we denote the maximum integer power of two that divides $t$. With the formula $f(t)=\frac12 ((t\ \mathrm{xor}\ (t-1))+1)$, values $f(t)$ can be computed in $\Oh(1)$ time in the word RAM. We use base intervals $[t-f(t) \dd t)$ for $t \in [1 \dd \mu]$. For a positive integer $j$, its \emph{$f$-sequence} is defined as a decreasing sequence of positive integers that is constructed as follows. We start with element $j$. If the last element was $t$, then the next element is $t-f(t)$, unless $t-f(t)=0$, in which case the sequence is terminated. Then the set of right endpoints of base intervals that cover the interval $[1 \dd j]$ is exactly the $f$-sequence of $j$.

The improved data structure stores, for every $\kappa \in [0 \dd h]$ and $j \in [1 \dd \mu]$, a compact trie $\T'_{\kappa,j}$ that consists of modified suffixes from $\T_{\kappa,i}$ for all $i \in [j-f(j) \dd j)$.
This way the modified suffixes from compact trie $\T_{\kappa,i}$ are stored in $\Oh(\log \mu)$ compact tries $\T'_{\kappa,j}$. Thus, the space complexity grows just by a factor of $\Oh(\log \mu)$. Naturally, we construct data structures for unrooted $\TreeLCP$ queries in tries $\T'_{\kappa,j}$. The total space complexity is $\Oh(n \mu^h \log \mu)$.

The query algorithm is adjusted as follows.
First, $h$-mismatch occurrences of $P$ are reported by computing the locus of $P$ in $\T'_{\kappa,t}$ for all $\kappa \in [0 \dd h]$ and $t$ in the $f$-sequence of $m$. Then, to compute $k$-mismatch occurrences with at least $h+1$ substitutions, we generate all modified patterns $P'$ that are at Hamming distance at least 1 and at most $k-h$ from $P$. For each such string $P'$, if $j$ is the first position where $P$ and $P'$ mismatch, we find the locus of $P'$ in each compact trie $\T'_{h,t}$ for $t$ in the $f$-sequence of $j$. Each $k$-mismatch occurrence of $P$ is still reported exactly once.

The length of an $f$-sequence of $j \in [1 \dd m]$ is $\Oh(\log m)$. We ask $\Oh(k\log m)$ $\TreeLCP$ queries in the compact tries to compute $h$-mismatch occurrences of $P$ and $\Oh(m^{k-h}\log m)$ $\TreeLCP$ queries to compute the remaining $k$-mismatch occurrences. Thus, the time complexity of a query is $\Oh(m^{k-h}\log m \log\log n + \occ)$.

\paragraph{Second improvement.} In the lemma below, we extend rank and select queries of Raman et al.~\cite{DBLP:journals/talg/RamanRS07} to storing a collection of increasing sequences. The proof applies Jensen's inequality.

\begin{lemma}\label{lem:jensen}
A collection of $c$ increasing sequences of total length $\ell$, each composed of integers in $[1 \dd r]$, can be stored using $\Oh(\ell(1+\log(cr/\ell)))$ bits to support rank and select queries on each sequence in $\Oh(1)$ time.
\end{lemma}
\begin{proof}
Let $\ell_1,\ell_2,\ldots,\ell_c$ be the lengths of the subsequent sequences. By \cref{thm:Raman}, the total number of bits required to represent all sequences is $\Oh(c + \sum_{i=1}^c \ell_i (1+\log(r/\ell_i)))$. Since $f(x)=x\log(r/x)$ is concave, this value is maximized for $\ell_i=\ell/c$; then, the total size is as required (as $c \le \ell$).
\end{proof}

Let us fix a tree $\T'_{\kappa,j}$. Let $S_1,S_2,\ldots,S_t$ ($S_1 < S_2 < \cdots < S_t$) be strings corresponding to terminals of $\T'_{\kappa,j}$. We select every $\floor{\log n}$th terminal from the sequence as well as $S_1$ and $S_t$ and form a compact trie $\T''_{\kappa,j}$ of the strings corresponding to the selected terminals. The terminal $S_i$ in $\T''_{\kappa,j}$ is numbered by $i$. The new tries take $\Oh(\mu+n\mu^h\log \mu/\log n)$ space in total. As before, after the improvement we construct data structures for unrooted $\TreeLCP$ queries in tries $\T''_{\kappa,j}$. The trie $\T'_{\kappa,j}$ is not stored; instead, we design a compact data structure that allows us to retrieve $S_i$ in $\Oh(1)$ time using \cref{lem:jensen}, as shown in the lemma below.

\newcommand{\LL}{\mathcal{L}}
\begin{lemma}\label{lem:claim}
If $\mu=\Omega(\log n)$, there is a data structure of size $\Oh(n\mu^h\log^2 \mu/\log n)$ that, given $\kappa \in [0 \dd h]$, $j \in [1 \dd \mu]$, and $i \in [1 \dd |\T'_{\kappa,j}|]$, returns the $i$-th lexicographically terminal of $\T'_{\kappa,j}$ in $\Oh(1)$ time.
\end{lemma}
\begin{proof}
\emph{Data structure.} We store a lexicographically sorted list $\LL$ of all ($\le h-1$)-modified suffixes of $T$ with modifications at the first $\mu$ positions only. The subsequent elements of $\LL$ are denoted as $\LL_1,\LL_2,\ldots,\LL_{|\LL|}$. Each modified suffix can be represented in $\Oh(h)=\Oh(1)$ space and $|\LL| = \Oh(n\mu^{h-1})$, so $\Oh(n\mu^{h-1}) \subseteq \Oh(n \mu^h/\log n)$ space is required.

\newcommand{\oold}{\mathit{old}}
\newcommand{\nnew}{\mathit{new}}
\newcommand{\Short}{\mathit{Short}}
\newcommand{\MyTriad}{\mathit{MyTriad}}
\newcommand{\Cut}{\mathit{Cut}}
\newcommand{\Triads}{\mathit{Triads}}

For each terminal ($\le h$)-modified suffix $S$ in $\T'_{\kappa,j}$, a triad $(p,c_\oold,c_\nnew)$ is stored such that the rightmost modification in $S$ is at position $p$, the character before modification is $c_\oold$, and the character after modification is $c_\nnew$. The sequence of triads for all terminals of $\T'_{\kappa,j}$ in lexicographic order is denoted as $\Triads_{\kappa,j}$. If $S$ contains no modifications, we can set $p=-1$. Each triad is represented in $\Oh(\log \mu)$ bits and the number of triads is the same as the total size of tries $\T'_{\kappa,j}$, that is, $\Oh(n\mu^h \log \mu)$. Hence, the space complexity of this component $\Oh(n\mu^h \log^2 \mu/\log n)$.

For each trie $\T'_{\kappa,j}$ and triad $(p,c_\oold,c_\nnew)$, we store an increasing sequence $\MyTriad_{\kappa,j,p,c_\oold,c_\nnew}$ of integers in $[1 \dd |\T'_{\kappa,j}|]$ such that $x$ occurs in this sequence if and only if the lexicographically $x$th terminal of $\T'_{\kappa,j}$ has a triad $(p,c_\oold,c_\nnew)$ (i.e., $\Triads_{\kappa,j}[x]=(p,c_\oold,c_\nnew)$). Thus, for a given trie $\T'_{\kappa,j}$, we form $\Oh(\mu)$ increasing sequences of total length $|\T'_{\kappa,j}|$ and values in $[1 \dd |\T'_{\kappa,j}|]$ and store them using \cref{lem:jensen}. This part of the data structure takes $\Oh(|\T'_{\kappa,j}| \log \mu/\log n)$ space per each trie, for a total of $\Oh(n\mu^h\log^2 \mu/\log n)$ space.

Finally, for each trie $\T'_{\kappa,j}$ and valid triad $(p,c_\oold,c_\nnew)$, we store a sequence $\Cut_{\kappa,j,p,c_\oold,c_\nnew}$ of integers in $[1 \dd |\LL|]$ such that $x$ occurs as the $i$th element in this sequence if and only if the lexicographically $i$th terminal with triad $(p,c_\oold,c_\nnew)$ in $\T'_{\kappa,j}$ (i.e., the globally $\MyTriad_{\kappa,j,p,c_\oold,c_\nnew}[i]$'th terminal in the trie) without its rightmost modification (i.e., the one at position $p$) equals $\LL_x$. The key observation is that the sequence $\Cut_{\kappa,j,p,c_\oold,c_\nnew}$ is \emph{increasing} because all the considered terminals have the same character at position $p$ ($c_\nnew$) and all the terminals without the last modification also have the same character at this position ($c_\oold$). Here $\Oh(\mu^2)$ increasing sequences are stored using \cref{lem:jensen}. Each terminal of each trie implies exactly one element in one of the sequences, so the total length of the sequences is $\Oh(n\mu^k\log \mu)$ and their values are in $[1 \dd |\LL|]$; thus again $\Oh(n\mu^h \log^2 \mu/\log n)$ space is used.

\emph{Queries.} Say we need to compute a terminal specified by indices $\kappa,j,i$. First we recover the triad $\Triads_{\kappa,j}[i]=(p,c_\oold,c_\nnew)$. With a rank query in sequence $\MyTriad_{\kappa,j,p,c_\oold,c_\nnew}$, we recover an index $a$ of this sequence that stores element $i$. With a select query on sequence $\Cut_{\kappa,j,p,c_\oold,c_\nnew}$, we compute the $a$th element of this sequence, say $x$. 
In the end, we return the ($\le h-1$)-modified string $\LL_{x}$ with an additional substitution $\LL_{x}[p]:= c_\nnew$. Each step takes $\Oh(1)$ time by \cref{lem:jensen}.
\end{proof}

\begin{remark}
Without the assumption that $\mu = \Omega(\log n)$, the space complexity in \cref{lem:claim} becomes $\Oh(n\mu^{h-1}+n \mu^h/\log n)$.
\end{remark}

The next simple technical lemma is also reused in our general compact index. We use static predecessor/successor data structure. Such a data structure is constructed for a specified set $A$ of integers. A predecessor query to $A$, given an integer $x$, returns $\max\{y \in A\,:\,y \le x\}$.
A successor query returns $\min\{y \in A\,:\,y \ge x\}$.
Using $y$-fast tries~\cite{DBLP:journals/ipl/Willard83} one can obtain a static $\Oh(m)$-size data structure on a set $A$ of $m$ sorted keys from $[1 \dd n]$ that answers predecessor/successor queries in $A$ in $\Oh(\log \log n)$ time.

\begin{lemma}\label{lem:utility}
Let $k=\Oh(1)$ and $\T_1,\ldots,\T_t$ be compact tries of ($\le k$)-modified suffixes of $T$. Further, let $\T'_i$, for $i \in [1 \dd t]$, be a compact trie formed by selecting every $\floor{\log n}$th terminal of $\T_i$ in the lexicographic order as well as the first and the last terminal. We assume that tries $\T_1,\ldots,\T_t$ are not stored; instead, we can compute in $\Oh(1)$ time the lexicographically $j$th terminal of $\T_i$. 

Let $k'=\Oh(1)$ and $P$ be a string pattern such that for every ($\le k'$)-modified suffix $P'$ of $P$, any query $\TreeLCP(\T'_i,P')$ can be answered in $\Oh(\log \log n)$ time and any ($\le k$)-modified suffix of $T$ can be lexicographically compared with $P'$ in $\Oh(1)$ time. Using $\Oh(\sum_{i=1}^t|\T'_i|)$ additional space, for a given index $i \in [1 \dd t]$ and ($\le k'$)-modified suffix $P'$ of $P$, we can compute the range of terminals in $\T_i$ that have $P'$ as a prefix in $\Oh(\log \log n)$ time.
\end{lemma}
\begin{proof}
Assume that terminals in $\T'_i$ are numbered the same as the original terminals from $\T_i$.
Whenever the locus of a $(\le k')$-modified suffix $P'$ of $P$ in $\T_i$ is to be computed, a $\TreeLCP(\T'_i,P')$ query is issued instead.
Let $v$ be the returned node. The remainder of the algorithm depends on if the string depth of $v$ is $|P'|$. If so, let $x_{\mathit{min}}$ and $x_{\mathit{max}}$ be the minimum and maximum number of a terminal in the subtree of $v$ in $\T'_i$. The desired sublist of numbers of terminals of $\T_i$ contains the interval $[x_{\mathit{min}} \dd x_{\mathit{max}}]$. Let $y_{\mathit{max}}$ and $z_{\mathit{min}}$ be the maximum number of a terminal in $\T'_i$ that is smaller than $x_{\mathit{min}}$ and the minimum number of a terminal that is greater than $x_{\mathit{max}}$, respectively. If $y_{\mathit{max}}$ exists, we need to further inspect terminals of $\T_i$ with numbers in $(y_{\mathit{max}}\dd x_{\mathit{min}})$. Similarly, if $z_{\mathit{min}}$ exists, we need to further inspect terminals of $\T_i$ with numbers in $(x_{\mathit{max}}\dd z_{\mathit{min}})$.

Now assume that the string depth $d$ of $v$ is smaller than $|P'|$. Let $x_{\mathit{left}}$ be the maximum number of a terminal in $\T'_i$ such that the represented modified suffix $S$ satisfies $S < P'[0 \dd d]$. If $x_{\mathit{left}}$ does not exist, we are done. Otherwise, let $x_{\mathit{right}}$ be the minimum number of a terminal in $\T'_i$ that is greater than $x_{\mathit{min}}$. Then we need to inspect terminals of $\T'_i$ with numbers in $(x_{\mathit{left}}\dd x_{\mathit{right}})$.

\begin{claim}
In each case, the required terminal numbers can be retrieved from $\T'_i$ in $\Oh(\log \log n)$ time using $\Oh(|\T'_i|)$ data stored for $\T'_i$.
\end{claim}
\begin{proof}
We store the following data for $\T'_i$:
(1) for each explicit node $v$, the minimal and maximal number of terminal in $v$'s subtree;
(2) for each terminal number in $\T'_i$, its predecessor and successor terminal number;
(3) for each explicit node $v$, the maximum terminal with label that does not exceed the path label of $v$;
(4) for each explicit node $v$, a predecessor data structure of~\cite{DBLP:journals/ipl/Willard83} on children of $v$ indexed by the first characters of their edges.

If the string depth of $v$ is $|P'|$, we compute $x_{\mathit{min}}$ and $x_{\mathit{max}}$ using (1) and $y_{\mathit{max}}$, $z_{\mathit{min}}$ using (2). Otherwise, $x_{\mathit{left}}$ is computed as follows. Using (4), we compute the child of $v$ along the maximum character smaller than $P'[d]$; if such a child exists, we report the maximal terminal in its subtree using (1). If it does not exist, we choose the terminal according to (3). The predecessor query takes $\Oh(\log \log n)$ time and the remaining queries take $\Oh(1)$ time. The space of the data structures (1)--(4) is $\Oh(|\T'_i|)$.
\end{proof}

In both cases, after $\Oh(\log \log n)$ time, we are left with $\Oh(1)$ intervals of terminal numbers of length $\Oh(\log n)$ to be inspected. We can retrieve the $j$th terminal of $\T_i$ in $\Oh(1)$ time and compare it lexicographically with $P'$ in $\Oh(1)$ time. Thus, we can use binary search to find a subinterval of a given interval of terminals that corresponds to terminals whose prefix is $P'$. This requires just $\Oh(\log \log n)$ time.
\end{proof}

We apply \cref{lem:utility} to tries $\T'_{\kappa,j}$ and tries after sampling $\T''_{\kappa,j}$.
By \cref{lem:claim}, we can compute the $j$th terminal of $\T'_{\kappa,j}$ in $\Oh(1)$ time using $\Oh(n \mu^h \log^2 \mu /\log n)$ space. By \cref{lem:kangaroo_k}, we can compare any ($\le k'$)-modified suffix $P'$ of $P$ with any ($\le k$)-modified suffix $S$ of $T$ in $\Oh(1)$ time using just $\Oh(n)$ space. With \cref{lem:utility}, we can compute an interval of terminals of $\T'_{\kappa,j}$ that have $P'$ as a prefix. This gives a total of $\Oh(m^{k-h} \log m \log \log n)$ time. Then all occurrences from the interval can be reported in $\Oh(1)$ time per occurrence, using \cref{lem:claim}.

\medskip
With both improvements, we obtain \cref{thm:sk1}. This, in turn, yields the trade-off data structure.

\thmtwo*
\begin{proof}
Recall that the index of \cref{thm:sk1} takes $\Oh(n\mu^h \log^2 \mu/\log n)$ space and answers queries for patterns of length $m \le \mu$ in $\Oh(m^{k-h} \log m \log\log n + \occ)$ time.
The index of \cref{thm:long} takes $\Oh(n+ n \log^{k+\eps} n / \gamma)$ space and answers queries for patterns of length $m \ge (k+1)\gamma$ in $\Oh(\log^k n \log \log n + m + \occ)$ time.

Assume first that $k$ is even.
We use the data structure of \cref{thm:sk1} for $\mu = \floor{ \log^{\frac{2k+2}{k+2}} n }$ with $h=k/2$ and the data structure for long patterns (\cref{thm:long}) for $\gamma=\mu/(k+1)$.
The former uses
$$\Oh(n \log^{\frac{k(k+1)}{k+2}-1} n \cdot (\log\log n)^2) = \Oh(n \log^{k-2+\frac{2}{k+2}} n \cdot (\log \log n)^2) \subset \Oh(n \log^{k-2+\eps+\frac{2}{k+2}} n)$$
space and answers queries in $\Oh(\log^{k-1+\frac{2}{k+2}} n (\log \log n)^2 + \occ) \subset \Oh(\log^k n + \occ)$ time.
The latter uses $\Oh(n \log^{k-2+\eps+\frac{2}{k+2}} n)$ space and answers queries in $\Oh(\log^k n \log \log n + m + \occ)$ time.
The bounds for even $k$ follow.

For odd $k$, we use the data structure of \cref{thm:sk1} for $\mu = \floor{ \log^{\frac{2k}{k+1}-\eps'} n }$ with $\eps'=\frac{2}{k+1}\eps$ and $h=(k-1)/2$ and the data structure for long patterns (\cref{thm:long}) for $\gamma=\mu/(k+1)$.
The former uses
$$\Oh(n \log^{\frac{k(k-1)}{k+1}-1-\frac{k-1}{2}\eps'} n (\log\log n)^2) \subset \Oh(n \log^{k-3+\eps+\frac{2}{k+1}} n (\log \log n)^2) \subset \Oh(n \log^{k-2+\eps+\frac{2}{k+1}} n)$$
space and answers queries in $\Oh(\log^{k-\frac{k+1}{2}\eps'} n (\log \log n)^2 + \occ)=\Oh(\log^{k-\eps} n + \occ)$ time.
The latter uses $\Oh(n \log^{k-2+\eps'+\frac{2}{k+1}} n)\subset\Oh(n \log^{k-2+\eps+\frac{2}{k+1}} n)$ space; query time is $\Oh(\log^k n \log \log n + m + \occ)$.
The bounds for odd $k$ follow.
\end{proof}

\section{Errata Tree for Hamming Distance}\label{sec:errata}
The data structure for $k$-mismatch indexing of Cole et al.~\cite{DBLP:conf/stoc/ColeGL04} 
actually solves an ``all-to-all'' problem, in which we are given $x$ suffixes $R_1,R_2,\ldots,R_x$ of the text $T$, and for a query with pattern $P$, given a suffix $P'$ of $P$ we are to report all $i \in [1 \dd x]$ such that $P'$ matches a prefix of $R_i$ with up to $k$ substitutions. It is assumed that the suffix tree of the text $T$ is accessible. The next theorem summarizes the complexity of the data structure.

\begin{theorem}[\cite{DBLP:conf/stoc/ColeGL04}]\label{thm:kerrata}
Let $T$ be a text of length $n$. The $k$-errata tree for Hamming distance for $x$ suffixes $R_1,\ldots,R_x$ of $T$ uses space $\Oh(n+x \log^k x)$ and, for a pattern $P$ of length $m$, after $\Oh(m)$ time preprocessing, for any suffix $P'$ of $P$ can compute the set $\{i \in [1 \dd x]\,:\,\delta_H(P',R_i[0 \dd |P'|)) \le k\}$ in $\Oh(\log^k x\log\log x + \occ)$ time, where $\occ$ is the size of the set.
\end{theorem}

Let us list basic properties of the $k$-errata tree.
The $k$-errata tree is a collection of compact tries. Each compact trie is assigned a level in the data structure. The compact trie at level 0 is the compact trie of the suffixes $R_1,\ldots,R_x$. At each level $i \in [1 \dd k]$, each compact trie is of one of the following two types:
\begin{enumerate}[label=(\arabic*)]
\item\label{it1} a sparse suffix tree of $T$ attached to a root with a wildcard-character edge; or 
\item\label{it2} a rooted path $\pi$ with several sparse suffix trees of $T$ attached at different locations of $\pi$.
\end{enumerate}
The rooted path $\pi$ is called the \emph{main path} of the compact trie. Each main path has a string label being a suffix of $T$. For each sparse suffix tree in the $k$-errata tree (in particular, for each main path), an unrooted $\TreeLCP$ data structure is stored (\cref{thm:unrootedLCP}).

A $(k-1)$-errata tree is simply composed of compact tries at levels $0,1,\ldots,k-1$ of a $k$-errata tree.

Let us describe in more detail how compact tries at subsequent levels of the $k$-errata tree are formed. 
Let $\T$ be a compact trie at level $i \in [0 \dd k)$. We use a heavy-light decomposition of $\T$, that is, we classify explicit nodes of $\T$ as light and heavy. The root is always light. A non-root node is called heavy if the size of its subtree (measured as the number of terminals in the subtree) is maximal among all the subtrees of its siblings (breaking ties arbitrarily); otherwise the node is light. A heavy path is a path that starts at a light node and goes down through heavy nodes until a heavy leaf is reached.

Let $\pi$ be a heavy path of $\T$. A subtree of $\T$ rooted at an explicit child $u$ of an explicit node $v$ from $\pi$, together with the compact edge from $v$ to $u$, is called an \emph{off-path subtree} if $u$ is not on $\pi$. The off-path subtree such that the first character on the edge from $v$ to $u$ is $a$ is denoted as $\T_{v,a}$.

The following types of \emph{substitution trees} are formed from off-path subtrees:
\begin{enumerate}[label=(\alph*)]
\item\label{a} $\T^s_{v,a}$ which is $\T_{v,a}$ with the first character substituted from $a$ to a special character $\psi$
\item\label{b} $\T^s_v$ which is the merge of all substitution trees $\T^s_{v,a}$, over all characters $a$ different from the next character $b$ after $v$ on the heavy path $\pi$, with character $\psi$ replaced by character $b$.
\end{enumerate}
Here, a merge of subtrees is defined naturally as a compact trie representing all terminals from subsequent subtrees.

Finally, substitution trees are organized into group trees. Specifically, substitution trees of type~\ref{a} are merged into group trees for each node $v$ separately, and substitution trees of type \ref{b} are merged over a whole heavy path. For a collection of substitution trees, a group tree being the merge of all substitution trees is formed, and then it is subdivided into smaller group trees consisting of subsets of substitution trees. The order of merging is lexicographic by character $a$ in type~\ref{a} and by the string depth of node $v$ in type~\ref{b}. Let $\T_1,\T_2,\ldots,\T_r$ be all substitution trees being merged into group trees, for a given node $v$ in case \ref{a} or for a given path $\pi$ in case \ref{b}. Then the group trees are organized as in weighted search trees so that for a given substitution tree $\T_j$, the sizes (i.e., numbers of terminals) of group trees that contain $\T_j$ decrease top to bottom, and halve at least in every second step. The group trees become compact tries at level $i+1$ of the data structure.

Technically, for substitution trees of type \ref{b}, in order to make the merging meaningful, each substitution tree is assumed to contain the path from $v$ to the root of the heavy path $\pi$. Then the substitution tree (hence, the resulting group tree) is actually a compact trie of 1-modified suffixes of $T$. The modification takes place only on the heavy path $\pi$, so the trees hanging from the main path in the group tree are sparse suffix trees of $T$. We treat the main path as heavy in the next step of recursion.


Let $\T$ be a compact trie in the $k$-errata tree. Each terminal in $\T$ can be assigned a \emph{label} $\ell$ being a position in $[0 \dd n)$. (Actually, it is a position in $\{n-|R_j|\,:j \in [1\dd x]\}$.) 
If a node $v$ in $\T$ is computed as a result of a query, then for every terminal in the subtree of $v$, the query algorithm reports a $k$-mismatch occurrence at position equal to the terminal's label.

We prove the following lemma using the recursive construction of $k$-errata trees.

\begin{lemma}\label{lem:kerrata_copies}
In a $k$-errata tree built from $x$ suffixes of $T$ for $k=\Oh(1)$, each terminal label occurs in compact tries at most $\Oh(\log^k x)$ times overall.
\end{lemma}
\begin{proof}
Let $\ell$ be a terminal label in the compact trie at level 0. If $\T$ is a compact trie at level $i$ and contains a terminal with label $\ell$, then by construction every group tree formed from $\T$ contains at most one terminal with label $\ell$. By induction, every terminal label is present in every compact trie of a $k$-errata tree at most once.

Let $\T$ be a compact trie at level $i \in [0 \dd k)$ in the $k$-errata tree that contains a terminal with label $\ell$. For a heavy path $C$ in $\T$, by $\T_C$ we denote the subtree of $\T$ rooted at the root of $C$. Let $C_1,C_2,\ldots,C_t$ be all heavy paths, top-down, visited on the path from the root of $\T$ to the terminal with label $\ell$. The organization of group trees implies that, for every $j \in [1 \dd t)$, the terminal label $\ell$ occurs in $\Oh(1 + \log |\T_{C_j}| - \log |\T_{C_{j+1}}|)$ group trees (of both types) formed from heavy path $C_j$; for $j=t$, in $\Oh(\log |\T_{C_t}|)$ group trees. This telescopes to $\Oh(\log |\T_{C_1}|)=\Oh(\log |\T|)=\Oh(\log x)$ group trees, i.e., $\Oh(\log x)$ compact tries at level $i+1$ formed from $\T$ that contain a terminal with label $\ell$. Overall, terminal label $\ell$ is present in $\Oh(\sum_{i=0}^k \log^i x)=\Oh(\log^k x)$ compact tries.
\end{proof}

A $k$-mismatch query is asked for a pattern $P$ for the compact trie at level 0 and with a budget of $k$ mismatches and is handled recursively, level by level. A query in a compact trie $\T$ at level $i$ for a suffix $P'$ of $P$ with budget $k'$ reduces to possibly several queries for suffixes of $P'$ in compact tries at level $i+1$ and budget at most $k'-1$. In case a $k$-mismatch occurrence of $P$ in $T$ is actually a $k'$-mismatch occurrence for some $k'<k$, it is reported in a compact trie at level $\le k'$. In this case, a node $v$ in the compact trie is computed and all terminals in the subtree of $v$ are reported. In total, $\Oh(\log^k x)$ such nodes are returned in a query of \cref{thm:kerrata}.

When the query algorithm is called recursively for a compact trie $\T$, a suffix $P'$ of the pattern $P$, and budget $k'$, its goal is to report labels of all terminals in $\T$ whose length-$|P'|$ prefix is at Hamming distance at most $k'$ from $P'$. Some of these terminals are located directly by traversing $\T$; others are identified in recursive calls.

\section{\texorpdfstring{$k$}{k}-Mismatch Index for Long Patterns}\label{sec:long}
For a string $U$, an integer $p \in [1 \dd |U|]$ is called a \emph{period} of $U$ if $U[i]=U[i+p]$ holds for all $i \in [0 \dd |U|-p)$. The smallest period of $U$ is denoted as $\per(U)$.

We use the notion of string synchronizing sets of Kempa and Kociumaka~\cite{DBLP:conf/stoc/KempaK19}. For a string $T$ of length $n$ and a positive integer $\tau \leq n/2$, a set $A\subseteq [1\dd n-2\tau+1]$ is a \emph{$\tau$-synchronizing set} of $T$ if it satisfies the following two conditions:
\begin{enumerate}
    \item\label{item1} Consistency: If $T[i\dd i+2\tau)$ matches $T[j\dd j+2\tau)$, then $i\in A$ if and only if $j\in A$.
    \item\label{item2} Density: For $i\in [1\dd n-3\tau +2]$, $A\cap [i\dd i+\tau)=\emptyset$ if and only if $\per(T[i\dd i+3\tau -1))\leq \tau/3$.
\end{enumerate}
Let us fix a $\tau$-synchronizing set $A$. If a position $i$ belongs to $A$, we call $T[i \dd i+2\tau)$ a \emph{$\tau$-synchronizing fragment} of $T$.

\begin{theorem}[{\cite[Proposition 8.10]{DBLP:conf/stoc/KempaK19}}]\label{thm:sss}
For a string $T$ of length $n$ and $\tau\le n/2$, there exists a $\tau$-synchronizing set of $T$ of size $\Oh(n/\tau)$.
\end{theorem}

We note that the synchronizing set can be constructed in $\Oh(n)$ time if $T$ is over an integer alphabet.

The case that a long substring of a string does not contain synchronizing positions is covered by considering periodic substrings of $T$ called \emph{runs}.
String $U$ is called \emph{periodic} if $2 \cdot \per(U) \le |U|$. A \emph{run} in $T$ is a periodic substring $T[i \dd j]$ that is maximal, i.e.\ such that $T[i-1] \ne T[i-1+p]$ and $T[j+1] \ne T[j+1-p]$ for $p=\per(T[i \dd j])$ provided that the respective positions exist. All runs in a string can be computed in linear time~\cite{DBLP:journals/siamcomp/BannaiIINTT17,DBLP:conf/icalp/Ellert021,DBLP:conf/focs/KolpakovK99}. A run $R$ is called a \emph{$\tau$-run} if $|R| \ge 3\tau-1$ and $\per(R) \le \tau/3$. We use the following lemma.

\begin{lemma}[{\cite[Section 6.1.2]{DBLP:conf/stoc/KempaK19}}, {\cite[Lemma 10]{DBLP:conf/esa/Charalampopoulos21}}]\label{lem:tauruns}
A string of length $n$ contains $\Oh(n/\tau)$ $\tau$-runs, for $\tau \in [1 \dd n]$.
\end{lemma}

We also use the notion of \emph{misperiods} that originates from Charalampopoulos et al.~\cite{DBLP:journals/jcss/Charalampopoulos21}. We say that position $a$ in string $S$ is a \emph{misperiod} with respect to a substring $S[i \dd j)$ if $S[a] \ne S[b]$ where $b$ is the unique position such that $b \in [i \dd j)$ and $(j-i)$ divides $(b-a)$.
We define the set $\LeftMis_k(S,i,j)$ as the set of $k$ maximal misperiods that are smaller than $i$ and $\RightMis_k(S,i,j)$ as the set of $k$ minimal misperiods that are at least $j$.
Each of the sets can have less than $k$ elements if the corresponding misperiods do not exist.
We also define $\Mis_k(S,i,j)=\LeftMis_k(S,i,j) \cup \RightMis_k(S,i,j)$ and $\Mis(S,i,j) = \Mis_{|S|}(S,i,j)$.

\begin{lemma}[{\cite[Lemma 13]{DBLP:journals/jcss/Charalampopoulos21}}]\label{lemma:Bringmann}
  Assume that $|U|=|V|$, $\delta_H(U,V) \le k$, and that $U[i \dd j) = V[i \dd j)$.
  Let
    $I=\Mis_{k+1}(U,i,j) \text{ and } I'=\Mis_{k+1}(V,i,j).$
  If $I \cap I' = \emptyset$, then $\delta_H(U,V) = |I| + |I'|$, $I=\Mis(U,i,j)$, and $I'=\Mis(V,i,j)$.
\end{lemma}

We design a data structure that answers $k$-mismatch pattern queries with patterns of length $m \ge (k+1)\gamma$, for a specified integer $\gamma \ge 2$. We select two sets of positions in $T$, $A_1$ and $A_2$, called \emph{anchors}.
The set $A_1$ is a $\tau$-synchronizing set of $T$ for $\tau=\floor{\gamma/3}$. The set $A_2$ is a union of sets $\Mis_{k+1}(T,i,i+p)$ over all $\tau$-runs $T[i \dd j]$ in $T$ where $p=\per(T[i \dd j])$.

\begin{lemma}\label{lem:A1A2}
Let $T$ be a text of length $n$ and $\gamma \in [2 \dd \floor{n/(k+1)}]$. Then the sets of anchors $A_1$, $A_2$ for $T$ satisfy $|A_1| = \Oh(n/\gamma)$ and $|A_2|=\Oh(nk/\gamma)$.
\end{lemma}
\begin{proof}
By \cref{thm:sss}, $|A_1|=\Oh(n/\tau)$. By \cref{lem:tauruns}, $|A_2|=\Oh(nk/\tau)$. Finally, $\tau=\Theta(\gamma)$.
\end{proof}

Given a query for a pattern $P$ of length $m \ge (k+1)\gamma$, we select two sets of positions in $P$, $B_1$ and $B_2$, also called \emph{anchors}.
We partition a prefix of $P$ into $k+1$ substrings of length $\gamma$. For each $i \in [0 \dd k]$, if $P[i\gamma \dd (i+1)\gamma)$ is a substring of the text $T$, the set $B_1$ contains the starting position $j$ of the leftmost $2\tau$-length string that occurs in $P[i\gamma \dd (i+1)\gamma)$ and is a $\tau$-synchronizing fragment of $T$, provided that such $j$ exists. That is, $j=\min\{j' \in [i\gamma \dd (i+1)\gamma-2\tau)\,:\,P[j' \dd j'+2\tau)=T[a \dd a+2\tau)\text{ for some }a \in A_1\}$, provided that the set under $\min$ is non-empty.

For every $i \in [0 \dd k]$, if $P[i\gamma \dd (i+1)\gamma)$ is a substring of the text $T$ and $p:= \per(P[i\gamma \dd (i+1)\gamma)) \le \tau/3$, the misperiods $\Mis_{k+1}(P,i\gamma,i\gamma+p)$ are inserted to $B_2$.

\begin{lemma}\label{lem:B1B2}
Let $T$ be a text of length $n$ and $\gamma \in [1 \dd \floor{n/(k+1)}]$. There is a data structure of size $\Oh(n)$ using which, given a pattern $P$ of length $m \ge (k+1)\gamma$, one can construct its sets of anchors $B_1$, $B_2$ in $\Oh(mk)$ time. Moreover, we have $|B_1| \le k+1$ and $|B_2| \le 2(k+1)^2$.
\end{lemma}
\begin{proof}
As for the sizes of sets, we have $|B_1|\le k+1$ as in each $P[i\gamma \dd (i+1)\gamma)$, for $i \in [0 \dd k]$, we select at most one position. Moreover, $|B_2| \le 2(k+1)^2$ as for each $P[i\gamma \dd (i+1)\gamma)$ that is highly periodic, at most $k+1$ left misperiods and at most $k+1$ right misperiods are inserted to $B_2$.

\emph{Data structure:}
We store the suffix tree of $T$ enhanced with the data structure for $\TreeLCP$ computation of \cref{thm:WExp}.
Each node of the suffix tree stores the label of some terminal node in its subtree.
We also use the data structure of \cref{lem:rootedLCP_special}.
Finally, for each position $i \in [0 \dd n)$, we store the smallest anchor in $A_1 \cap [i \dd n)$, if it exists.
The data structures use $\Oh(n)$ space.

\emph{Queries:}
For each $i \in [0 \dd k]$, we check if $P[i\gamma \dd (i+1)\gamma)$ has a locus in the suffix tree of $T$.
This takes $\Oh(m)$ time in total by \cref{thm:WExp}.
We compute anchors based on values of $i$ for which the locus exists.

To compute the set $B_1$, if the locus is $v$, let $\ell$ be the label of any terminal node in the subtree of $v$.
We set $a = \min(A_1 \cap [\ell \dd n))$.
If $a$ exists and $a-\ell < \gamma-2\tau$, $a-\ell+i\gamma$ is inserted to $B_1$.
The algorithm uses just $\Oh(k)$ time in total.
Correctness of the query algorithm follows by the consistency condition of a synchronizing set.

To compute the set $B_2$ for a given $i \in [0 \dd k]$, we compute the smallest period $p$ of $P[i\gamma \dd (i+1)\gamma)$ using the Morris-Pratt algorithm~\cite{DBLP:journals/siamcomp/KnuthMP77} in $\Oh(\gamma)$ time, for a total of $\Oh(m)$ time. If $p \le \tau/3$, we compute the set $\Mis_{k+1}(P,i\gamma,i\gamma+p)$ by definition in $\Oh(m)$ time and insert its elements to $B_2$. We spend $\Oh(m)$ time for a given $i$, which totals in $\Oh(km)$ time.
\end{proof}

\begin{definition}\label{def:nearly}
We say that pattern $P$ has a \emph{$k$-nearly periodic occurrence} $T'=T[j \dd j+m)$ in text $T$ if $\delta_H(P,T') \le k$ and
there are integers $p \in [1 \dd \floor{\tau/3}]$ and $i \in [0 \dd k]$ such that
$\per(P[i\gamma \dd (i+1)\gamma))=p$,
$P[i\gamma \dd (i+1)\gamma)=T'[i\gamma \dd (i+1)\gamma)$ and $\delta_H(P,T')=|\Mis(P,i\gamma,i\gamma+p)|+|\Mis(T',i\gamma,i\gamma+p)|$.
\end{definition}

For two sets $A$, $B$ of integers, we denote $A \oplus B = \{a+b\,:\,a \in A,\,b \in B\}$ and $A \ominus B = \{a-b\,:\,a \in A,\,b \in B\}$.
For an integer $j$, we denote $A \oplus j = A \oplus \{j\}$, $A \ominus j = A \ominus \{j\}$.

\begin{lemma}\label{lem:anchors}
Let $T$ be a text of length $n$, $\gamma \in [1 \dd \floor{n/(k+1)}]$, and $P$ be a pattern of length $m \ge (k+1)\gamma$. Every $k$-mismatch occurrence of $P$ in $T$ starts at a position in $(A_1 \ominus B_1) \cup (A_2 \ominus B_2)$ or is a $k$-nearly periodic occurrence.
\end{lemma}
\begin{proof}
Let $j \in [0 \dd n-m]$ satisfy $\delta_H(T[j \dd j+m),P) \le k$. We will show that $j \in (A_1 \ominus B_1) \cup (A_2 \ominus B_2)$ or $T'=T[j \dd j+m)$ is a $k$-nearly periodic occurrence.

As $m \ge (k+1)\gamma$, at least one of the fragments $P[i\gamma\dd(i+1)\gamma)$, for $i \in [0 \dd k]$, matches the corresponding fragment $T[j+i\gamma \dd j+(i+1)\gamma)$ exactly. Let it be the fragment $P_0=P[i_0\gamma\dd(i_0+1)\gamma)$.

If $\per(P_0)>\tau/3$, by the density condition of synchronizers, $A_1$ contains an element in $[j+i_0\gamma \dd j+i_0\gamma+\tau)$. Let $a \in A_1$ be the minimum such element. We have $\gamma \ge 3\tau-1$, so $T[a \dd a+2\tau)$ is a synchronizing fragment that occurs in $P$ at position $b=a-j$. By the consistency condition of synchronizers, no synchronizing fragment of $T$ occurs in $P_0$ at a position smaller than $b-i_0\gamma$.
Thus, $b \in B_1$. Then $j=a-b \in A_1 \ominus B_1$, as required.

Now assume that $p:= \per(P_0)\le \tau/3$. Then $T[j+i_0\gamma \dd j+(i_0+1)\gamma)$ extends uniquely to a run $T[x \dd y]$ with period $p$. As $\gamma \ge 3\tau-1$, this run is a $\tau$-run. Let $A'_2=\Mis_{k+1}(T,x,x+p)=\Mis_{k+1}(T,j+i_0\gamma,j+i_0\gamma+p)$ and $B'_2=\Mis_{k+1}(P,i_0\gamma,i_0\gamma+p)$. Assume first that $(A'_2 \ominus j) \cap B'_2 \ne \emptyset$ and let $q$ be an element of this set. By definition, $A'_2 \subseteq A_2$ and $B'_2 \subseteq B_2$. Then $j=(q+j)-q \in A_2 \ominus B_2$, as desired.

Finally, assume that $(A'_2 \ominus j) \cap B'_2 = \emptyset$. We note that $A'_2 \ominus j\supseteq \Mis_{k+1}(T',i_0\gamma,i_0\gamma+p)$ for $T'=T[j \dd j+m)$ and recall that $B'_2=\Mis_{k+1}(P,i_0\gamma,i_0\gamma+p)$. By \cref{lemma:Bringmann}, $\Mis_{k+1}(T',i_0\gamma,i_0\gamma+p)=\Mis(T',i_0\gamma,i_0\gamma+p)$, $B'_2=\Mis(P,i_0\gamma,i_0\gamma+p)$, and $\delta_H(P,T')=|\Mis(T',i_0\gamma,i_0\gamma+p)| + |\Mis(P,i_0\gamma,i_0\gamma+p)|$. Hence, $T'$ is a $k$-nearly periodic occurrence of $P$.
\end{proof}

We deal with nearly periodic occurrences using properties of runs in the text.
A proof of the following lemma is deferred until the end of the section.

\begin{restatable}{lemma}{nearper}\label{lem:nearly_per}
For a text $T$ of length $n$ over an integer alphabet and $k=\Oh(1)$, there is an index of size $\Oh(n)$ that, given a pattern $P$ of length $m$, reports a set of $k$-mismatch occurrences of $P$ in $T$ that contains all $k$-nearly periodic occurrences in $\Oh(\log \log n + m + \occ)$ time where $\occ$ is the number of $k$-mismatch occurrences of $P$.
\end{restatable}

In the 2D orthogonal range reporting queries problem, we are to construct a data structure over a specified set of points in a 2D grid that supports the following queries: given a rectangle, list all points that belong to the rectangle. We use the following data structure.

\begin{theorem}[Alstrup et al.~{\cite[Theorem 2]{DBLP:conf/focs/AlstrupBR00}}]\label{thm:RR}
For a set of $n$ points in an $n \times n$ grid, there exists a data structure of size $\Oh(n \log^{\eps} n)$, for any constant $\epsilon>0$, that answers orthogonal range reporting queries in $\Oh(t + \log \log n)$ time, where $t$ is the number of reported points.
\end{theorem}

Let us restate the main result of this section for convenience.
\llong*
\begin{proof}
Nearly periodic $k$-mismatch occurrences are handled by the data structure of \cref{lem:nearly_per}. The query complexity of that data structure is dominated by the final query complexity. The remaining $k$-mismatch occurrences are handled using the anchors from \cref{lem:anchors}.

\emph{Data structure.}
The index uses $k$-errata trees $\TT_t$ for $t \in \{1,2\}$ for the sets of suffixes $\{T[i \dd n)\,:\,i \in A_t\}$ of $T$ and $\TT'_t$ for $t \in \{1,2\}$ for the sets of suffixes $\{(T[0 \dd i))^R\,:\,i \in A_t\}$ of $T^R$.

Let us recall that for $k_1 \in [0 \dd k)$, a $k_1$-errata tree is formed by compact tries of the $k$-errata tree at levels in $[0 \dd k_1]$. Let $\TT_{t,k_1}$ be the $k_1$-errata tree formed from $\TT_t$. For every $t \in \{1,2\}$ and $k_1 \in [0 \dd k]$, let $S_{t,k_1}$ be a concatenation of labels of terminals in all compact tries in $\TT_{t,k_1}$. In each compact trie, the terminals are concatenated from left to right. The order of compact tries can be arbitrary. Each explicit node of each compact trie in $\TT_{t,k}$ stores the fragment of $S_{t,k_1}$ that corresponds to the terminal labels in its subtree. Similarly, we let $\TT'_{t,k_2}$ for $k_2 \in [0 \dd k]$ be the $k_2$-errata tree formed from $\TT'_t$ and $S'_{t,k_2}$ be a concatenation of labels of terminals in all compact tries in $\TT'_{t,k_2}$.

For every $t \in \{1,2\}$ and $k_1 \in [0 \dd k]$, we construct a 2D orthogonal range reporting data structure for the set of 2D points $\{(x,y)\,:\,x \in [0 \dd |S_{t,k_1}|),\, y \in [0 \dd |S'_{t,k_2}|),\,S_{t,k_1}[x]=n-S'_{t,k_2}[y]\}$ where $k_2=k-k_1$.

\emph{Space complexity.} By \cref{lem:A1A2}, the sizes of the sets of anchors $A_1,A_2$ are $\Oh(n/\gamma)$. Hence, by \cref{thm:kerrata}, the sizes of the $k$-errata trees are $\Oh(n+n \log^k n/\gamma)$.

By \cref{lem:kerrata_copies}, for each $k_1 \in [0 \dd k]$, string $S_{t,k_1}$ contains $\Oh(\log^{k_1} n)$ copies of each of the $\Oh(n/\gamma)$ characters. Similarly, string $S'_{t,k_2}$ for $k_2=k-k_1$ contains $\Oh(\log^{k_2} n)$ copies of every character. Over all $k_1$, the number of points created in the 2D orthogonal range reporting data structure is 
$\Oh(n\log^k n/\gamma)$.
By \cref{thm:RR}, this gives $\Oh(n\log^{k+\eps} n/\gamma)$ space of the range reporting data structures.

\emph{Queries.} Given a pattern $P$ of length $m \ge (k+1)\gamma$, we construct its sets of anchors $B_1$ and $B_2$ using \cref{lem:B1B2}. For each $t \in \{1,2\}$, each $k_1 \in [0 \dd k]$ and each anchor $b \in B_t$, we ask a query for the suffix $P[b \dd m)$ to the $k_1$-errata tree $\TT_{t,k_1}$ and a query for the suffix $(P[0 \dd b))^R$ of $P^R$ to the $k_2$-errata tree $\TT'_{t,k_2}$ where $k_2=k-k_1$. In each query, we do \emph{not} report all occurrences. Instead, each query results in a set of nodes in $\TT_{t,k_1}$ which translate to fragments of $S_{t,k_1}$ and a set of nodes in $\TT'_{t,k_2}$ which translate to fragments of $S'_{t,k_2}$. For every fragment $S_{t,k_1}[x_1 \dd x_2]$ and every fragment $S'_{t,k_2}[y_1 \dd y_2]$, we ask an orthogonal range reporting query for the rectangle $[x_1 \dd x_2] \times [y_1 \dd y_2]$ in the data structure constructed for $t$ and $k_1$. For every reported point $(x,y)$, we report a $k$-mismatch occurrence of $P$ in $T$ at position $S_{t,k_1}[x]-b$.

\emph{Query complexity.}
First, we analyze all steps excluding reporting of occurrences. The sets of anchors $B_1$, $B_2$ are constructed in $\Oh(m)$ time. For a given anchor, we ask $\Oh(k)$ queries to $(\le k)$-errata trees. After $\Oh(m)$ time preprocessing, a query for $k_1$ returns in $\Oh(\log^{k_1} n\log \log n)$ time $\Oh(\log^{k_1} n)$ nodes in $\TT_{t,k_1}$ and a query for $k_2$ returns in $\Oh(\log^{k_2} n\log \log n)$ time $\Oh(\log^{k_2} n)$ nodes in $\TT'_{t,k_2}$. Overall, in $\Oh(\log^{k} n\log \log n)$ time we obtain $\Oh(\log^k n)$ queries to a 2D orthogonal range reporting data structure, for a total of $\Oh(\log^k n\log \log n)$ time per anchor. The time complexity follows by \cref{thm:RR}.

Finally, let us count how many times a given $k$-mismatch occurrence of $P$ in $T$ can be reported. For every $k_1 \in [0 \dd k]$, $k_2=k-k_1$, $t \in \{1,2\}$, and $b \in B_t$, each $k_1$-mismatch occurrence of $P[b \dd m)$ in $T$ at a position in $A_t$ is considered in exactly one interval and each $k_2$-mismatch occurrence of $P[0 \dd b)$ in $T$ ending at a position in $A_t \ominus 1$ is also considered exactly once. Hence, for fixed $k_1$, $t$, and $b$, each resulting $k$-mismatch occurrence of $P$ will be reported at most once. Over all $k_1$ and $b$, each occurrence is reported $\Oh(k^3)=\Oh(1)$ times, by \cref{lem:B1B2}.
\end{proof}

We say that a string $U$ is \emph{primitive} if $U=V^t$ for string $V$ and positive integer $t$ implies that $t=1$.
Strings $U$ and $V$ are called \emph{cyclic shifts} if there exist strings $X$, $Y$ such that $U=XY$ and $V=YX$.
A \emph{Lyndon string} is a string $U$ that is primitive and is lexicographically minimal among its cyclic shifts.
All cyclic shifts of a primitive string are primitive and different~\cite{DBLP:books/daglib/0019130}.
Hence, every primitive string $U$ has a unique cyclic shift $V$ (the minimal one) that is a Lyndon string.
We extend this notation to runs as follows: a Lyndon root of a periodic string $U$ with period $p$ is the minimum cyclic shift of the string period $U[0 \dd p)$ of $U$
(see \cite{DBLP:journals/tcs/CrochemoreIKRRW14}).

The problem of interval stabbing is defined as follows: Preprocess a set of $n$ intervals so that, given a query point $a$, all intervals that contain the point $a$ can be reported efficiently.
Interval stabbing queries enjoy a known reduction to three-sided (actually, even two-sided) 2D orthogonal range reporting queries. An interval $[\ell \dd r]$ is treated as a point $(\ell,r)$ in 2D and a stabbing query for point $a$ requires to compute all points in a rectangle $(-\infty \dd a] \times [a \dd \infty)$. Using the data structure for three-sided 2D orthogonal range reporting of Alstrup et al.~\cite{DBLP:conf/focs/AlstrupBR00}, we obtain the following data structure for interval stabbing.

\begin{theorem}[see {\cite[Corollary 1]{DBLP:conf/focs/AlstrupBR00}}]\label{thm:Stab}
A set of $n$ intervals with endpoints in $[0 \dd N]$ can be preprocessed into a data structure of size $\Oh(n)$ that answers 
interval stabbing queries in $\Oh(t + \log \log N)$ time, where $t$ is the number of reported intervals.
\end{theorem}

We are ready to prove \cref{lem:nearly_per}.

\begin{proof}[Proof of \cref{lem:nearly_per}]
\emph{Data structure.}
The data structure contains a trie $\M$ of all strings being a Lyndon root of a $\tau$-run in $T$.
Each terminal node of $\M$, representing a substring of $T$ that is a Lyndon root $U$ of some $\tau$-runs in $T$, holds $(k+1)|U|$ interval stabbing data structures indexed
by pairs in $[0 \dd k] \times [0 \dd |U|)$ that we now describe.

Let $T[x \dd y]$ be a $\tau$-run in $T$ with period $p$ and $Q=T[d \dd d+p)$ for $d \in [x \dd x+p)$ be its Lyndon root.
We consider the sets $L=\LeftMis_{k+1}(T,x,x+p)$ and $R=\RightMis_{k+1}(T,x,x+p)$.
If $|L| \le k$, we insert a sentinel position $-1$ to $L$.
Similarly, if $|R| \le k$, we insert a sentinel position $n$ to $R$.
Let $L=\{\ell_1,\ell_2,\ldots,\ell_{k_1}\}$ with $\ell_{k_1}< \ell_{k_1-1} < \cdots < \ell_1$ ($k_1 \le k+1$).
Let $R=\{r_1,r_2,\ldots,r_{k_2}\}$ with $r_1<r_2< \cdots < r_{k_2}$ ($k_2 \le k+1$).
We set $r_0=\ell_1$.

Let us consider each $a \in [1 \dd k_1)$ and $b \in [0 \dd k_2)$ such that $a+b \le k$ and each remainder $t \in [0 \dd p)$.
Let $\alpha$ and $\beta$ be the minimum and maximum element of the set $\{j \in (\ell_{a+1} \dd \ell_a]\,:\,d-j \equiv t \pmod{p}\}$, if the set is non-empty.
We then insert an interval $I=[r_b-\beta+1 \dd r_{b+1}-\alpha]$ to the data structure corresponding to the Lyndon root $Q$ indexed by the pair $(a+b,t)$.
Intuitively, the interval $I$ has the following interpretation: for every $m$ in $I$, there exists an index $j \in [\alpha \dd \beta]$ such that
$j \equiv \alpha \pmod{p}$ and $j+m-1 \in [r_b \dd r_{b+1})$.
Then for $T'=T[j \dd j+m)$, $|\Mis(T',d-j,d-j+p)|=a+b$, where $T'[d-j \dd d-j+p)=Q$, and $d-j \equiv t \pmod{p}$.
The interval $I$ in the data structure stores additional data: numbers $\alpha$, $\beta$ as well as the interval $[\alpha' \dd \beta'] = [r_b \dd r_{b+1})$.
Formally, the additional data can be stored using perfect hashing~\cite{DBLP:journals/jacm/FredmanKS84}.
With the additional data, we will be able to recover all values $j \in [\alpha \dd \beta]$ such that $j \equiv \alpha \pmod{p}$ and $j+m-1 \in [\alpha' \dd \beta']$ using simple arithmetics.

The case in which a pattern $P$ can match a length-$m$ periodic substring of $T$ is handled similarly.
For each remainder $t \in [0 \dd p)$, we compute the minimum $\alpha$ and the maximum $\beta$ of the set $\{j \in (\ell_1 \dd r_1)\,:\,d-j \equiv t \pmod{p}\}$ and insert
an interval $[0 \dd r_1-\alpha]$ to the data structure corresponding to the Lyndon root $Q$ of the run $T[x \dd y]$ indexed by the pair $(0,d)$.
The additional data stored with the interval are $\alpha$, $\beta$, and $\beta'=r_1$.

We also store the data structure of \cref{lem:rootedLCP_special} for $T$.

\emph{Data structure size.}
By \cref{lem:tauruns}, the total number of $\tau$-runs in $T$ is $\Oh(n/\tau)$.
The length of each Lyndon root of a $\tau$-run is at most $\tau/3$, so the total size of the trie $\M$ is $\Oh(n)$, even without compactification.
There are $\Oh((n/\tau) \cdot k \cdot \tau)=\Oh(n)$ interval stabbing data structures.
For each $\tau$-run, we consider $\Oh(k^2)$ pairs of indices $(a,b)$ and $\Oh(\tau)$ remainders $t$ and for each $a,b,t$ insert one interval into a data structure.
The interval stabbing data structures contain $\Oh(k^2 n)$ intervals in total and thus, by \cref{thm:Stab}, can be represented in $\Oh(n)$ total space. The data structure of \cref{lem:rootedLCP_special} takes just $\Oh(n)$ space.

\emph{Queries.}
Let $P$ be a query pattern of length $m \ge (k+1)\gamma$.
For every $i \in [0 \dd k]$, we compute the smallest period $p$ of $P[i\gamma \dd (i+1)\gamma)$.
If $p \le \tau/3$ and $P[i\gamma \dd (i+1)\gamma)$ is a substring of $T$, which we check using \cref{lem:rootedLCP_special}, we proceed as follows.
We compute the set $X=\Mis_{k+1}(P,i\gamma,i\gamma+p)$ by definition.
If $|X| \le k$, we proceed as follows.
The Lyndon root $Q=P[d \dd d+p)$, with $d \in [i\gamma \dd i\gamma+p)$, of $P[i\gamma \dd i\gamma+p)$ is computed using an algorithm for
a minimal cyclic shift~\cite{DBLP:journals/ipl/Booth80,DBLP:journals/jal/Duval83}.
We then check if $Q$ is a terminal in $\M$ using \cref{thm:WExp}.
If so, we query the interval stabbing data structures for $Q$ indexed by pairs $(h,d \bmod p)$ for all $h\in[0\dd k-|X|]$ for the length $m$.
Assume that $h>0$ and the data structure returned an interval $I$ ($m \in I$) that stores, as additional data, numbers $\alpha$, $\beta$ as well as the interval $[\alpha' \dd \beta']$.
We then report a $k$-mismatch occurrence of $P$ in $T$ at each index $j \in [\alpha \dd \beta]$ such that $j \equiv \alpha \pmod{p}$
and $m \in [\alpha'+1-j \dd \beta'+1-j]$; the final condition can be restated equivalently as
$j \in [\alpha'+1-m \dd \beta'+1-m]$.
There will be at least one such index $j$ by the fact that the interval $I$ was stabbed by $m$.
For $h=0$, if the data structure returned an interval with additional data $\alpha$, $\beta$, and $\beta'$, we report all $j \in [\alpha \dd \beta]$ such that $j \equiv \alpha \pmod{p}$
and $j+m-1 \le \beta'$.
They will be initial elements of the sequence $\alpha,\alpha+p,\ldots,\beta$.

\emph{Query correctness.}
Let us argue for the correctness of the query algorithm.
First, we show that if position $j$ is reported in the algorithm, then $\delta_H(T[j \dd j+m),P) \le k$, i.e.\ $j$ is the starting position of a $k$-mismatch occurrence of $P$ in $T$.
Assume that $j$ is reported for substring $P[i\gamma \dd (i+1)\gamma)$ for $i \in [0 \dd k]$, with smallest period $p$ and Lyndon root $Q=T[d \dd d+p)$,
when querying the data structure indexed by pair $(h,d \bmod p)$ for $h \in [0 \dd k-|X|]$, where $X=\Mis_{k+1}(P,i\gamma,i\gamma+p)$, $|X| \le k$.
Let $W = Q^\infty[c\dd c+m)$ for $c=(-d) \bmod p$.
We have $\delta_H(W,P) = |X|$ and $\delta_H(W,T[j \dd j+m)) = h \le k-|X|$.
By the triangle inequality, $\delta_H(P,T[j \dd j+m)) \le \delta_H(W,P)+\delta_H(W,T[j \dd j+m))\le k$, as required.

Now we show that if $T'=T[j \dd j+m)$ is a $k$-nearly periodic occurrence of $P$ in $T$, then $j$ is reported in the algorithm.
By definition, there is $i \in [0 \dd k]$ such that $P[i\gamma \dd (i+1)\gamma)$ has smallest period $p \le \tau/3$.
Let $d \in [i\gamma \dd i\gamma+p)$ be such that $Q=P[d \dd d+p)$ is the Lyndon root of $P[i\gamma \dd i\gamma+p)$.
By definition, $P[i\gamma \dd (i+1)\gamma) = T'[i\gamma \dd (i+1)\gamma)$.
The fragment $T[j+i\gamma \dd j+(i+1)\gamma)$ extends to a run $T[x \dd y]$ in $T$ with period $p$; it is a $\tau$-run as $\gamma \ge 3\tau-1$ and $p \le \tau/3$.
Let $X = \Mis(P,i\gamma,i\gamma+p)$ and $Y = \Mis(T',i\gamma,i\gamma+p)$.
By definition, $|X|+|Y| \le k$.
We have $(Y \oplus j) \subseteq (L \cup R)$ where $L=\LeftMis_{k+1}(T,x,x+p)$ and $R=\RightMis_{k+1}(T,x,x+p)$.
Then $j$ is reported for $a=|L \cap (Y \oplus j)|$, $b=|R \cap (Y \oplus j)|$, and the remainder $t=d \bmod{p}$.

\emph{Query complexity.}
The smallest period of the fragment $P[i\gamma \dd (i+1)\gamma)$ is computed in $\Oh(\gamma)$ time~\cite{DBLP:journals/siamcomp/KnuthMP77}, for a total of $\Oh(m)$ time over $i \in [0 \dd k]$.
We check if $P[i\gamma \dd (i+1)\gamma)$ is a substring of $T$ in $\Oh(1)$ time after $\Oh(m)$ preprocessing of \cref{lem:rootedLCP_special}.
The Lyndon roots for all the fragments are also computed in $\Oh(m)$ total time using a linear-time algorithm~\cite{DBLP:journals/ipl/Booth80,DBLP:journals/jal/Duval83}.
All the sets $\Mis_{k+1}(P,i\gamma,i\gamma+p)$ for $i \in [0 \dd k]$ such that $P[i\gamma \dd (i+1)\gamma)$ has period $p \le \tau/3$ are computed in $\Oh(mk)$ time by definition.
The trie $\M$ is descended for $U$ in $\Oh(\tau)$ time (\cref{thm:WExp}), for a total of $\Oh(m)$ time.
Then, for a given $i$ we ask $\Oh(k)$ interval stabbing queries, which requires $\Oh(k^2 \log \log n)$ time in addition to the time for reporting occurrences.

Finally, let us show that each $k$-mismatch occurrence $T[j \dd j+m)$ of $P$ is reported at most $\Oh(k^2)$ times by the data structure.
First, it can be reported for each fragment $P[i\gamma \dd (i+1)\gamma)$ with $i \in [0 \dd k]$.
This fragment implies exactly one Lyndon root $Q$.
If $T[j \dd j+m)$ has period $p$, it is reported for exactly one $\tau$-run in $T$ with Lyndon root $Q$, namely, for the $\tau$-run that extends $T[j \dd j+m)$.
In total, if $T[j \dd j+m)$ has period $p$, it is reported $\Oh(k)$ times for $i \in [0 \dd k]$.

Henceforth, we assume that $T[j \dd j+m)$ does not have period $p$.
For a given $\tau$-run $T[x \dd y]$ in $T$ with Lyndon root $Q$, at most one interval $I$ is inserted to an interval stabbing data structure such that $m \in I$,
$j \in [\alpha \dd \beta]$, $j \equiv \alpha \pmod{p}$, and $j+m-1 \in [\alpha' \dd \beta']$ where $\alpha$, $\beta$, and $[\alpha' \dd \beta']$ are the additional data stored for the interval $I$.
It suffices to note that $T[j \dd j+m)$ will be reported by at most $k$ $\tau$-runs with Lyndon root $Q$ in $T$.
Indeed, each such $\tau$-run $T[x \dd y]$ must satisfy $x > j$.
Let $T[x_1 \dd y_1],\ldots,T[x_{k'} \dd y_{k'}]$ be all such $\tau$-runs ordered by increasing starting positions.
(Two runs with period $p$ can intersect on at most $p-1$ positions; otherwise they would be the same run.)
By maximality of runs, for any $k''>k$, $\LeftMis(T,x_{k''},x_{k''}+p)$ contains at least $k+1$ elements that are at least $j$.
This means that $T[j \dd j+m)$ will not be reported by $T[x_{k''} \dd y_{k''}]$.

This concludes that the query complexity is $\Oh(m + \log \log n + \occ)$.
\end{proof}

\section{\texorpdfstring{$\Oh(n \log^{k-1} n)$}{O(n log\^(k-1) n)}-space \texorpdfstring{$k$}{k}-Approximate Index}\label{sec:general}
We improve upon the approach of Chan et al.~\cite{DBLP:journals/algorithmica/ChanLSTW10} to make it work for a text over any alphabet within the same complexity.
First we show a simpler data structure for (exact) indexing with up to $k$ wildcards in the pattern, as it highlights our approach better. Then we show how it can be adjusted to $k$-Mismatch Indexing.

\subsection{Index for Wildcards in the Pattern}\label{sec:gen_wild}
A wildcard is a don't care symbol that is not in the alphabet of the text or the pattern. A wildcard matches every character of the alphabet; the relation of matching extends naturally to strings with wildcards. We improve the space complexity of an index for pattern matching queries for patterns with up to $k$ wildcards in a string text as shown in \cref{tab:wild_patt}.

\begin{table}[htpb]
\centering
\begin{tabular}{c|c|c|c}
\multirow{2}{*}{\textbf{Space}} & \textbf{Query time} & \multirow{2}{*}{\textbf{Reference}} & \multirow{2}{*}{\textbf{Comment}} \\
& (plus $\Oh(m+\occ)$) & & \\\hline
$\Oh(n \sigma^{k^2} (\log \log n)^k)$ & $\Oh(k)$ & \cite{DBLP:journals/mst/BilleGVV14} &\\\hline
$\Oh(n^{k+1})$ & $\Oh(k)$ & \cite{DBLP:journals/mst/BilleGVV14} &\\\hline
\textcolor{gray}{$\Oh(n \log^k n)$} & \multirow{2}{*}{$\Oh(2^k \log \log n)$} & \cite{DBLP:conf/stoc/ColeGL04} & improved here \\\cline{1-1}\cline{3-4}
$\Oh(n \log^{k-1} n)$ & & \textbf{Thm.~\ref{thm:kwildP}} & \\\hline
$\Oh(n \log n (\log_\alpha n)^{k-1})$ & $\Oh(\alpha^k\log \log n)$ & \cite{DBLP:journals/mst/BilleGVV14} & \\\hline
\textcolor{gray}{$\Oh(n \log n)$} & \textcolor{gray}{$\Oh(\sigma^k \log \log n)$} & \cite{DBLP:conf/stoc/ColeGL04} & improved by \cite{DBLP:journals/mst/BilleGVV14} \\\hline
$\Oh(n)$ & $\Oh(\sigma^k\min\{m,\log \log n\})$ & \cite{DBLP:journals/mst/BilleGVV14} &
\end{tabular}
\caption{Indexing for patterns with up to $k$ wildcards}\label{tab:wild_patt}
\end{table}

\subsubsection{Original \texorpdfstring{$k$}{k}-errata tree for wildcards in the pattern}
Cole et al.~\cite{DBLP:conf/stoc/ColeGL04} proposed the following version of $k$-errata tree for indexing patterns with wildcards.

\begin{theorem}[\cite{DBLP:conf/stoc/ColeGL04}]\label{thm:kerrata_wild_patt}
Let $T$ be a string text of length $n$. The $k$-errata tree for $T$ uses space $\Oh(n \log^k n)$ and, for a pattern $P$ of length $m$ with up to $k$ wildcards, answers a pattern matching query for $P$ in $\Oh(2^k \log\log n + m + \occ)$ time, where $\occ$ is the number of occurrences of $P$ in $T$.
\end{theorem}

The recursive structure of the $k$-errata tree in this case is very similar as for mismatches; only the details on how compact tries at subsequent levels are formed are different (actually, simpler). More precisely, all substitution trees of the type \ref{a}, i.e., off-path subtrees $\T^s_{v,a}$, are organized into a single group tree, whereas substitution trees and group trees of type \ref{b} are not formed. Thus, each group tree can be viewed as a sparse suffix tree of $T$ (technically, with a wildcard-character edge as in item \ref{it1} in \cref{sec:errata}).

The query algorithm also has the same recursive structure, but it is simpler, as here the positions of wildcards are known, whereas the positions of mismatches were not known. Here, if the query algorithm is called recursively on a compact trie $\T$ with a suffix $P'$ of the pattern $P$, then the level of the trie plus the number of wildcards in $P'$ does not exceed $k$. The query time is $\Oh(2^k \log \log n)$.

We show how to improve the space complexity from $\Oh(n \log^k n)$ to $\Oh(n \log^{k-1} n)$ retaining the query complexity. We note that this variant was not considered by Chan et al.~\cite{DBLP:journals/algorithmica/ChanLSTW10} as their data structure does not give a significant improvement here. Particularly, it would give space $\Oh(n \log^{k-1} n)$ but make the query time $\Oh(\sigma \cdot 2^k \log \log n)$.

\subsubsection{Storing suffixes at level \texorpdfstring{$k$}{k} compactly}\label{subsec:k_wild}
We save space by explicitly constructing only the $(k-1)$-errata tree. The final step of the recursion is replaced by storing arrays equivalent to group trees at level $k$ compactly.

Let $\T$ be a compact trie at level $k-1$ of the errata tree. For a node $u$ of $\T$, let $\T_{u}$ be the subtree of $\T$ rooted at $u$. By $|\T_u|$ we denote the number of terminals in $\T_u$ (these are the terminals from $\T$ that are located in $\T_u$). We use the heavy-light decomposition of $\T$. For an explicit node $v$ of $\T$, let $u_1,\ldots,u_t$ be all explicit children of $v$ that are not located on the same heavy path as $v$. Let us create a list of path labels of all terminals of the tree $\T_v$ that are located in the subtrees $\T_{u_1},\ldots,\T_{u_t}$, trim their first characters (that immediately follow $v$), and order the resulting strings lexicographically. We denote the list of resulting suffixes of $T$ as $S'_{v,1},\ldots,S'_{v,b_v}$. Furthermore, by $S_{v,1},\ldots,S_{v,b_v}$ we denote the list of these suffixes without their first character removed, but \emph{in the order} of their truncated versions $S'_{v,1},\ldots,S'_{v,b_v}$.

\newcommand{\llabel}{\mathit{label}}
Over all compact tries at level $k-1$ of the errata tree, the lists of trimmed suffixes have total length $\Oh(n \log^k n)$ (as each terminal in $\T$ generates suffixes in $\Oh(\log n)$ nodes $v$, by the properties of heavy paths). Instead of storing them explicitly, we will store the indices
$$\llabel'_v[i]:= n-1-|S'_{v,i}|$$
of the trimmed suffixes using just $\Oh(n \log^k n)$ bits of space in total, i.e., $\Oh(n \log^{k-1} n)$ words of space. This will be done using similar techniques as in Chan et al.~\cite{DBLP:journals/algorithmica/ChanLSTW10}.

\newcommand{\strank}{\mathit{st\mbox{-}rank}}
\newcommand{\treepointer}{\mathit{tree\mbox{-}pointer}}
\newcommand{\modifiedrank}{\mathit{modified\mbox{-}rank}}

Let $v$ be an explicit node of a compact trie $\T$ at level $k-1$ and $u_1,\ldots,u_t$ be all the explicit children of $v$ not on the same heavy path. We assign subsequent \emph{ranks} $1,2,\ldots$ to the subtrees $\T_{u_i}$ so that trees with more terminals receive a smaller rank, breaking ties arbitrarily. We store the following arrays:
\begin{itemize}
\item $\strank'_v[1 \dd b_v]$: If $S_{v,i}$ is a terminal of $\T_v$ located in $\T_{u_j}$, $\strank'_v[i] \in [1 \dd t]$ is the rank of $\T_{u_j}$.
\item $\treepointer'_v[1 \dd t]$: $\treepointer'_v[j]$ points to the subtree $\T_{u_i}$ with rank $j$.
\item $\modifiedrank'_{v,u_a}[1 \dd |\T_{u_a}|]$: $\modifiedrank'_{v,u_a}[j]=i \in [1 \dd b_v]$ if the $j$th lexicographic terminal in $\T_{u_a}$ generates $S'_{v,i}$.
\end{itemize}

\begin{remark}
In \cite{DBLP:journals/algorithmica/ChanLSTW10}, similar arrays $\strank$, $\treepointer$, $\modifiedrank$ were used but in a different context; see \cref{sec:compact_kMis}.
\end{remark}

The proof of next lemma resembles the proof of \cite[Lemma 2]{DBLP:journals/algorithmica/ChanLSTW10}.

\begin{lemma}\label{lem:strank}
The sum of values $\lceil{\log \strank'_v[i]}\rceil$, over all explicit nodes $v$ in compact tries $\T$ at level $k-1$ and $i \in [1 \dd b_{v}]$, is $\Oh(n \log^k n)$.
\end{lemma}
\begin{proof}
Let $\ell$ be a terminal in a compact trie $\T$ at level $k-1$. We will sum up values $\ceil{\strank'_v[i]}$ for $v$, $i$ such that $S_{v,i}$ corresponds to $\ell$.

Let $C_1,C_2,\ldots,C_\alpha$ be all heavy paths, top-down, visited on the root-to-$\ell$ path $\pi$. By $\T_{C_i}$ we denote the subtree of $\T$ rooted at the root of path $C_i$.
For $j \in [1 \dd \alpha)$, let $v$ be a node in heavy path $C_j$ where the path $\pi$ leaves the path $C_j$, $u_a$ be the next explicit node on $\pi$, and let $r_j$ be the rank of $\T_{u_a}$. By the definition of rank,
\[r_j\ \le\ \tfrac{|\T_v|}{|\T_{u_a}|}\ \le\ \tfrac{|\T_{C_j}|}{|\T_{u_a}|}\ =\ \tfrac{|\T_{C_j}|}{|\T_{C_{j+1}}|}.\]
We thus have
\[\sum_{j=1}^{\alpha-1} \log r_j\ \le\ \sum_{j=1} ^{\alpha-1} \log\tfrac{|\T_{C_j}|}{|\T_{C_{j+1}}|} \ \le\ \log|\T_{C_1}|\ \le\ \log n.\]

We notice that for every node $v$ and index $i \in [1 \dd b_v]$, $\strank'_v[i]$ uniquely corresponds to $r_j$ for $\ell$ being the locus of $L_vS_{v,i}$ where $L_v$ is the path label of $v$ and $j$ such that the root-to-$\ell$ path leaves the $j$th heavy path $C_j$ at node $v$.
By \cref{thm:kerrata_wild_patt}, there are $\Oh(n \log^{k-1} n)$ terminals in compact tries at level $k-1$, so the sum of all values $\log\strank'_v[i]$ is $\Oh(n \log^k n)$. The sum of values $\ceil{\log\strank'_v[i]}$ is also $\Oh(n \log^k n)$ as the number of these values is $\Oh(n \log^k n)$.
\end{proof}

By \cref{lem:strank}, the $\strank'$ arrays can be represented in $\Oh(n \log^k n)$ bits using variable size encoding that allows $\Oh(1)$-time access~\cite{DBLP:conf/fsttcs/Munro96}. The $\treepointer'$ arrays contain $\Oh(n \log^{k-1} n)$ values in total, one per each compact edge of a level-$(k-1)$ compact trie. Finally, $\modifiedrank'$ arrays can be stored efficiently according to the next lemma that resembles \cite[Lemma 11]{DBLP:journals/algorithmica/ChanLSTW10}.

\begin{lemma}\label{lem:modifiedrank}
All $\modifiedrank'$ arrays can be stored in $\Oh(n \log^k n)$ bits of space so that a (restricted) rank query on an array $\modifiedrank'_{v,u_a}$ can be answered in $\Oh(1)$ time.
\end{lemma}
\begin{proof}
Let $v$ be an explicit node in a compact trie of level $k-1$ and $u_1,\ldots,u_t$ be its explicit children located on a different heavy path. Recall that in a rank query, given array $\modifiedrank'_{v,u_a}$ and a value $i$, we are to retrieve the index $j$ such that $\modifiedrank'_{u_a}[j]=i$ or state that $j$ does not exist.
The sequence $\modifiedrank'_{v,u_a}$ is increasing, has length $|\T_{u_a}|$ and contains elements in $[1 \dd |\T_v|]$, so by \cref{thm:Raman}, it can be represented in $\Oh\left(|T_v| \ceil{\log\tfrac{|\T_v|}{|\T_{u_a}|}}\right)$ bits of space to allow rank queries. We have $\Oh\left(|T_v| \ceil{\log\tfrac{|\T_v|}{|\T_{u_a}|}}\right)=\Oh\left(|T_v| \log\tfrac{|\T_v|}{|\T_{u_a}|}\right)$ as $|\T_{u_a}| \le \tfrac12 |\T_v|$ (the child $u_a$ is light).

Let $f(\T_v)$ be the total space in bits required to store the $\modifiedrank'$ arrays for all valid pairs of nodes in the subtree $\T_v$. We have
\[f(\T_v) \le \sum_{i=1}^t \Oh\left(|\T_{u_i}| \log\tfrac{|\T_v|}{|\T_{u_i}|}\right)+f(\T_{u_i})\]
which solves to $f(\T_v)=\Oh(|\T_v| \log |\T_v|)$ for any node $v$. By taking as nodes $v$ the roots of all compact tries at level $k-1$, we obtain the desired bound.
\end{proof}

Finally, an array called a sparse suffix array storing a lexicographic list of all terminals of $\T$ is stored, with each terminal represented as the position in $T$ of the corresponding suffix. Each explicit node $u$ of $\T$ stores a fragment of the list that corresponds to the terminals of $\T$ that are located in $\T_{u}$. The total size of the sparse suffix arrays stored is $\Oh(n \log^{k-1} n)$.

In the next lemma, we show that the arrays are sufficient to store labels of the trimmed suffixes.

\begin{lemma}\label{lem:label_wild}
Let $v$ be an explicit node in a compact trie $\T$ at level $k-1$ of the errata tree. Assuming random access to the sparse suffix array of $\T$, arrays $\strank'_v$ and $\treepointer'_v$, and $\Oh(1)$-time rank queries on all arrays $\modifiedrank'_{v,u_a}$, given $i$, $\llabel'_v[i]$ can be computed in $\Oh(1)$ time. 
\end{lemma}
\begin{proof}
The original suffix $S_{v,i}$ that generates the trimmed suffix $S'_{v,i}$ originates from the trie $\T_{u_a}=\treepointer'_v[\strank'_v[i]]$ that can be retrieved in $\Oh(1)$ time. We compute an index $j$ such that $\modifiedrank'_{v,u_a}[j]=i$ using a rank query in $\Oh(1)$ time. The suffix represented by the $j$th terminal in $\T_{u_a}$ can be extracted from the $j$th element in the fragment of the sparse suffix array for $\T$ that corresponds to $\T_{u_a}$. From the suffix we retrieve $S'_{v,i}$ which implies $\llabel'_v[i]$.
\end{proof}

\subsubsection{Compact \texorpdfstring{$k$}{k}-errata tree for wildcards in the pattern}
The key improvements that allow to achieve a more compact data structure are the developments of \cref{subsec:k_wild} as well as sampling suffixes among $S'_{v,1},\ldots,S'_{v,b_v}$ according to \cref{lem:utility}.

\wildP*
\begin{proof}
\emph{Data structure.}
As mentioned before, we store the $(k-1)$-errata tree explicitly together with the data structure for answering unrooted $\TreeLCP$ queries in compact tries (\cref{thm:unrootedLCP}). By \cref{thm:kerrata_wild_patt}, this takes $\Oh(n \log^{k-1} n)$ space. For all explicit nodes $v$ of compact tries at level $k-1$, we store (1) the array $\strank'_v$ using variable size encoding~\cite{DBLP:conf/fsttcs/Munro96}, which take $\Oh(n \log^{k-1} n)$ total space by \cref{lem:strank}; (2) the $\treepointer'_v$ array in $\Oh(n \log^{k-1} n)$ total space; and (3) the $\modifiedrank'_{v,u_a}$ arrays for all $u_a$ with rank queries using \cref{lem:modifiedrank} also in $\Oh(n \log^{k-1} n)$ total space. We also store the sparse suffix array of each compact trie $\T$ at level $k-1$.

For each node $v$ of each compact trie $\T$ at level $k-1$ of the data structure, let $\T'_v$ denote the compact trie of trimmed suffixes $S'_{v,1},\ldots,S'_{v,b_v}$. We select every $\floor{\log n}$th trimmed suffix among $S'_{v,1},\ldots,S'_{v,b_v}$, as well as $S'_{v,1}$ and $S'_{v,b_b}$, and construct a compact trie $\T''_v$ of these trimmed suffixes. Each terminal $S'_{v,i}$ in $\T''_v$ stores $i$ as its number. The total size of the compact tries $\T''_v$ is $\Oh(n \log^{k-1} n)$, whereas trees $\T'_v$ are \emph{not} stored explicitly. We construct a data structure for answering unrooted $\TreeLCP$ queries on each of $\T''_v$ (\cref{thm:unrootedLCP}), which does not affect the space complexity.

We also store the data structure of \cref{lem:rootedLCP_special} together with the data structure for $\LCP$-queries in $T$~\cite{DBLP:conf/latin/BenderF00}; both take just $\Oh(n)$ space.

\emph{Queries.}
We execute a $(k-1)$-errata tree query for $P$ (\cref{thm:kerrata_wild_patt}). As a result, in time $\Oh(2^k \log \log n)$, $\Oh(2^{k-1})$ pairs $(\T,P')$ are computed where $\T$ is a compact trie at level $k-1$ and $P'$ is a suffix of $P$ with at most 1 wildcard; our goal is to match $P'$ against terminals in $\T$. We compute $v=\TreeLCP(\T,P')$ (in $\Oh(\log \log n)$ time). If the whole $P'$ was matched, it did not contain the wildcard symbol and all the terminals in the subtree of $v$ can be reported. If the first unmatched character of $P'$ was not the wildcard, there are no occurrences. Otherwise, we need to select a character on an edge outgoing from $v$ to be matched against the wildcard. First, we consider choosing the symbol on the heavy path of $v$. To this end, we ask an unrooted $\TreeLCP$ query on the node located one symbol below $v$ on the heavy path and the suffix $P''$ of $P$ starting immediately after the last wildcard symbol.

Next, we consider choosing the symbol that is not on the heavy path of $v$. In other words, we need to identify all trimmed suffixes among $S'_{v,1},\ldots,S'_{v,b_v}$ that have $P''$ as a prefix. The result will consist of a connected sublist of this list that is computed in $\Oh(\log \log n)$ time using \cref{lem:utility} for $k=k'=0$. Here we use the fact that for any index $i$, we can identify the suffix $S'_{v,i}$ by computing $\llabel'_v[i]$ in $\Oh(1)$ time by \cref{lem:label_wild}. To compare a suffix $P'$ of $P$ without wildcards with a suffix of $T$ in $\Oh(1)$ time, we use \cref{lem:rootedLCP_special} together with the data structure for $\LCP$-queries in $T$~\cite{DBLP:conf/latin/BenderF00}.\footnote{We could also use \cref{lem:kangaroo_k} but it is not necessary, as there are no modifications in suffixes.}

Overall, the query complexity is $\Oh(2^k \log \log n + m + \occ)$, as desired.
\end{proof}

\subsection{\texorpdfstring{$k$}{k}-Mismatch Indexing}\label{sec:compact_kMis}
We show how our approach from \cref{sec:gen_wild} can be used for $k$-Mismatch Indexing. Let us recall some parts of the approach of Chan et al.~\cite{DBLP:journals/algorithmica/ChanLSTW10}.\footnote{Actually, their approach was stated explicitly only for the case of $k=1$; we make the necessary adaptations for an arbitrary $k$.}

\subsubsection{Storing 1-modified suffixes at level \texorpdfstring{$k$}{k} compactly}
As before, we explicitly construct only the $(k-1)$-errata tree.
Let $\T$ be a compact trie at level $k-1$. For a heavy path $C$ in $\T$, let $\T_C$ be the subtree of the root of $C$. We consider all terminals in $\T_C$; for a terminal representing suffix $S$, we find the node $v$ of $C$ where the path representing $S$ diverges from $C$ (if any) and create a 1-modified suffix by changing the character after $v$ in $S$ to the character that follows $v$ in the heavy path. Let $S'_{C,1},\ldots,S'_{C,b_C}$ be the list of all such 1-modified suffixes, ordered lexicographically, and let $S_{C,1},\ldots,S_{C,b_C}$ be the list of original suffixes but in the order of the 1-modified suffixes. Each terminal $\ell$ of $\T$ implies a 1-modified suffix $S'_{C,i}$ for $\Oh(\log n)$ heavy paths $C$ on the root-to-$\ell$ path, so the total number of modified suffixes $S'_{C,i}$ is $\Oh(n \log^k n)$. Instead of storing modified suffixes $S'_{C,i}$ explicitly, we store in a compact way the indices $\llabel_C[i]:= n-1-|S'_{C,i}|$.

The next lemma is proved by representing in a compact way arrays $\strank$, $\treepointer$, and $\modifiedrank$ defined for heavy paths of each compact trie, as well as a sparse suffix array for each compact trie.

\begin{lemma}[{\cite[Section 3.3 and Lemma 11]{DBLP:journals/algorithmica/ChanLSTW10}}]\label{lem:labelChan}
There is a data structure of size $\Oh(n \log^{k-1} n)$ that, given a heavy path $C$ of a compact trie $\T$ at level $k-1$ of the errata tree and index $i \in [1 \dd b_C]$, returns $\llabel_C[i]$ in $\Oh(1)$ time.
\end{lemma}

\subsubsection{Compact \texorpdfstring{$k$}{k}-errata tree}
The main differences between the proof of the theorem below and of \cref{thm:kwildP} are that: (1) here we use compact representations of both types of labels, $\llabel_C$ for heavy paths $C$ and $\llabel'_v$ for nodes $v$, and (2) for reporting 1-mismatch occurrences using labels $\llabel'_v$, we use an additional data structure to skip over potential exact occurrences to avoid reporting such occurrences multiple times. The data structure consists of bit arrays $B$ together with queries: $\mathit{rank}_q(x)=|\{i \in [1 \dd x]\,:\,B[i]=q\}|$, $\mathit{select}_q(i)=\min\{x \in [1 \dd |B|]\,:\,\mathit{rank}_q(x)=i\}$, for $q \in \{0,1\}$. It is known \cite{DBLP:conf/focs/Jacobson89} that such queries can be answered in $\Oh(1)$ time using a data structure of $o(|B|)$ bits.

\thmone*
\begin{proof}
\emph{Data structure.}
We store the $(k-1)$-errata tree explicitly. We also store the data structure of \cref{lem:labelChan} for computing $\llabel_C$. Finally, we store the arrays $\strank'_v$, $\treepointer'_v$, and $\modifiedrank'_{v,u_a}$ with rank queries, as well as sparse suffix arrays for all compact tries at level $k-1$ as in wildcard indexing. The data structures take $\Oh(n \log^{k-1} n)$ space.

For each heavy path $C$ of each compact trie $\T$ at level $k-1$ of the data structure, let $\T'_C$ denote the compact trie of 1-modified suffixes $S'_{C,1},\ldots,S'_{C,b_C}$. We select every $\floor{\log n}$th modified suffix among $S'_{C,1},\ldots,S'_{C,b_C}$, as well as $S'_{C,1}$ and $S'_{C,b_C}$, and construct a compact trie $\T''_C$ of these modified suffixes. Each terminal $S'_{C,i}$ in $\T''_C$ stores $i$ as its number. The total size of the compact tries $\T''_C$ is $\Oh(n \log^{k-1} n)$, whereas trees $\T'_C$ are \emph{not} stored explicitly. We construct our data structure for answering unrooted $\TreeLCP$ queries on each of $\T''_C$, with just $k=1$ modification (\cref{lem:modifiedLCP}).

We also store the (analogous) compact tries $\T''_v$ of sampled terminals $S'_{v,1}$ with unrooted $\TreeLCP$ data structure exactly as in the proof of \cref{thm:kwildP}. A data structure for $\LCP$-queries in $T$ is stored~\cite{DBLP:conf/latin/BenderF00}. Also, a data structure for kangaroo jumps (for $k=1$) of \cref{lem:kangaroo_k} is stored.

\newcommand{\diff}{\mathit{diff}}
Let $v$ be an explicit node in a compact trie at level $k-1$. For $i \in [1 \dd b_v]$, by $c_{v,i}$ we denote the first character of suffix $S_{v,i}$ (i.e., $S_{v,i}=c_{v,i} S'_{v,i}$). We store a bit array $\diff_v$ of size $b_v-1$ such that $\diff_v[i]=1$ if and only if $c_{v,i} \ne c_{v,i+1}$, together with a data structure for $\mathit{rank}_q$ and $\mathit{select}_q$ queries~\cite{DBLP:conf/focs/Jacobson89} on $\diff_v$. The bitmasks and rank/select data structures take $\Oh(n \log^k n)$ bits in total, i.e., $\Oh(n \log^{k-1} n)$ words of space.

\emph{Queries.}
We execute a $(k-1)$-errata tree query for $P$. As a result, in time $\Oh(\log^{k-1} \log \log n)$, $\Oh(\log^{k-1} n)$ pairs $(\T,P')$ are computed where $\T$ is a compact trie at level $k-1$ and $P'$ is a suffix of $P$; our goal is to match $P'$ against terminals in $\T$ with up to 1 mismatch.
We compute $v=\TreeLCP(\T,P')$ (in $\Oh(\log \log n)$ time). If the whole $P'$ was matched, all the terminals in the subtree of $v$ can be reported for the start.

Let us decompose the path from the root of $\T$ to $v$ into $h=\Oh(\log n)$ prefix fragments $C'_1,\ldots,C'_h$ of heavy paths $C_1,\ldots,C_h$, top-down. First, we find 1-mismatch matches of $P'$ that contain the mismatch within one of $C'_1,\ldots,C'_h$, say $C'_p$. In other words, we need to identify all 1-modified suffixes among $S'_{C_p,1},\ldots,S'_{C_p,b_{C_p}}$ that have a certain suffix $P''$ of $P'$ as a prefix. The result will consist of a connected sublist of this list that is computed in $\Oh(\log \log n)$ time using \cref{lem:utility} for $k=1$ and $k'=0$, with the aid of queries to $\llabel_{C_p}$ array (\cref{lem:labelChan}) and kangaroo jumps (\cref{lem:kangaroo_k}).

We find the remaining 1-mismatch matches of $P'$ similarly as in \cref{thm:kwildP}, but using the $\diff$ arrays. If such a match contains the character that immediately follows $C'_p$ on heavy path $C_p$, it is found using an unrooted $\TreeLCP$ query on $\T$. Otherwise, let $v_p$ be the bottommost node of $C'_p$. First let $p<h$. Then the result will correspond to a connected sublist of $S'_{v_p,1},\ldots,S'_{v_p,b_{v_p}}$ but without the elements $i$ such that $c_{v_p,i}$ matches the first character of $P'$ after having matched $C'_p$ (to avoid reporting exact matches of $P'$ multiple times). We find the connected sublist in $\Oh(\log \log n)$ time using \cref{lem:utility} for $k=k'=0$ and a suffix of $P'$, with the aid of queries to $\llabel'_{v_p}$ array. Then we report only the trimmed suffixes for which the preceding character is not equal to the corresponding character of $P'$ using rank/select on $\diff_{v_p}$ array. More precisely, whenever we find a trimmed suffix $S'_{v_p,i}$ that does have the unwanted preceding character $c_{v_p,i}$, we skip to the next one that does not by identifying the nearest index $i' \ge i$ such that $\diff_{v_p}[i']=1$. Thus, at most every second step a new occurrence is reported. If $p=h$, we do not need to use the $\diff_{v_p}$ array as $P'$ was not matched beyond $C'_h$.

After $\Oh(m)$ preprocessing, for each initial pair $(\T,P')$, we consider $\Oh(\log n)$ heavy paths and for each of them perform computations in $\Oh(\log \log n)$ time, plus the output size. Overall, the query complexity is $\Oh(\log^k n \log \log n + m + \occ)$, as desired.
\end{proof}

\section{Proofs of Lemmas from Section \ref{sec:unrooted}}\label{sec:unrooted_proofs}

In the weighted ancestor queries problem over compact trie $\T$, given a node $v$ of $\T$ and an integer $\ell \ge 0$, we are to compute the (explicit or implicit) node $w$ that is an ancestor of $v$ and has string depth $\ell$. For a compact trie $\T$, there is a data structure of $\Oh(|\T|)$ size that can answer weighted ancestor queries on $\T$ in $\Oh(\log \log |\T|)$ time~\cite{DBLP:journals/talg/AmirLLS07}.

\canonical*
\begin{proof}
For each tree $\T_i$, $i \in [1 \dd t]$, we create a compact trie $\T'_i$ of suffixes of $T$ as follows: for every fragment $T[a \dd b)$ that is a terminal in $\T_i$, suffix $T[a \dd n)$ is represented in $\T'_i$. Each suffix represented in $\T'_i$ stores one fragment from $\T_i$ it originates from. Each explicit node of each $\T'_i$ stores an arbitrary terminal from its subtree. Moreover, each explicit node of each $\T'_i$ stores the corresponding node---that is, the node with the same path label---in $\T_i$, if it exists, and \emph{vice versa}. We create the data structure of \cref{thm:unrootedLCP} for answering unrooted $\TreeLCP$ queries in tries $\T'_1,\ldots,\T'_t$. For each of the tries $\T'_i$, we also store a data structure for weighted ancestor queries on $\T'_i$~\cite{DBLP:journals/talg/AmirLLS07}.

Let $v'$ be the node of $\T'_i$ that corresponds to $v$ in $\T_i$. Node $v'$ is computed via the nearest explicit descendant $y$ of $v$, its corresponding node $y'$ in $\T'_i$, and a weighted ancestor query from $y'$. Upon a query $\TreeLCP_v(\T_i,P')$, we ask a query $\TreeLCP_{v'}(\T'_i,P')$. Let $w'$ be the resulting node, $T[a \dd n)$ be the string label of any terminal in the subtree of $w'$, and $T[a \dd b)$ be the corresponding fragment of $T$ that was represented in $\T_i$. We ask a weighted ancestor query for $w'$ at depth $\min(b-a,d)$, where $d$ is the string depth of $w'$, that returns a node $u'$. Then $\TreeLCP_v(\T_i,P')$ returns the node $u$ in $\T_i$ that corresponds to $u'$.

Let us argue for the correctness of the query procedure. Let node $x$ in $\T_i$ be the correct result of $\TreeLCP_v(\T_i,P')$ query. The path label of $u'$ is a prefix of $T[a \dd b)$, so node $u$ indeed exists in $\T_i$. We will show that $x=u$. The path from $v'$ to $w'$ has a label being a prefix of $P'$. Hence, the path from $v'$ to $u'$ (equivalently, the path from $v$ to $u$) has a label $P'[0 \dd z)$ for some index~$z$. This concludes that $x$ is a descendant of $u$. Assume to the contrary that $x$ is a strict descendant of $u$. In this case, $u' \ne w'$ and $u'$ is a terminal with path label $T[a \dd b)$. Then the label of the path from $v$ to $x$ has a prefix $P'[0 \dd z]$. Let $X$ be the path label of $v$. String $X\cdot P'[0 \dd z]$ is a prefix of $T[a \dd n)$, so it equals $T[a \dd b]$. As $x$ in present in $\T_i$, this means that $\T_i$ is not given in a canonical form, as the terminal representing $T[a \dd b)$ has an outgoing edge starting with label $T[b]$; a contradiction.

All data structures use $\Oh(n+\sum_{i=1}^t (|\T_i|+|\T'_i|))=\Oh(n+\sum_{i=1}^t |\T_i|)$ space and any unrooted query $\TreeLCP_v(\T_i,P)$ is answered in $\Oh(\log \log n)$ time (after $\Oh(m)$ preprocessing for pattern $P$, $|P|=m$, that is required in \cref{thm:unrootedLCP}).
\end{proof}

\modified*
\begin{proof}
\emph{Data structure.}
Each trie $\T_i$, for $i \in [1 \dd t]$, is (edge-)decomposed into compact tries of factors of $T$ and compact tries of single-character strings recursively in phases numbered 1,2,\ldots; see \cref{fig:decomp}. Odd and even phases are processed differently. In an odd phase, assume a subtree $\T$ of $\T_i$ is still to be decomposed. For each $(\le k)$-modified suffix of $T$ that is present in $\T$, we select its longest prefix until the next modification. A compact trie of these prefixes, denoted $\T'$, forms a tree in the decomposition. The subtrees of $\T$ obtained upon the removal of $\T'$ become the input of an even phase. In an even phase, for each remaining tree, we form a tree of string depth one by taking the first character of every edge going from the root; the remaining trees are input to the next odd phase. The process ends when there are no further trees to be decomposed. Naturally, each tree in the decomposition stores links to the trees formed at a subsequent phase in the appropriate nodes (that were explicit in $\T$).

\begin{figure}[htpb]
\centering
\renewcommand{\tabcolsep}{0pt}
\newcommand{\dy}{0.6}
\begin{tikzpicture}[scale=1]

\begin{scope}[xshift=-9cm]
\draw (0,0) node[above right] {aa\textcolor{red}{b}aaba};
\draw (0,-0.5) node[above right] {aaba\textcolor{red}{a}cb};
\draw (0,-1) node[above right] {aba\textcolor{red}{a}aa};
\draw (0,-1.5) node[above right] {abaabb};
\draw (0,-2) node[above right] {bca\textcolor{red}{b}};
\draw (0,-2.5) node[above right] {bc\textcolor{red}{b}a};
\draw (0,-3) node[above right] {bc\textcolor{red}{c}aaa};
\draw (0,-3.5) node[above right] {bc\textcolor{red}{c}abc};
\end{scope}

\draw (0,0.5*\dy) -- node[sloped,left,rotate=270] {
\begin{tabular}{c}
a
\end{tabular}
} (-2,-1*\dy);

\draw (-2,-1*\dy) -- node[sloped,left,rotate=270] {
\begin{tabular}{c}
a\\\textcolor{red!70!black}{\,b\,}\\a
\end{tabular}
} (-3,-4*\dy);

\draw (-3,-4*\dy) -- node[sloped,right,rotate=90] {
\begin{tabular}{c}
\textcolor{red!70!black}{\,a\,}
\end{tabular}
} (-3,-5*\dy);

\draw (-3,-5*\dy) -- node[sloped,left,rotate=270] {
\begin{tabular}{c}
b\\a
\end{tabular}
} (-3.5,-7*\dy);

\draw (-3,-5*\dy) -- node[sloped,right,rotate=90] {
\begin{tabular}{c}
c\\b
\end{tabular}
} (-2.5,-7*\dy);

\draw (-2,-1*\dy) -- node[sloped,right,rotate=90] {
\begin{tabular}{c}
b\\a\\\textcolor{red!70!black}{\,a\,}
\end{tabular}
} (-1,-4*\dy);

\draw (-1,-4*\dy) -- node[sloped,left,rotate=270] {
\begin{tabular}{c}
a
\end{tabular}
} (-1.5,-5*\dy);

\draw (-1.5,-5*\dy) -- node[sloped,left,rotate=90] {
\begin{tabular}{c}
a
\end{tabular}
} (-1.5,-6*\dy);

\draw (-1,-4*\dy) -- node[sloped,right,rotate=90] {
\begin{tabular}{c}
b\\b
\end{tabular}
} (-0.5,-6*\dy);

\draw (0,0.5*\dy) -- node[sloped,right,rotate=90] {
\begin{tabular}{c}
b\\c
\end{tabular}
} (1.5,-2*\dy);

\draw (1.5,-2*\dy) -- node[sloped,left,rotate=270] {
\begin{tabular}{c}
a
\end{tabular}
} (0.5,-3*\dy);

\draw (0.5,-3*\dy) -- node[sloped,right,rotate=90] {
\begin{tabular}{c}
\textcolor{red}{b}
\end{tabular}
} (0.5,-4*\dy);

\draw (1.5,-2*\dy) -- node[sloped,right,rotate=270] {
\begin{tabular}{c}
\textcolor{red}{b}
\end{tabular}
} (1.3,-3*\dy);

\draw (1.3,-3*\dy) -- node[sloped,right,rotate=90] {
\begin{tabular}{c}
a
\end{tabular}
} (1.3,-4*\dy);

\draw (1.5,-2*\dy) -- node[sloped,right,rotate=90] {
\begin{tabular}{c}
\textcolor{red}{c}
\end{tabular}
} (2.5,-3*\dy);

\draw (2.5,-3*\dy) -- node[sloped,right,rotate=90] {
\begin{tabular}{c}
a
\end{tabular}
} (2.5,-4*\dy);

\draw (2.5,-4*\dy) -- node[sloped,left,rotate=270] {
\begin{tabular}{c}
a\\a
\end{tabular}
} (2,-6*\dy);

\draw (2.5,-4*\dy) -- node[sloped,right,rotate=90] {
\begin{tabular}{c}
b\\c
\end{tabular}
} (3,-6*\dy);

\draw[ultra thick,blue] (0,0.5*\dy) -- (-2,-1*\dy) -- (-3,-4*\dy)  (-2,-1*\dy) -- (-1,-4*\dy) -- (-0.5,-6*\dy)  (0,0.5*\dy) -- (1.5,-2*\dy) -- (0.5,-3*\dy);
\draw[ultra thick,brown] (-3.5,-7*\dy) -- (-3,-5*\dy) -- (-2.5,-7*\dy)  (-1.5,-5*\dy) -- (-1.5,-6*\dy)  (1.3,-3*\dy) -- (1.3,-4*\dy)  (2.5,-3*\dy) -- (2.5,-4*\dy) -- (2,-6*\dy)  (2.5,-4*\dy) -- (3,-6*\dy);


\end{tikzpicture}
\vspace*{-1cm}
\caption{Left: an ordered list of $(\le 1)$-modified fragments of some hypothetical text $T$; modifications are shown in red. (One of the fragments is actually 0-modified.) Right: the compact trie of these fragments decomposed as in the proof of \cref{lem:modifiedLCP}. The subtrees in bold (blue and brown) are created in odd phases and are compact tries of substrings of $T$; the remaining even-phase subtrees have string depth 1. The characters drawn in half red, half black were modified in one fragment but unmodified in some other one.}\label{fig:decomp}
\end{figure}

Each compact trie created in an odd phase represents a set of fragments of $T$. As described in \cref{sec:unrooted}, it can be represented in a canonical form by extending the terminals. Thus, a data structure of \cref{thm:unrootedLCPfactors} for all trees in the decomposition created in odd phases can be stored. The suffix tree of $T$ is also stored. Each explicit node of a tree $\T'$ of the decomposition stores the corresponding explicit node of $\T_i$, and each explicit node of $\T_i$ stores the corresponding node(s) of trees in the decomposition (there are up to two such trees, one from an odd phase and one from an even phase). We also store the data structure of \cref{thm:WExp} for each tree in the decomposition.

\begin{claim}
If $k\ge1$, the total size of trees in the decomposition is $\Oh(Nk)$ and the number of phases is $\Oh(k)$.
\end{claim}
\begin{proof}
Let $S$ be a $k''$-modified suffix of $T$ being a terminal of $\T_i$. It suffices to show that fragments originating from $S$ are present in at most $2k''$ subsequent trees in the decomposition of $\T_i$. If $k''=0$, the whole $S$ is contained already in the topmost tree. Assume now that $k''>0$.

Let $\T$ be the topmost tree in the decomposition of $\T_i$ and $a$ be the position of the first modification in $S$. By definition, $\T$ contains a prefix of $S$ of length at least $a$ which is trimmed for the next phase. The next tree, created at an even phase, removes one character from the remaining suffix. Effectively, by the next odd phase, a prefix of $S$ containing the first modification is removed, and we are left with a $(k''-1)$-modified suffix of $S$. The conclusion follows by induction.
\end{proof}

\emph{Preprocessing of the pattern.} We perform the preprocessing of \cref{thm:unrootedLCPfactors} for odd-phase trees that takes $\Oh(m)$ time.

\emph{Query.} Let us consider a query $\TreeLCP_v(\T_i,P')$ for $v$ being a node of $\T_i$ and $P'$ being a ($\le k'$)-modified suffix of the pattern. Then $P'$ can be partitioned into $\Oh(k'+1)$ fragments, each of which is a fragment of $P$ or a single character originating from a modification. The fragments in the partition are processed consecutively. At each point, a node $u$ in a tree $\T$ from the decomposition of $\T_i$ is stored; the node is such that the path label of the corresponding node $u'$ in $\T_i$ equals the already processed prefix of $P'$. Initially, $u$ and $\T$ are such that $u$ corresponds to the node $v$ in $\T_i$.

A fragment $F$ in the partition of $P'$ is processed as follows. First, we ask a $\TreeLCP_u(\T,F)$ query and obtain a node $w$ in $\T$ that corresponds to a node $w'$ of $\T_i$. If $\T$ was formed at an even phase, this query is simply processed in $\Oh(1)$ time. If $|F|=1$, we apply \cref{thm:WExp} in $\Oh(1)$ time. Otherwise, the query is answered in $\Oh(\log \log n)$ time by \cref{thm:unrootedLCPfactors}. If the length $\ell$ of the path from $u$ to $w$ is smaller than $|F|$, a prefix of $F$ of length $\ell$ is removed, and the next call will be for the root of the next tree in the decomposition that starts in $w$, if there is any such tree, and otherwise, node $w'$ of $\T_i$ is returned. Finally, if $\ell=|F|$, the next call will be for node $w$ in $\T$ and the next fragment of $P'$.

Each subsequent call removes one of the $\Oh(k'+1)$ fragments of $P'$ or moves to the next tree in the decomposition. By the claim, the query complexity is $\Oh((k+k'+1) \log \log n)$.
\end{proof}

\kangaroo*
\begin{proof}
\emph{Data structure.} We use the suffix tree $\TS$ of $T$ and the data structure~\cite{DBLP:conf/latin/BenderF00} for $\LCP$-queries on $T$. We also use the $\Oh(n)$-space data structure of \cref{lem:rootedLCP_special}.

\emph{Preprocessing of the pattern.} We  compute $\TreeLCP(\TS,U)$ for all suffixes $U$ of pattern $P$ in $\Oh(m)$ total time using \cref{lem:rootedLCP_special}.

\emph{Query.} Let $S$ and $P'$ be as in the statement of the lemma. We partition $S$ into $\Oh(k)$ fragments $X_1,\ldots,X_x$, each being a fragment of $T$ or a single-character modification. Similarly, we partition $P'$ into $\Oh(k')$ fragments $Y_1,\ldots,Y_y$, each being a fragment of $P$ or a single-character modification. Each subsequent fragment $Y_i$ such that $|Y_i|>1$ is partitioned recursively by cutting off its longest prefix being a substring of $T$. That is, if the current $Y_i$ is a prefix of a suffix $U$ of $P$, we cut off the prefix of $Y_i$ of length $\min(|Y_i|,\ell)$ where $\ell$ is the string depth of node $\TreeLCP(\TS,U)$. This process is interrupted if more than $y+2x+1$ fragments in the partition of $P'$ have been formed; in this case, the remaining suffix $V$ of $P'$ is discarded. Thus, the partitioning takes $\Oh(k+k'+1)$ time. Let $Z_1 \cdots Z_z V$ be the final partition of $P'$.

For each $Y_i$ such that $|Y_i|>1$, all fragments formed from it except for the last one are called \emph{special}.

\begin{claim}
Let $X$ be a single-character string or a substring of $T$ and $j \in [1 \dd z)$ be such that $Z_j$ is special. Denote $Z'=Z_j(Z_{j+1}[0])$ and $\ell=\LCP(X,Z')$. Then either $\ell=|X|$ or $\ell < \min(|X|,|Z'|)$.
\end{claim}
\begin{proof}
If $|X|=1$, then $\ell \in \{0,1\}$ and the conclusion holds. Otherwise, $X$ is a substring of $T$. If $X$ is a prefix of $Z'$, then $\ell=|X|$. By definition, string $Z'$ is not a substring of $T$. Hence, $Z'$ is not a prefix of $X$. If $X$ is not a prefix of $Z'$, then the longest common prefix of $X$ and $Z'$ is shorter than the two strings.
\end{proof}

\begin{algorithm}[h!]
    \SetCommentSty{textrm}
    \caption{Compute $\LCP(S,P')$}\label{algo:LCP}
    $S:= X_1 \cdots X_x$\tcp*{$|X_i|=1$ or $X_i$ is a fragment of $T$}
    $P':= Z_1 \cdots Z_zV$\tcp*{$|Z_j|=1$ or $Z_j$ is a fragment of $T$ and $V$ is a discarded suffix}
    $i:= j:= 1$\;
    $\mathit{lcp}:= 0$\;
    \While{$i \le x$ \KwSty{and} $j \le z$}{
        \If{$|X_i|=1$}{
            \lIf{$Z_j[0] \ne X_i$}{\KwSty{break}}
            \lIf{$|Z_j|=1$}{$j:= j+1$}
            \lElse{$Z_j:= Z_j[1 \dd |Z_j|)$}
            $i:= i+1$\;
            $\mathit{lcp}:= \mathit{lcp}+1$\;
        }
        \ElseIf{$|Z_j|=1$}{
            \lIf{$X_i[0] \ne Z_j$}{\KwSty{break}}
            $X_i:= X_i[1 \dd |X_i|)$\;
            $j:= j+1$\;
            $\mathit{lcp}:= \mathit{lcp}+1$\;
        }
        \Else(\tcp*[f]{$X_i$ and $Z_j$ are fragments of $T$}){
            $\ell:= \LCP(X_i,Z_j)$\;
            $\mathit{lcp}:= \mathit{lcp}+\ell$\;
            \lIf{$\ell<\min(|X_i|,|Z_j|)$}{\KwSty{break}}
            \lIf{$\ell=|X_i|$}{$i:= i+1$}
            \lElse{$X_i:= X_i[\ell \dd |X_i|)$}
            \lIf{$\ell=|Z_j|$}{$j:= j+1$}
            \lElse{$Z_j:= Z_j[\ell \dd |Z_j|)$}
        }
    }
    \Return{$\mathit{lcp}$;}
\end{algorithm}

We compute $\LCP(S,P')$ by subsequently removing common prefixes of the two strings, as shown in Algorithm~\ref{algo:LCP}. Let $X_i$ and $Z_j$ be the fragments being the current prefixes of $S$ and $P'$, respectively. If $|X_i|=1$ or $|Z_j|=1$, we check if it agrees with the first character of the other string and either break the computations or cut the first character from both strings. Otherwise, we are left with computing the $\LCP$ of two substrings of $T$, which can be done in $\Oh(1)$ time~\cite{DBLP:conf/latin/BenderF00}.

Each step of the while-loop is performed in $\Oh(1)$ time and removes a fragment from the decomposition of $S$ or $P'$, or breaks the loop. By the claim, whenever a special fragment of $P'$ followed by another fragment is processed, a whole fragment of $S$ is removed or the algorithm breaks. In summary, at most $2x$ special fragments can be processed in the algorithm, which shows correctness of interrupting the partitioning of $P'$. The query algorithm works in $\Oh(k+k'+1)$ time.
\end{proof}

\section{Conclusions}
We presented the first general space improvement—by a factor of $\Theta(\log n)$—for $k$-mismatch indexing over the errata trees of Cole, Gottlieb, and Lewenstein~\cite{DBLP:conf/stoc/ColeGL04}, while retaining the same query time for $k=\Oh(1)$.
A variant of our index optimized for the case of $k,\sigma=\Oh(1)$ achieves an overall space improvement of nearly $\log^2 n$, surpassing the $\Theta(\log n)$-factor improvement of Chan, Lam, Sung, Tam, and Wong~\cite{DBLP:journals/algorithmica/ChanLSTW10} for $\sigma=\Oh(1)$.

An immediate open problem is to design even more space-efficient $k$-mismatch indexes; our intermediate results imply that one can focus on patterns of polylogarithmic length. 
Another natural question is whether the query time of $k$-errata trees can be reduced.

The original $k$-errata tree can be constructed in time proportional to its size~\cite{DBLP:conf/stoc/ColeGL04} for $k=\Oh(1)$. 
We believe our data structures can be built in similar time, possibly with a small polylogarithmic-factor overhead.
However, some of the black-box components we rely on, particularly \cref{thm:Raman} (from \cite{DBLP:journals/talg/RamanRS07}) and \cref{thm:unrootedLCP} (from \cite{DBLP:journals/algorithmica/ChanLSTW10}), lack efficient construction algorithms in the literature. 
Unpacking these components to provide efficient preprocessing is likely possible, though somewhat tedious.

We hope that the techniques developed in this work will prove useful in other applications of $k$-errata trees (see the references in \cref{sec:intro}). In ongoing work, we are exploring extensions to $k$-error text indexing.

\bibliographystyle{alphaurl}
\bibliography{references}

\newcommand{\etalchar}[1]{$^{#1}$}
\begin{thebibliography}{GGM{\etalchar{+}}25}

\bibitem[ABR00]{DBLP:conf/focs/AlstrupBR00}
Stephen Alstrup, Gerth~St{\o}lting Brodal, and Theis Rauhe.
\newblock New data structures for orthogonal range searching.
\newblock In {\em 41st Annual Symposium on Foundations of Computer Science,
  {FOCS} 2000, 12-14 November 2000, Redondo Beach, California, {USA}}, pages
  198--207. {IEEE} Computer Society, 2000.
\newblock \href {https://doi.org/10.1109/SFCS.2000.892088}
  {\path{doi:10.1109/SFCS.2000.892088}}.

\bibitem[AKL{\etalchar{+}}00]{DBLP:journals/jal/AmirKLLLR00}
Amihood Amir, Dmitry Keselman, Gad~M. Landau, Moshe Lewenstein, Noa Lewenstein,
  and Michael Rodeh.
\newblock Text indexing and dictionary matching with one error.
\newblock {\em J. Algorithms}, 37(2):309--325, 2000.
\newblock \href {https://doi.org/10.1006/JAGM.2000.1104}
  {\path{doi:10.1006/JAGM.2000.1104}}.

\bibitem[Al{-}15]{DBLP:journals/jcb/Al-Okaily15}
Anas Al{-}Okaily.
\newblock Error tree: {A} tree structure for {Hamming} and edit distances and
  wildcards matching.
\newblock {\em J. Comput. Biol.}, 22(12):1118--1128, 2015.
\newblock \href {https://doi.org/10.1089/CMB.2015.0132}
  {\path{doi:10.1089/CMB.2015.0132}}.

\bibitem[ALLS07]{DBLP:journals/talg/AmirLLS07}
Amihood Amir, Gad~M. Landau, Moshe Lewenstein, and Dina Sokol.
\newblock Dynamic text and static pattern matching.
\newblock {\em {ACM} Trans. Algorithms}, 3(2):19, 2007.
\newblock \href {https://doi.org/10.1145/1240233.1240242}
  {\path{doi:10.1145/1240233.1240242}}.

\bibitem[ALP04]{DBLP:journals/jal/AmirLP04}
Amihood Amir, Moshe Lewenstein, and Ely Porat.
\newblock Faster algorithms for string matching with k mismatches.
\newblock {\em J. Algorithms}, 50(2):257--275, 2004.
\newblock \href {https://doi.org/10.1016/S0196-6774(03)00097-X}
  {\path{doi:10.1016/S0196-6774(03)00097-X}}.

\bibitem[ALP25]{DBLP:journals/vldb/AyadLP25}
Lorraine A.~K. Ayad, Grigorios Loukides, and Solon~P. Pissis.
\newblock Text indexing for long patterns using locally consistent anchors.
\newblock {\em {VLDB} J.}, 34(5):58, 2025.
\newblock \href {https://doi.org/10.1007/S00778-025-00935-7}
  {\path{doi:10.1007/S00778-025-00935-7}}.

\bibitem[AN16]{DBLP:conf/icalp/AfshaniN16}
Peyman Afshani and Jesper~Sindahl Nielsen.
\newblock Data structure lower bounds for document indexing problems.
\newblock In Ioannis Chatzigiannakis, Michael Mitzenmacher, Yuval Rabani, and
  Davide Sangiorgi, editors, {\em 43rd International Colloquium on Automata,
  Languages, and Programming, {ICALP} 2016, July 11-15, 2016, Rome, Italy},
  volume~55 of {\em LIPIcs}, pages 93:1--93:15. Schloss Dagstuhl -
  Leibniz-Zentrum f{\"{u}}r Informatik, 2016.
\newblock \href {https://doi.org/10.4230/LIPICS.ICALP.2016.93}
  {\path{doi:10.4230/LIPICS.ICALP.2016.93}}.

\bibitem[Bel09]{DBLP:conf/cpm/Belazzougui09}
Djamal Belazzougui.
\newblock Faster and space-optimal edit distance "1" dictionary.
\newblock In Gregory Kucherov and Esko Ukkonen, editors, {\em Combinatorial
  Pattern Matching, 20th Annual Symposium, {CPM} 2009, Lille, France, June
  22-24, 2009, Proceedings}, volume 5577 of {\em Lecture Notes in Computer
  Science}, pages 154--167. Springer, 2009.
\newblock \href {https://doi.org/10.1007/978-3-642-02441-2\_14}
  {\path{doi:10.1007/978-3-642-02441-2\_14}}.

\bibitem[Bel15]{DBLP:journals/algorithmica/Belazzougui15}
Djamal Belazzougui.
\newblock Improved space-time tradeoffs for approximate full-text indexing with
  one edit error.
\newblock {\em Algorithmica}, 72(3):791--817, 2015.
\newblock \href {https://doi.org/10.1007/S00453-014-9873-9}
  {\path{doi:10.1007/S00453-014-9873-9}}.

\bibitem[BF00]{DBLP:conf/latin/BenderF00}
Michael~A. Bender and Martin Farach{-}Colton.
\newblock The {LCA} problem revisited.
\newblock In Gaston~H. Gonnet, Daniel Panario, and Alfredo Viola, editors, {\em
  {LATIN} 2000: Theoretical Informatics, 4th Latin American Symposium, Punta
  del Este, Uruguay, April 10-14, 2000, Proceedings}, volume 1776 of {\em
  Lecture Notes in Computer Science}, pages 88--94. Springer, 2000.
\newblock \href {https://doi.org/10.1007/10719839\_9}
  {\path{doi:10.1007/10719839\_9}}.

\bibitem[BGP{\etalchar{+}}24]{BGPSSZ24}
Giulia Bernardini, Est{\'{e}}ban Gabory, Solon~P. Pissis, Leen Stougie,
  Michelle Sweering, and Wiktor Zuba.
\newblock Elastic-degenerate string matching with 1 error or mismatch.
\newblock {\em Theory Comput. Syst.}, 68(5):1442--1467, 2024.
\newblock \href {https://doi.org/10.1007/S00224-024-10194-8}
  {\path{doi:10.1007/S00224-024-10194-8}}.

\bibitem[BGVV14]{DBLP:journals/mst/BilleGVV14}
Philip Bille, Inge~Li G{\o}rtz, Hjalte~Wedel Vildh{\o}j, and S{\o}ren Vind.
\newblock String indexing for patterns with wildcards.
\newblock {\em Theory Comput. Syst.}, 55(1):41--60, 2014.
\newblock \href {https://doi.org/10.1007/S00224-013-9498-4}
  {\path{doi:10.1007/S00224-013-9498-4}}.

\bibitem[BGW00]{DBLP:conf/esa/BuchsbaumGW00}
Adam~L. Buchsbaum, Michael~T. Goodrich, and Jeffery~R. Westbrook.
\newblock Range searching over tree cross products.
\newblock In Mike Paterson, editor, {\em Algorithms - {ESA} 2000, 8th Annual
  European Symposium, Saarbr{\"{u}}cken, Germany, September 5-8, 2000,
  Proceedings}, volume 1879 of {\em Lecture Notes in Computer Science}, pages
  120--131. Springer, 2000.
\newblock \href {https://doi.org/10.1007/3-540-45253-2\_12}
  {\path{doi:10.1007/3-540-45253-2\_12}}.

\bibitem[BII{\etalchar{+}}17]{DBLP:journals/siamcomp/BannaiIINTT17}
Hideo Bannai, Tomohiro I, Shunsuke Inenaga, Yuto Nakashima, Masayuki Takeda,
  and Kazuya Tsuruta.
\newblock The "runs" theorem.
\newblock {\em {SIAM} J. Comput.}, 46(5):1501--1514, 2017.
\newblock \href {https://doi.org/10.1137/15M1011032}
  {\path{doi:10.1137/15M1011032}}.

\bibitem[Boo80]{DBLP:journals/ipl/Booth80}
Kellogg~S. Booth.
\newblock Lexicographically least circular substrings.
\newblock {\em Inf. Process. Lett.}, 10(4/5):240--242, 1980.
\newblock \href {https://doi.org/10.1016/0020-0190(80)90149-0}
  {\path{doi:10.1016/0020-0190(80)90149-0}}.

\bibitem[CCI{\etalchar{+}}18]{CCIKPRRW18}
Panagiotis Charalampopoulos, Maxime Crochemore, Costas~S. Iliopoulos, Tomasz
  Kociumaka, Solon~P. Pissis, Jakub Radoszewski, Wojciech Rytter, and Tomasz
  Walen.
\newblock Linear-time algorithm for long {LCF} with k mismatches.
\newblock In Gonzalo Navarro, David Sankoff, and Binhai Zhu, editors, {\em
  Annual Symposium on Combinatorial Pattern Matching, {CPM} 2018}, volume 105
  of {\em LIPIcs}, pages 23:1--23:16. Schloss Dagstuhl - Leibniz-Zentrum
  f{\"{u}}r Informatik, 2018.
\newblock \href {https://doi.org/10.4230/LIPICS.CPM.2018.23}
  {\path{doi:10.4230/LIPICS.CPM.2018.23}}.

\bibitem[CEK{\etalchar{+}}21]{CEKNP21}
Anders~Roy Christiansen, Mikko~Berggren Ettienne, Tomasz Kociumaka, Gonzalo
  Navarro, and Nicola Prezza.
\newblock Optimal-time dictionary-compressed indexes.
\newblock {\em {ACM} Trans. Algorithms}, 17(1):8:1--8:39, 2021.
\newblock \href {https://doi.org/10.1145/3426473} {\path{doi:10.1145/3426473}}.

\bibitem[CFP{\etalchar{+}}16]{DBLP:conf/soda/CliffordFPSS16}
Rapha{\"{e}}l Clifford, Allyx Fontaine, Ely Porat, Benjamin Sach, and Tatiana
  Starikovskaya.
\newblock The \emph{k}-mismatch problem revisited.
\newblock In Robert Krauthgamer, editor, {\em Proceedings of the Twenty-Seventh
  Annual {ACM-SIAM} Symposium on Discrete Algorithms, {SODA} 2016, Arlington,
  VA, USA, January 10-12, 2016}, pages 2039--2052. {SIAM}, 2016.
\newblock \href {https://doi.org/10.1137/1.9781611974331.CH142}
  {\path{doi:10.1137/1.9781611974331.CH142}}.

\bibitem[CFS19]{DBLP:conf/soda/Cohen-AddadFS19}
Vincent Cohen{-}Addad, Laurent Feuilloley, and Tatiana Starikovskaya.
\newblock Lower bounds for text indexing with mismatches and differences.
\newblock In Timothy~M. Chan, editor, {\em Proceedings of the Thirtieth Annual
  {ACM-SIAM} Symposium on Discrete Algorithms, {SODA} 2019, San Diego,
  California, USA, January 6-9, 2019}, pages 1146--1164. {SIAM}, 2019.
\newblock \href {https://doi.org/10.1137/1.9781611975482.70}
  {\path{doi:10.1137/1.9781611975482.70}}.

\bibitem[CGL04]{DBLP:conf/stoc/ColeGL04}
Richard Cole, Lee{-}Ad Gottlieb, and Moshe Lewenstein.
\newblock Dictionary matching and indexing with errors and don't cares.
\newblock In L{\'{a}}szl{\'{o}} Babai, editor, {\em Proceedings of the 36th
  Annual {ACM} Symposium on Theory of Computing, Chicago, IL, USA, June 13-16,
  2004}, pages 91--100. {ACM}, 2004.
\newblock \href {https://doi.org/10.1145/1007352.1007374}
  {\path{doi:10.1145/1007352.1007374}}.

\bibitem[CIK{\etalchar{+}}14]{DBLP:journals/tcs/CrochemoreIKRRW14}
Maxime Crochemore, Costas~S. Iliopoulos, Marcin Kubica, Jakub Radoszewski,
  Wojciech Rytter, and Tomasz Walen.
\newblock Extracting powers and periods in a word from its runs structure.
\newblock {\em Theor. Comput. Sci.}, 521:29--41, 2014.
\newblock \href {https://doi.org/10.1016/J.TCS.2013.11.018}
  {\path{doi:10.1016/J.TCS.2013.11.018}}.

\bibitem[CIK{\etalchar{+}}22]{CIKPRS22}
Panagiotis Charalampopoulos, Costas~S. Iliopoulos, Tomasz Kociumaka, Solon~P.
  Pissis, Jakub Radoszewski, and Juliusz Straszynski.
\newblock Efficient computation of sequence mappability.
\newblock {\em Algorithmica}, 84(5):1418--1440, 2022.
\newblock \href {https://doi.org/10.1007/S00453-022-00934-Y}
  {\path{doi:10.1007/S00453-022-00934-Y}}.

\bibitem[CKP{\etalchar{+}}21]{DBLP:journals/jcss/Charalampopoulos21}
Panagiotis Charalampopoulos, Tomasz Kociumaka, Solon~P. Pissis, Jakub
  Radoszewski, Wojciech Rytter, Juliusz Straszynski, Tomasz Walen, and Wiktor
  Zuba.
\newblock Circular pattern matching with \emph{k} mismatches.
\newblock {\em J. Comput. Syst. Sci.}, 115:73--85, 2021.
\newblock \href {https://doi.org/10.1016/J.JCSS.2020.07.003}
  {\path{doi:10.1016/J.JCSS.2020.07.003}}.

\bibitem[CKPR21]{DBLP:conf/esa/Charalampopoulos21}
Panagiotis Charalampopoulos, Tomasz Kociumaka, Solon~P. Pissis, and Jakub
  Radoszewski.
\newblock Faster algorithms for longest common substring.
\newblock In Petra Mutzel, Rasmus Pagh, and Grzegorz Herman, editors, {\em 29th
  Annual European Symposium on Algorithms, {ESA} 2021, September 6-8, 2021,
  Lisbon, Portugal (Virtual Conference)}, volume 204 of {\em LIPIcs}, pages
  30:1--30:17. Schloss Dagstuhl - Leibniz-Zentrum f{\"{u}}r Informatik, 2021.
\newblock \href {https://doi.org/10.4230/LIPICS.ESA.2021.30}
  {\path{doi:10.4230/LIPICS.ESA.2021.30}}.

\bibitem[CLS{\etalchar{+}}10]{DBLP:journals/algorithmica/ChanLSTW10}
Ho{-}Leung Chan, Tak~Wah Lam, Wing{-}Kin Sung, Siu{-}Lung Tam, and Swee{-}Seong
  Wong.
\newblock Compressed indexes for approximate string matching.
\newblock {\em Algorithmica}, 58(2):263--281, 2010.
\newblock \href {https://doi.org/10.1007/S00453-008-9263-2}
  {\path{doi:10.1007/S00453-008-9263-2}}.

\bibitem[CLS{\etalchar{+}}11]{DBLP:journals/jda/ChanLSTW11}
Ho{-}Leung Chan, Tak~Wah Lam, Wing{-}Kin Sung, Siu{-}Lung Tam, and Swee{-}Seong
  Wong.
\newblock A linear size index for approximate pattern matching.
\newblock {\em J. Discrete Algorithms}, 9(4):358--364, 2011.
\newblock \href {https://doi.org/10.1016/J.JDA.2011.04.004}
  {\path{doi:10.1016/J.JDA.2011.04.004}}.

\bibitem[CN21]{DBLP:journals/pvldb/ChenN21}
Yangjun Chen and Hoang~Hai Nguyen.
\newblock On the string matching with k differences in {DNA} databases.
\newblock {\em Proc. {VLDB} Endow.}, 14(6):903--915, 2021.
\newblock \href {https://doi.org/10.14778/3447689.3447695}
  {\path{doi:10.14778/3447689.3447695}}.

\bibitem[CO06]{DBLP:conf/spire/CoelhoO06}
Lu{\'{\i}}s~Pedro Coelho and Arlindo~L. Oliveira.
\newblock Dotted suffix trees. {A} structure for approximate text indexing.
\newblock In Fabio Crestani, Paolo Ferragina, and Mark Sanderson, editors, {\em
  String Processing and Information Retrieval, 13th International Conference,
  {SPIRE} 2006, Glasgow, UK, October 11-13, 2006, Proceedings}, volume 4209 of
  {\em Lecture Notes in Computer Science}, pages 329--336. Springer, 2006.
\newblock \href {https://doi.org/10.1007/11880561\_27}
  {\path{doi:10.1007/11880561\_27}}.

\bibitem[Cob95]{DBLP:conf/cpm/Cobbs95}
Archie~L. Cobbs.
\newblock Fast approximate matching using suffix trees.
\newblock In Zvi Galil and Esko Ukkonen, editors, {\em Combinatorial Pattern
  Matching, 6th Annual Symposium, {CPM} 95, Espoo, Finland, July 5-7, 1995,
  Proceedings}, volume 937 of {\em Lecture Notes in Computer Science}, pages
  41--54. Springer, 1995.
\newblock \href {https://doi.org/10.1007/3-540-60044-2\_33}
  {\path{doi:10.1007/3-540-60044-2\_33}}.

\bibitem[CW18]{Chen2018}
Yangjun Chen and Yujia Wu.
\newblock {BWT}: An index structure to speed-up both exact and inexact string
  matching.
\newblock In Sanjiban~Sekhar Roy, Pijush Samui, Ravinesh Deo, and Stavros
  Ntalampiras, editors, {\em Big Data in Engineering Applications}, pages
  221--264, Singapore, 2018. Springer Singapore.
\newblock \href {https://doi.org/10.1007/978-981-10-8476-8_12}
  {\path{doi:10.1007/978-981-10-8476-8_12}}.

\bibitem[Duv83]{DBLP:journals/jal/Duval83}
Jean{-}Pierre Duval.
\newblock Factorizing words over an ordered alphabet.
\newblock {\em J. Algorithms}, 4(4):363--381, 1983.
\newblock \href {https://doi.org/10.1016/0196-6774(83)90017-2}
  {\path{doi:10.1016/0196-6774(83)90017-2}}.

\bibitem[EF21]{DBLP:conf/icalp/Ellert021}
Jonas Ellert and Johannes Fischer.
\newblock Linear time runs over general ordered alphabets.
\newblock In Nikhil Bansal, Emanuela Merelli, and James Worrell, editors, {\em
  48th International Colloquium on Automata, Languages, and Programming,
  {ICALP} 2021, July 12-16, 2021, Glasgow, Scotland (Virtual Conference)},
  volume 198 of {\em LIPIcs}, pages 63:1--63:16. Schloss Dagstuhl -
  Leibniz-Zentrum f{\"{u}}r Informatik, 2021.
\newblock \href {https://doi.org/10.4230/LIPICS.ICALP.2021.63}
  {\path{doi:10.4230/LIPICS.ICALP.2021.63}}.

\bibitem[EGM{\etalchar{+}}07]{DBLP:journals/tcs/EpifanioGMRS07}
Chiara Epifanio, Alessandra Gabriele, Filippo Mignosi, Antonio Restivo, and
  Marinella Sciortino.
\newblock Languages with mismatches.
\newblock {\em Theor. Comput. Sci.}, 385(1-3):152--166, 2007.
\newblock \href {https://doi.org/10.1016/J.TCS.2007.06.006}
  {\path{doi:10.1016/J.TCS.2007.06.006}}.

\bibitem[Far97]{DBLP:conf/focs/Farach97}
Martin Farach.
\newblock Optimal suffix tree construction with large alphabets.
\newblock In {\em 38th Annual Symposium on Foundations of Computer Science,
  {FOCS} '97, Miami Beach, Florida, USA, October 19-22, 1997}, pages 137--143.
  {IEEE} Computer Society, 1997.
\newblock \href {https://doi.org/10.1109/SFCS.1997.646102}
  {\path{doi:10.1109/SFCS.1997.646102}}.

\bibitem[FG15]{DBLP:conf/cpm/0001G15}
Johannes Fischer and Pawel Gawrychowski.
\newblock Alphabet-dependent string searching with wexponential search trees.
\newblock In Ferdinando Cicalese, Ely Porat, and Ugo Vaccaro, editors, {\em
  Combinatorial Pattern Matching - 26th Annual Symposium, {CPM} 2015, Ischia
  Island, Italy, June 29 - July 1, 2015, Proceedings}, volume 9133 of {\em
  Lecture Notes in Computer Science}, pages 160--171. Springer, 2015.
\newblock \href {https://doi.org/10.1007/978-3-319-19929-0\_14}
  {\path{doi:10.1007/978-3-319-19929-0\_14}}.

\bibitem[FKS84]{DBLP:journals/jacm/FredmanKS84}
Michael~L. Fredman, J{\'{a}}nos Koml{\'{o}}s, and Endre Szemer{\'{e}}di.
\newblock Storing a sparse table with {O(1)} worst case access time.
\newblock {\em J. {ACM}}, 31(3):538--544, 1984.
\newblock \href {https://doi.org/10.1145/828.1884}
  {\path{doi:10.1145/828.1884}}.

\bibitem[FM05]{FM05}
Paolo Ferragina and Giovanni Manzini.
\newblock Indexing compressed text.
\newblock {\em J. {ACM}}, 52(4):552--581, 2005.
\newblock \href {https://doi.org/10.1145/1082036.1082039}
  {\path{doi:10.1145/1082036.1082039}}.

\bibitem[GG87]{DBLP:journals/tcs/GalilG87}
Zvi Galil and Raffaele Giancarlo.
\newblock Parallel string matching with k mismatches.
\newblock {\em Theor. Comput. Sci.}, 51:341--348, 1987.
\newblock \href {https://doi.org/10.1016/0304-3975(87)90042-9}
  {\path{doi:10.1016/0304-3975(87)90042-9}}.

\bibitem[GGM{\etalchar{+}}25]{GGMPP25}
Pawel Gawrychowski, Adam G{\'{o}}rkiewicz, Pola Marciniak, Solon~P. Pissis, and
  Karol Pokorski.
\newblock Faster approximate elastic-degenerate string matching - part {B}.
\newblock In Paola Bonizzoni and Veli M{\"{a}}kinen, editors, {\em 36th Annual
  Symposium on Combinatorial Pattern Matching, {CPM} 2025}, volume 331 of {\em
  LIPIcs}, pages 29:1--29:21. Schloss Dagstuhl - Leibniz-Zentrum f{\"{u}}r
  Informatik, 2025.
\newblock \href {https://doi.org/10.4230/LIPICS.CPM.2025.29}
  {\path{doi:10.4230/LIPICS.CPM.2025.29}}.

\bibitem[GHPT25]{GHPT25}
Daniel Gibney, Jackson Huffstutler, Mano~Prakash Parthasarathi, and Sharma~V.
  Thankachan.
\newblock Repetition aware text indexing for matching patterns with wildcards.
\newblock In Keren Censor{-}Hillel, Fabrizio Grandoni, Jo{\"{e}}l Ouaknine, and
  Gabriele Puppis, editors, {\em 52nd International Colloquium on Automata,
  Languages, and Programming, {ICALP} 2025}, volume 334 of {\em LIPIcs}, pages
  88:1--88:20. Schloss Dagstuhl - Leibniz-Zentrum f{\"{u}}r Informatik, 2025.
\newblock \href {https://doi.org/10.4230/LIPICS.ICALP.2025.88}
  {\path{doi:10.4230/LIPICS.ICALP.2025.88}}.

\bibitem[GLS18]{GLS18}
Pawel Gawrychowski, Gad~M. Landau, and Tatiana Starikovskaya.
\newblock Fast entropy-bounded string dictionary look-up with mismatches.
\newblock In Igor Potapov, Paul~G. Spirakis, and James Worrell, editors, {\em
  43rd International Symposium on Mathematical Foundations of Computer Science,
  {MFCS} 2018}, volume 117 of {\em LIPIcs}, pages 66:1--66:15. Schloss Dagstuhl
  - Leibniz-Zentrum f{\"{u}}r Informatik, 2018.
\newblock \href {https://doi.org/10.4230/LIPICS.MFCS.2018.66}
  {\path{doi:10.4230/LIPICS.MFCS.2018.66}}.

\bibitem[GMRS03]{DBLP:conf/ciac/GabrieleMRS03}
Alessandra Gabriele, Filippo Mignosi, Antonio Restivo, and Marinella Sciortino.
\newblock Indexing structures for approximate string matching.
\newblock In Rossella Petreschi, Giuseppe Persiano, and Riccardo Silvestri,
  editors, {\em Algorithms and Complexity, 5th Italian Conference, {CIAC} 2003,
  Rome, Italy, May 28-30, 2003, Proceedings}, volume 2653 of {\em Lecture Notes
  in Computer Science}, pages 140--151. Springer, 2003.
\newblock \href {https://doi.org/10.1007/3-540-44849-7\_20}
  {\path{doi:10.1007/3-540-44849-7\_20}}.

\bibitem[GNP20]{DBLP:journals/jacm/GagieNP20}
Travis Gagie, Gonzalo Navarro, and Nicola Prezza.
\newblock Fully functional suffix trees and optimal text searching in
  {BWT}-runs bounded space.
\newblock {\em J. {ACM}}, 67(1):2:1--2:54, 2020.
\newblock \href {https://doi.org/10.1145/3375890} {\path{doi:10.1145/3375890}}.

\bibitem[GU18]{DBLP:conf/icalp/GawrychowskiU18}
Pawel Gawrychowski and Przemyslaw Uznanski.
\newblock Towards unified approximate pattern matching for {Hamming} and
  {L{\_}1} distance.
\newblock In Ioannis Chatzigiannakis, Christos Kaklamanis, D{\'{a}}niel Marx,
  and Donald Sannella, editors, {\em 45th International Colloquium on Automata,
  Languages, and Programming, {ICALP} 2018, July 9-13, 2018, Prague, Czech
  Republic}, volume 107 of {\em LIPIcs}, pages 62:1--62:13. Schloss Dagstuhl -
  Leibniz-Zentrum f{\"{u}}r Informatik, 2018.
\newblock \href {https://doi.org/10.4230/LIPICS.ICALP.2018.62}
  {\path{doi:10.4230/LIPICS.ICALP.2018.62}}.

\bibitem[GV05]{DBLP:journals/siamcomp/GrossiV05}
Roberto Grossi and Jeffrey~Scott Vitter.
\newblock Compressed suffix arrays and suffix trees with applications to text
  indexing and string matching.
\newblock {\em {SIAM} J. Comput.}, 35(2):378--407, 2005.
\newblock \href {https://doi.org/10.1137/S0097539702402354}
  {\path{doi:10.1137/S0097539702402354}}.

\bibitem[HHLS04]{DBLP:conf/cpm/HuynhHLS04}
Trinh N.~D. Huynh, Wing{-}Kai Hon, Tak~Wah Lam, and Wing{-}Kin Sung.
\newblock Approximate string matching using compressed suffix arrays.
\newblock In S{\"{u}}leyman~Cenk Sahinalp, S.~Muthukrishnan, and Ugur
  Dogrus{\"{o}}z, editors, {\em Combinatorial Pattern Matching, 15th Annual
  Symposium, {CPM} 2004, Istanbul,Turkey, July 5-7, 2004, Proceedings}, volume
  3109 of {\em Lecture Notes in Computer Science}, pages 434--444. Springer,
  2004.
\newblock \href {https://doi.org/10.1007/978-3-540-27801-6\_33}
  {\path{doi:10.1007/978-3-540-27801-6\_33}}.

\bibitem[HKS{\etalchar{+}}13]{DBLP:journals/jda/HonKSTV13}
Wing{-}Kai Hon, Tsung{-}Han Ku, Rahul Shah, Sharma~V. Thankachan, and
  Jeffrey~Scott Vitter.
\newblock Compressed text indexing with wildcards.
\newblock {\em J. Discrete Algorithms}, 19:23--29, 2013.
\newblock \href {https://doi.org/10.1016/J.JDA.2012.12.003}
  {\path{doi:10.1016/J.JDA.2012.12.003}}.

\bibitem[HLS{\etalchar{+}}11]{DBLP:journals/tcs/HonLSTV11}
Wing{-}Kai Hon, Tak~Wah Lam, Rahul Shah, Siu{-}Lung Tam, and Jeffrey~Scott
  Vitter.
\newblock Cache-oblivious index for approximate string matching.
\newblock {\em Theor. Comput. Sci.}, 412(29):3579--3588, 2011.
\newblock \href {https://doi.org/10.1016/J.TCS.2011.03.004}
  {\path{doi:10.1016/J.TCS.2011.03.004}}.

\bibitem[Jac89]{DBLP:conf/focs/Jacobson89}
Guy Jacobson.
\newblock Space-efficient static trees and graphs.
\newblock In {\em 30th Annual Symposium on Foundations of Computer Science,
  Research Triangle Park, North Carolina, USA, 30 October - 1 November 1989},
  pages 549--554. {IEEE} Computer Society, 1989.
\newblock \href {https://doi.org/10.1109/SFCS.1989.63533}
  {\path{doi:10.1109/SFCS.1989.63533}}.

\bibitem[KJP77]{DBLP:journals/siamcomp/KnuthMP77}
Donald~E. Knuth, James H.~Morris Jr., and Vaughan~R. Pratt.
\newblock Fast pattern matching in strings.
\newblock {\em {SIAM} J. Comput.}, 6(2):323--350, 1977.
\newblock \href {https://doi.org/10.1137/0206024} {\path{doi:10.1137/0206024}}.

\bibitem[KK99]{DBLP:conf/focs/KolpakovK99}
Roman~M. Kolpakov and Gregory Kucherov.
\newblock Finding maximal repetitions in a word in linear time.
\newblock In {\em 40th Annual Symposium on Foundations of Computer Science,
  {FOCS} '99, 17-18 October, 1999, New York, NY, {USA}}, pages 596--604. {IEEE}
  Computer Society, 1999.
\newblock \href {https://doi.org/10.1109/SFFCS.1999.814634}
  {\path{doi:10.1109/SFFCS.1999.814634}}.

\bibitem[KK19]{DBLP:conf/stoc/KempaK19}
Dominik Kempa and Tomasz Kociumaka.
\newblock String synchronizing sets: sublinear-time {BWT} construction and
  optimal {LCE} data structure.
\newblock In Moses Charikar and Edith Cohen, editors, {\em Proceedings of the
  51st Annual {ACM} {SIGACT} Symposium on Theory of Computing, {STOC} 2019,
  Phoenix, AZ, USA, June 23-26, 2019}, pages 756--767. {ACM}, 2019.
\newblock \href {https://doi.org/10.1145/3313276.3316368}
  {\path{doi:10.1145/3313276.3316368}}.

\bibitem[KK21]{DBLP:journals/corr/abs-2106-12725}
Dominik Kempa and Tomasz Kociumaka.
\newblock Breaking the {O(n)}-barrier in the construction of compressed suffix
  arrays.
\newblock {\em CoRR}, abs/2106.12725, 2021.
\newblock \href {http://arxiv.org/abs/2106.12725} {\path{arXiv:2106.12725}}.

\bibitem[KK23]{DBLP:conf/soda/KempaK23}
Dominik Kempa and Tomasz Kociumaka.
\newblock Breaking the {O}(\emph{n})-barrier in the construction of compressed
  suffix arrays and suffix trees.
\newblock In Nikhil Bansal and Viswanath Nagarajan, editors, {\em Proceedings
  of the 2023 {ACM-SIAM} Symposium on Discrete Algorithms, {SODA} 2023,
  Florence, Italy, January 22-25, 2023}, pages 5122--5202. {SIAM}, 2023.
\newblock \href {https://doi.org/10.1137/1.9781611977554.CH187}
  {\path{doi:10.1137/1.9781611977554.CH187}}.

\bibitem[KNO24]{KNO24}
Tomasz Kociumaka, Gonzalo Navarro, and Francisco Olivares.
\newblock Near-optimal search time in {\(\delta\)}-optimal space, and vice
  versa.
\newblock {\em Algorithmica}, 86(4):1031--1056, 2024.
\newblock \href {https://doi.org/10.1007/S00453-023-01186-0}
  {\path{doi:10.1007/S00453-023-01186-0}}.

\bibitem[KSB06]{DBLP:journals/jacm/KarkkainenSB06}
Juha K{\"{a}}rkk{\"{a}}inen, Peter Sanders, and Stefan Burkhardt.
\newblock Linear work suffix array construction.
\newblock {\em J. {ACM}}, 53(6):918--936, 2006.
\newblock \href {https://doi.org/10.1145/1217856.1217858}
  {\path{doi:10.1145/1217856.1217858}}.

\bibitem[KST16]{DBLP:journals/tcs/KucherovST16}
Gregory Kucherov, Kamil Salikhov, and Dekel Tsur.
\newblock Approximate string matching using a bidirectional index.
\newblock {\em Theor. Comput. Sci.}, 638:145--158, 2016.
\newblock \href {https://doi.org/10.1016/J.TCS.2015.10.043}
  {\path{doi:10.1016/J.TCS.2015.10.043}}.

\bibitem[LD09]{DBLP:journals/bioinformatics/LiD09}
Heng Li and Richard Durbin.
\newblock Fast and accurate short read alignment with {Burrows}-{Wheeler}
  transform.
\newblock {\em Bioinform.}, 25(14):1754--1760, 2009.
\newblock \href {https://doi.org/10.1093/BIOINFORMATICS/BTP324}
  {\path{doi:10.1093/BIOINFORMATICS/BTP324}}.

\bibitem[Lew16]{L16}
Moshe Lewenstein.
\newblock Dictionary matching.
\newblock In Ming-Yang Kao, editor, {\em Encyclopedia of Algorithms}, pages
  533--538. Springer New York, 2016.
\newblock \href {https://doi.org/10.1007/978-1-4939-2864-4_109}
  {\path{doi:10.1007/978-1-4939-2864-4_109}}.

\bibitem[LLT{\etalchar{+}}09]{DBLP:conf/bibm/LamLTWWY09}
Tak~Wah Lam, Ruiqiang Li, Alan Tam, Simon C.~K. Wong, Edward Wu, and Siu{-}Ming
  Yiu.
\newblock High throughput short read alignment via bi-directional {BWT}.
\newblock In {\em 2009 {IEEE} International Conference on Bioinformatics and
  Biomedicine, {BIBM} 2009, Washington, DC, USA, November 1-4, 2009,
  Proceedings}, pages 31--36. {IEEE} Computer Society, 2009.
\newblock \href {https://doi.org/10.1109/BIBM.2009.42}
  {\path{doi:10.1109/BIBM.2009.42}}.

\bibitem[LMRT14]{DBLP:journals/tcs/LewensteinMRT14}
Moshe Lewenstein, J.~Ian Munro, Venkatesh Raman, and Sharma~V. Thankachan.
\newblock Less space: Indexing for queries with wildcards.
\newblock {\em Theor. Comput. Sci.}, 557:120--127, 2014.
\newblock \href {https://doi.org/10.1016/J.TCS.2014.09.003}
  {\path{doi:10.1016/J.TCS.2014.09.003}}.

\bibitem[LNV14]{DBLP:conf/stacs/LewensteinNV14}
Moshe Lewenstein, Yakov Nekrich, and Jeffrey~Scott Vitter.
\newblock Space-efficient string indexing for wildcard pattern matching.
\newblock In Ernst~W. Mayr and Natacha Portier, editors, {\em 31st
  International Symposium on Theoretical Aspects of Computer Science {(STACS}
  2014), {STACS} 2014, March 5-8, 2014, Lyon, France}, volume~25 of {\em
  LIPIcs}, pages 506--517. Schloss Dagstuhl - Leibniz-Zentrum f{\"{u}}r
  Informatik, 2014.
\newblock \href {https://doi.org/10.4230/LIPICS.STACS.2014.506}
  {\path{doi:10.4230/LIPICS.STACS.2014.506}}.

\bibitem[Lot97]{DBLP:books/daglib/0019130}
M.~Lothaire.
\newblock {\em Combinatorics on words, Second Edition}.
\newblock Cambridge mathematical library. Cambridge University Press, 1997.

\bibitem[LSTY07]{DBLP:conf/isaac/LamSTY07}
Tak~Wah Lam, Wing{-}Kin Sung, Siu{-}Lung Tam, and Siu{-}Ming Yiu.
\newblock Space efficient indexes for string matching with don't cares.
\newblock In Takeshi Tokuyama, editor, {\em Algorithms and Computation, 18th
  International Symposium, {ISAAC} 2007, Sendai, Japan, December 17-19, 2007,
  Proceedings}, volume 4835 of {\em Lecture Notes in Computer Science}, pages
  846--857. Springer, 2007.
\newblock \href {https://doi.org/10.1007/978-3-540-77120-3\_73}
  {\path{doi:10.1007/978-3-540-77120-3\_73}}.

\bibitem[LSW08]{DBLP:journals/algorithmica/LamSW08}
Tak~Wah Lam, Wing{-}Kin Sung, and Swee{-}Seong Wong.
\newblock Improved approximate string matching using compressed suffix data
  structures.
\newblock {\em Algorithmica}, 51(3):298--314, 2008.
\newblock \href {https://doi.org/10.1007/S00453-007-9104-8}
  {\path{doi:10.1007/S00453-007-9104-8}}.

\bibitem[LV86]{DBLP:journals/tcs/LandauV86}
Gad~M. Landau and Uzi Vishkin.
\newblock Efficient string matching with k mismatches.
\newblock {\em Theor. Comput. Sci.}, 43:239--249, 1986.
\newblock \href {https://doi.org/10.1016/0304-3975(86)90178-7}
  {\path{doi:10.1016/0304-3975(86)90178-7}}.

\bibitem[LV89]{DBLP:journals/jal/LandauV89}
Gad~M. Landau and Uzi Vishkin.
\newblock Fast parallel and serial approximate string matching.
\newblock {\em J. Algorithms}, 10(2):157--169, 1989.
\newblock \href {https://doi.org/10.1016/0196-6774(89)90010-2}
  {\path{doi:10.1016/0196-6774(89)90010-2}}.

\bibitem[Maa06]{DBLP:journals/algorithmica/Maass06}
Moritz~G. Maa{\ss}.
\newblock Average-case analysis of approximate trie search.
\newblock {\em Algorithmica}, 46(3-4):469--491, 2006.
\newblock \href {https://doi.org/10.1007/S00453-006-0126-4}
  {\path{doi:10.1007/S00453-006-0126-4}}.

\bibitem[MM93]{DBLP:journals/siamcomp/ManberM93}
Udi Manber and Eugene~W. Myers.
\newblock Suffix arrays: {A} new method for on-line string searches.
\newblock {\em {SIAM} J. Comput.}, 22(5):935--948, 1993.
\newblock \href {https://doi.org/10.1137/0222058} {\path{doi:10.1137/0222058}}.

\bibitem[MN05]{DBLP:journals/ipl/MaassN05}
Moritz~G. Maa{\ss} and Johannes Nowak.
\newblock A new method for approximate indexing and dictionary lookup with one
  error.
\newblock {\em Inf. Process. Lett.}, 96(5):185--191, 2005.
\newblock \href {https://doi.org/10.1016/J.IPL.2005.08.001}
  {\path{doi:10.1016/J.IPL.2005.08.001}}.

\bibitem[MN07]{DBLP:journals/jda/MaassN07}
Moritz~G. Maa{\ss} and Johannes Nowak.
\newblock Text indexing with errors.
\newblock {\em J. Discrete Algorithms}, 5(4):662--681, 2007.
\newblock \href {https://doi.org/10.1016/J.JDA.2006.11.001}
  {\path{doi:10.1016/J.JDA.2006.11.001}}.

\bibitem[Mun96]{DBLP:conf/fsttcs/Munro96}
J.~Ian Munro.
\newblock Tables.
\newblock In Vijay Chandru and V.~Vinay, editors, {\em Foundations of Software
  Technology and Theoretical Computer Science, 16th Conference, Hyderabad,
  India, December 18-20, 1996, Proceedings}, volume 1180 of {\em Lecture Notes
  in Computer Science}, pages 37--42. Springer, 1996.
\newblock \href {https://doi.org/10.1007/3-540-62034-6\_35}
  {\path{doi:10.1007/3-540-62034-6\_35}}.

\bibitem[Nav01]{DBLP:journals/csur/Navarro01}
Gonzalo Navarro.
\newblock A guided tour to approximate string matching.
\newblock {\em {ACM} Comput. Surv.}, 33(1):31--88, 2001.
\newblock \href {https://doi.org/10.1145/375360.375365}
  {\path{doi:10.1145/375360.375365}}.

\bibitem[Nav22]{N22}
Gonzalo Navarro.
\newblock Indexing highly repetitive string collections, part {II:} compressed
  indexes.
\newblock {\em {ACM} Comput. Surv.}, 54(2):26:1--26:32, 2022.
\newblock \href {https://doi.org/10.1145/3432999} {\path{doi:10.1145/3432999}}.

\bibitem[RI07]{DBLP:conf/sofsem/RahmanI07}
M.~Sohel Rahman and Costas~S. Iliopoulos.
\newblock Pattern matching algorithms with don't cares.
\newblock In Jan van Leeuwen, Giuseppe~F. Italiano, Wiebe van~der Hoek,
  Christoph Meinel, Harald Sack, Frantisek Pl{\'{a}}sil, and M{\'{a}}ria
  Bielikov{\'{a}}, editors, {\em {SOFSEM} 2007: Theory and Practice of Computer
  Science, 33rd Conference on Current Trends in Theory and Practice of Computer
  Science, Harrachov, Czech Republic, January 20-26, 2007, Proceedings Volume
  {II}}, pages 116--126. Institute of Computer Science {AS} CR, Prague, 2007.

\bibitem[RRR07]{DBLP:journals/talg/RamanRS07}
Rajeev Raman, Venkatesh Raman, and Srinivasa {Rao Satti}.
\newblock Succinct indexable dictionaries with applications to encoding
  \emph{k}-ary trees, prefix sums and multisets.
\newblock {\em {ACM} Trans. Algorithms}, 3(4):43, 2007.
\newblock \href {https://doi.org/10.1145/1290672.1290680}
  {\path{doi:10.1145/1290672.1290680}}.

\bibitem[TAA16]{TAA16}
Sharma~V. Thankachan, Alberto Apostolico, and Srinivas Aluru.
\newblock A provably efficient algorithm for the \emph{k}-mismatch average
  common substring problem.
\newblock {\em J. Comput. Biol.}, 23(6):472--482, 2016.
\newblock \href {https://doi.org/10.1089/CMB.2015.0235}
  {\path{doi:10.1089/CMB.2015.0235}}.

\bibitem[Tha13]{DBLP:journals/tcs/Thachuk13}
Chris Thachuk.
\newblock Compressed indexes for text with wildcards.
\newblock {\em Theor. Comput. Sci.}, 483:22--35, 2013.
\newblock \href {https://doi.org/10.1016/J.TCS.2012.08.011}
  {\path{doi:10.1016/J.TCS.2012.08.011}}.

\bibitem[Tsu10]{DBLP:journals/jda/Tsur10}
Dekel Tsur.
\newblock Fast index for approximate string matching.
\newblock {\em J. Discrete Algorithms}, 8(4):339--345, 2010.
\newblock \href {https://doi.org/10.1016/J.JDA.2010.08.002}
  {\path{doi:10.1016/J.JDA.2010.08.002}}.

\bibitem[TWLY09]{DBLP:conf/spire/TamWLY09}
Alan Tam, Edward Wu, Tak~Wah Lam, and Siu{-}Ming Yiu.
\newblock Succinct text indexing with wildcards.
\newblock In Jussi Karlgren, Jorma Tarhio, and Heikki Hyyr{\"{o}}, editors,
  {\em String Processing and Information Retrieval, 16th International
  Symposium, {SPIRE} 2009, Saariselk{\"{a}}, Finland, August 25-27, 2009,
  Proceedings}, volume 5721 of {\em Lecture Notes in Computer Science}, pages
  39--50. Springer, 2009.
\newblock \href {https://doi.org/10.1007/978-3-642-03784-9\_5}
  {\path{doi:10.1007/978-3-642-03784-9\_5}}.

\bibitem[Ukk93]{DBLP:conf/cpm/Ukkonen93}
Esko Ukkonen.
\newblock Approximate string-matching over suffix trees.
\newblock In Alberto Apostolico, Maxime Crochemore, Zvi Galil, and Udi Manber,
  editors, {\em Combinatorial Pattern Matching, 4th Annual Symposium, {CPM} 93,
  Padova, Italy, June 2-4, 1993, Proceedings}, volume 684 of {\em Lecture Notes
  in Computer Science}, pages 228--242. Springer, 1993.
\newblock \href {https://doi.org/10.1007/BFB0029808}
  {\path{doi:10.1007/BFB0029808}}.

\bibitem[Wei73]{DBLP:conf/focs/Weiner73}
Peter Weiner.
\newblock Linear pattern matching algorithms.
\newblock In {\em 14th Annual Symposium on Switching and Automata Theory, Iowa
  City, Iowa, USA, October 15-17, 1973}, pages 1--11. {IEEE} Computer Society,
  1973.
\newblock \href {https://doi.org/10.1109/SWAT.1973.13}
  {\path{doi:10.1109/SWAT.1973.13}}.

\bibitem[Wil83]{DBLP:journals/ipl/Willard83}
Dan~E. Willard.
\newblock Log-logarithmic worst-case range queries are possible in space
  $\theta(n)$.
\newblock {\em Inf. Process. Lett.}, 17(2):81--84, 1983.
\newblock \href {https://doi.org/10.1016/0020-0190(83)90075-3}
  {\path{doi:10.1016/0020-0190(83)90075-3}}.

\bibitem[ZLPT24]{ZLPT24}
Wiktor Zuba, Grigorios Loukides, Solon~P. Pissis, and Sharma~V. Thankachan.
\newblock Approximate suffix-prefix dictionary queries.
\newblock In Rastislav Kr{\'{a}}lovic and Anton{\'{\i}}n Kucera, editors, {\em
  49th International Symposium on Mathematical Foundations of Computer Science,
  {MFCS} 2024}, volume 306 of {\em LIPIcs}, pages 85:1--85:18. Schloss Dagstuhl
  - Leibniz-Zentrum f{\"{u}}r Informatik, 2024.
\newblock \href {https://doi.org/10.4230/LIPICS.MFCS.2024.85}
  {\path{doi:10.4230/LIPICS.MFCS.2024.85}}.

\end{thebibliography}

\end{document}